%% file: first_law.tex
\documentclass[a4paper,amsmath,onecolumn,floatfix,nofootinbib,dvipsnames,accepted=2020-04-17,10pt]{quantumarticle}
\pdfoutput=1
\input{./quantumdefinitions.tex}
\begin{document}
\title{The first law of general quantum resource theories}
\author{Carlo Sparaciari}\email{c.sparaciari@ucl.ac.uk}
\affiliation{Department of Physics and Astronomy, University College London,
London WC1E 6BT, United Kingdom} 
\author{L\'idia del Rio}
\affiliation{Institute for Theoretical Physics, ETH Zurich, 8093 Z{\"u}rich, Switzerland}
\author{Carlo Maria Scandolo}
\affiliation{Department of Computer Science, University of Oxford, Oxford OX1 3QD, UK}
\author{Philippe Faist}
\affiliation{Dahlem Center for Complex Quantum Systems, Freie Universit\"at Berlin, 14195 Berlin, Germany}
\affiliation{Institute for Quantum Information and Matter, Caltech, Pasadena CA, 91125 USA}
\author{Jonathan Oppenheim}
\affiliation{Department of Physics and Astronomy, University College London,
London WC1E 6BT, United Kingdom}
\date{20-04-2020}
\begin{abstract}
We extend the tools of quantum resource theories to scenarios in which multiple quantities
(or resources) are present, and their interplay governs the evolution of physical systems. We derive
conditions for the interconversion of these resources, which generalise the first law
of thermodynamics. We study reversibility conditions for multi-resource theories, and find that the
relative entropy distances from the invariant sets of the theory play a fundamental role in the quantification
of the resources. The first law for general multi-resource theories is a single relation which links the change
in the properties of the system during a state transformation and the weighted sum of the resources exchanged.
In fact, this law can be seen as relating the change in the relative entropy from different sets of states. In
contrast to typical single-resource theories, the notion of free states and invariant sets of states become distinct
in light of multiple constraints. Additionally, generalisations of the Helmholtz free energy, and of adiabatic and
isothermal transformations, emerge. We thus have a set of laws for general quantum resource theories, which
generalise the laws of thermodynamics. We first test this approach on thermodynamics with multiple
conservation laws, and then apply it to the theory of local operations under energetic restrictions.
\end{abstract}
\maketitle
\tableofcontents
\newpage
\section{Introduction}
\label{intro}
\textbf{Resource theories.} 
Resource theories are a versatile set of tools developed in quantum information theory. They are used
to describe the physical world from the perspective of an agent, whose ability to modify a quantum
system is restricted by either practical or fundamental constraints. These limitations mean that while
some states can still be created under the restricted class of operations (the {\it free} or {\it invariant
set} of states), other state transformations can only be done with the help of additional resources. 
The goal of resource theories is then to quantify this cost, and to consequently assign a price to every
state of the system, from the most expensive to the free ones.  Because of their very general structure,
which only involves the set of states describing a quantum system and a given set of allowed operations
for acting on such system, resource theories can be used to study many different branches of quantum
physics, from entanglement theory~\cite{bennett_mixed-state_1996,rains_entanglement_1997,
vedral_entanglement_1998,rains_bound_1999,horodecki_quantum_2009}
to thermodynamics~\cite{janzing_thermodynamic_2000,horodecki_reversible_2003,
rio_thermodynamic_2011,workvalue,brandao_resource_2013,horodecki_fundamental_2013,
skrzypczyk_work_2014,gallego_thermodynamic_2016}, from asymmetry~\cite{gour_resource_2008,
gour_measuring_2009,marvian_theory_2013} to the theory of magic states~\cite{mari_positive_2012,
veitch_resource_2014,veitch_negative_2012}. Additionally, these theories can often be formulated
within more abstract, axiomatic frameworks~\cite{lieb_physics_1999,lieb_entropy_2013,
weilenmann_axiomatic_2016,fritz_resource_2015,del_rio_resource_2015,coecke_mathematical_2016,
anshu_quantifying_2017}.
\par
Thanks to the underlying common structure present in all the theories described within this framework,
one can find general results which apply to all. For example, a resource theory may be equipped with
a zeroth, second, and even third law, i.e., relations that regulate the different aspects of the theory, which
are reminiscent of the Laws of Thermodynamics. In fact, we have that the \emph{zeroth law} for resource
theories states that there exists equivalence classes of free states, and that states from one of these
classes are the only ones that can be freely added to the system without trivialising the
theory~\cite{brandao_second_2015}. The \emph{second law} of resource theories states that some
quantities, linked to the amount of resource contained in a system, never increase under the action of
the allowed operations~\cite{popescu_thermodynamics_1997}, and for reversible resource theories
satisfying modest assumptions, this quantity is unique~\cite{horodecki_are_2002,brandao_entanglement_2008,
brandao_reversible_2010,horodecki_quantumness_2012,brandao_reversible_2015} --- an example of this is
the free energy, which is a monotone in
thermodynamics as it decreases in any cyclic process, and the local entropy for pure state entanglement
theory. Finally, one might have a generalisation of the third law which places limitations on the time needed
to reach a state when starting from another one, rather than simply telling us whether such transformation
is possible or not~\cite{masanes_general_2017}. With the present work, we aim to derive the
\emph{first law} for resource theories, and to do so we will have to extend the framework so as to include
multiple resources. The law we derive connects the amount of different resources exchanged during a
state transformation to the change, quantified by a specific monotone, between the initial and final state of
the system. When considering thermodynamics, this law connects the amount of work and heat exchanged
during a process to the internal energy of the systems.
\par
\textbf{Multiple resources.}
It is often the case that many resources are needed to perform a given task. For instance, thermodynamics
can be understood as a resource theory with multiple resources~\cite{sparaciari_resource_2016, bera_thermodynamics_2017},
where in order to transform the state of the system we need both \emph{energy} and \emph{information},
or equivalently, work and heat. As another example, some quantum computational schemes consider
the idealized case in which the input qubits are pure, and the gates acting on them create coherence. In order
to better understand the role played in quantum computation by these two resources, \emph{coherence} and
\emph{purity}, a possible approach might consist in combining the resource theories of purity~\cite{horodecki_reversible_2003,
gour_resource_2015} and coherence~\cite{aberg_quantifying_2006,baumgratz_quantifying_2014,winter_operational_2016}
together. Other examples of theories in which multiple resources are
considered can be found in the literature~\cite{slepian_noiseless_1973,horodecki_partial_2005,
ahmadi_wignerarakiyanase_2013,singh_maximally_2015,streltsov_entanglement_2016,chitambar_relating_2016,
sparaciari_resource_2016,erker_autonomous_2017,bera_thermodynamics_2017}.
Given the success of resource theories to describe physical situations where only one resource is
involved, it seems natural to ask the question whether the framework can be extended to the case
in which more resources are involved. For example, it is known that the resource theoretic approach
to thermodynamics allows us to derive a second law relation even in the case in which many (commuting,
non-commuting) conserved quantities are present~\cite{guryanova_thermodynamics_2016,
yunger_halpern_microcanonical_2016,yunger_halpern_beyond_2016,lostaglio_thermodynamic_2017,
halpern_beyond_2018}, and one can consider trade-offs of these \cite{popescu_quantum_2018}. We
are thus interested in understanding if one can extend these results to other resource theories, and
whether a first law of general resource theories exists.
\begin{figure}[t!]
\center
\includegraphics[width=0.4\textwidth]{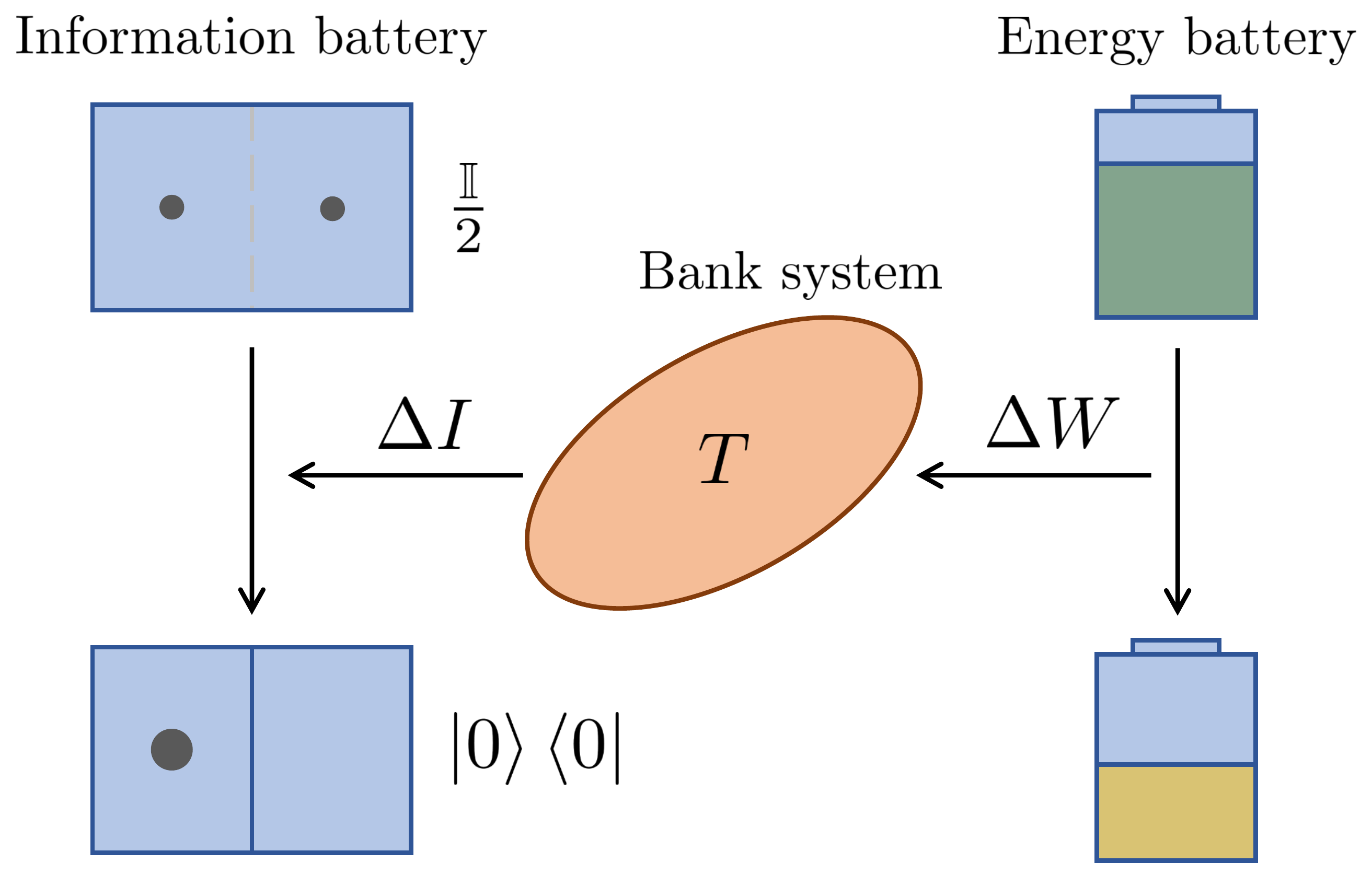}
\hspace{0.1\textwidth}
\includegraphics[width=0.4\textwidth]{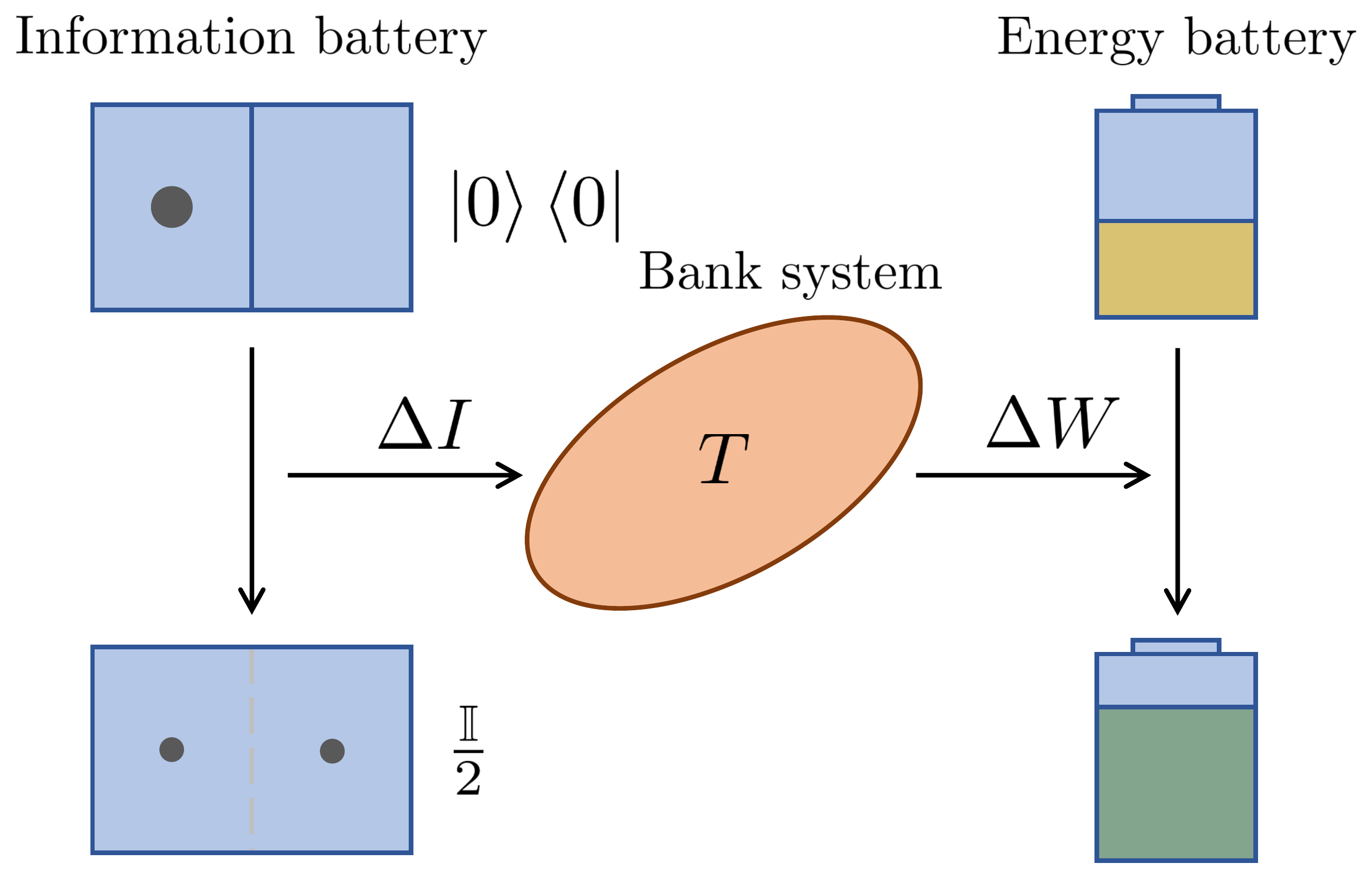}
\caption{An example of a multi-resource theory is thermodynamics, where energy and information
are resources which can be inter-converted. In the figure, we represent three
different systems. The main system is a Szil{\'a}rd box, i.e., a box which can be divided in two
partitions, here containing a single particle of gas. We can either know in which side of the
box the particle is, or we might not have this information (if, for instance, the partition is removed
and the particle is free to move between sides). We additionally have a thermal reservoir surrounding
the box, with a well-defined temperature $T$, and we have an ancillary system that we use
to store energy (or work), which we refer to as the battery. We can then consider the following
two processes. {\bf Left.} Landauer's erasure is the process of converting some of the energy
contained in the battery, $\Delta W$, into information, $\Delta I$, which is then used to reset the state
of the particle in the box (from completely unknown, $\frac{\Id}{2}$, to perfectly known, $\ket{0}$,
in this case). The conversion is realised using the thermal environment, and energy and information
are exchanged at the rate $k_B T$, which only depends on the properties (the temperature)
of the reservoir. {\bf Right.} In the other direction we can convert information,
$\Delta I$, into work, $\Delta W$, at the same exchange rate. Information is extracted form
the box, and converted using the thermal bath into energy, which is then stored in the battery.
Here, we generalise the function of the thermal reservoir to other multi-resource theories, and we
name this system the {\it bank}, since it allows for the exchange of one resource into another.}
\label{fig:slizard_boxes}
\end{figure}
\par
\textbf{Contribution of this work.}
In this paper we present a framework for resource theories with multiple resources, introduced in
Sec.~\ref{multi_res_framework}. In our framework we first consider the different constraints
and conservation laws that the model needs to satisfy, and for each of these constraints, we introduce
the corresponding single-resource theory. Then, we define the class of allowed operations of the
multi-resource theory as the set of maps lying in the intersection of all the classes of allowed
operations of the single-resource theories. Due to this construction, we find that a multi-resource
theory with $m$ resources has at least $m$ invariant sets (i.e., sets of states that are mapped
into themselves by the action of the allowed operations of the theory), each of them corresponding
to the set of free states of one of the $m$ single-resource theories. In order to make
the paper self-contained, in Sec.~\ref{multi_res_framework} we also provide a brief review of the
resource theoretic formalism (see Ref.~\cite{horodecki_quantumness_2012,chitambar_quantum_2018}
for reviews on this topic).
\par
We then study, in Sec.~\ref{rev_theory_mult}, the properties of general multi-resource theories
in the asymptotic limit, that is, when the agent is allowed to act globally over many identical copies of
the system. This limit is of fundamental importance in resource theories since it allows us to investigate
reversibility and the emergence of unique measures for quantifying different resources~\cite{horodecki_quantumness_2012}.
In a reversible theory, we have that the resources consumed to perform a given state transformation
can always be completely recovered with the reverse transformation, so that no resource is ever lost.
In single-resource theories, we can rephrase this notion of reversibility in terms of rates of conversion,
but for general multi-resource theories this is not always possible. As a result, we focus our study on
multi-resource theories that satisfy an additional property, which we refer to as the \emph{asymptotic equivalence
property}~\cite{fritz_resource_2015,sparaciari_resource_2016}, see Def.~\ref{def:asympt_equivalence_multi}
below. We show that, when a multi-resource theory satisfies the asymptotic equivalence property, there is
a unique measure associated with each resource present in the theory. Furthermore, when the invariant sets
of the theory satisfy some natural properties, we find that the unique measures are given by the (regularised)
relative entropy distances from these sets, each of those associated with a different resource.
Finally, we show that when a resource theory satisfies asymptotic equivalence,
it is also reversible in the sense that resources are never lost during a state transformation, and they
can be recovered. This result can be seen as the extension of what has already been shown for reversible
single-resource theories~\cite{popescu_thermodynamics_1997,horodecki_are_2002,horodecki_quantumness_2012,
brandao_reversible_2015}.
\par
In Sec.~\ref{interconv} we address the question of whether it is possible to exchange resources. We
consider the case in which different resources are individually stored in separate systems, which we
call \emph{batteries}. Then, we investigate under which conditions it is possible to find an additional
system, which we refer to as a \emph{bank}\footnote{We apologise in advance for introducing this
terminology into the field of resource theories, but the banks considered here exchange resources
without charging interest or fees, and are thus more akin to community cooperative banks than their
more exploitative cousins.}, that allows us to reduce the amount of resource contained
in one battery while simultaneously increasing the amount of resource in another battery. During such
conversion, we ask the bank not to change its properties -- with respect to a specific measure defined
in Eq.~\eqref{f3_monotone} -- so as to be able to use this system again. For example, in thermodynamics
the thermal bath plays the role of the bank, as it allows us to exchange energy for information and vice versa,
see Fig.~\ref{fig:slizard_boxes}. In order to study interconversion, we demand the invariant sets of the
theory to satisfy an "additivity" condition, which is satisfied by some resource theories, for example by
thermodynamics and purity theory. We find that a multi-resource theory needs to have an empty set of free
states for a bank to exist, and when this condition is satisfied we derive an interconversion relation, see
Thm.~\ref{thm:interconvert_relation}, which defines the rates at which resources are exchanged.
\par
We additionally show that, when the agent is allowed to use batteries and bank, they can perform any
state transformation using variable amounts of resources. Indeed, since the agent can use the bank to
inter-convert between resources, they can decide to invest a higher amount of one resource to save on
the others. This freedom is reflected in our framework by a single relation, the first law of resource theories,
which connects the different resources, each of them weighted by the corresponding exchange rate, to the
change of a particular monotone between the initial and final state of the system, see Cor.~\ref{coro:first_law}.
This equality is a generalisation of the first law of thermodynamics, where the sum of the work performed on
the system and the heat absorbed from the environment is equal to the change in internal energy of the system.
In fact, the first law of thermodynamics can be understood as equating various relative entropy distances which
quantify different types of resources, as we discuss at the beginning of Sec.~\ref{interconv}.
\par
Finally, in Sec.~\ref{examples} we provide two examples of multi-resource
theories which admit an interconversion relation between their resources. The first example concerns
thermodynamics of multiple conserved quantities, for which the interconversion of resources was shown
in Ref.~\cite{guryanova_thermodynamics_2016}. The second example concerns the theory of local
control under energy restrictions. Here we consider a system with a non-local Hamiltonian, and we
assume that the experimentalists acting on this system only have access to a portion of the system.
In this scenario, the entanglement between the different portions of the system and the overall energy
of the global system are the main resources of the theory, and we study under which conditions we
can inter-convert energy and entanglement. For a summary of how to apply our work to an arbitrary
resource theory, see the flowchart in Fig.~\ref{fig:flowchart}.
\section{Framework for multi-resource theories}
\label{multi_res_framework}
Let us now introduce the framework for multi-resource theories. A multi-resource
theory is useful when we need to describe a physical task or process which is subjected to different
constraints and conservation laws. The first step consists in associating each of these constraints with a
single-resource theory, whose class of allowed operations satisfies the specific constraint or conservation
law. The multi-resource theory is then obtained by defining its class of allowed operations as the intersection
between the sets of allowed operations of the different single-resource theories previously defined. In this
way, we are sure of acting on the quantum system with operations that do not violate the multiple constraints
imposed on the task.
\subsection{Single-resource theory}
\label{single_resource}
For simplicity, we restrict ourselves to the study of finite-dimensional quantum systems. Therefore,
the system under investigation is described by a Hilbert space $\hil$ with dimension $d$.
The state-space of this quantum system is given by the set of density operators acting on the Hilbert
space, $\SH = \left\{ \rho \in \BH \ | \ \rho \geq 0, \ \tr{\rho} = 1 \right\}$, where $\BH$ is the set
of bounded operators acting on $\hil$. A single-resource theory for the quantum system under
examination is defined through a class of allowed operations $\A$, that is, a constrained set of
completely positive maps acting on the state-space $\SH$\footnote{Although the operations we
consider are endomorphisms of a given state space, our formalism is still able to describe the general
case in which the agent modifies the quantum system. If the agent's action transforms the state of the
original system, associated with $\hil_1$, into the state of a final system $\hil_2$, we can model this
action with a map acting on the state space of $\hil = \hil_1 \otimes \hil_2$. Suppose the operation maps
$\rho_1 \in \SHi$ into $\sigma_2 \in \SHii$. Then, the map acting on $\SH$ takes the state $\rho_1
\otimes \gamma_2$ and outputs the state $\gamma_1' \otimes \sigma_2$, where $\gamma_1$
and $\gamma_2'$ are free states for the systems described by $\hil_2$ and $\hil_1$, respectively.} \cite{horodecki_are_2002}.
The constraints posed on the set of allowed operations are specific to the resource theory under
consideration. For example, in the theories that study entanglement it is often the case that we
constrain the set of allowed operations to be composed by the maps that are local, and only make
use of classical communication~\cite{bennett_mixed-state_1996}. In asymmetry theory, instead,
we only allow the maps whose action is covariant with respect to the elements of a given
group~\cite{gour_resource_2008}. Furthermore, in the resource theoretic approach to
thermodynamics we can, without loss of generality, constrain this set to those operations, known
as Thermal Operations, which preserve the energy of a closed system, and can thermalise the
system with respect to a background temperature~\cite{janzing_thermodynamic_2000,
brandao_resource_2013,horodecki_fundamental_2013,renes_work_2014}. Once the set of
allowed operations is defined, it is usually possible to identify which states in $\SH$ are resourceful,
and which ones are not. In particular, the set of \emph{free states} for a single-resource theory,
$\f \subset \SH$, is composed of those states that can always be prepared using the allowed
operations, no matter the initial state of the system. Mathematically, this set of states is
defined as
\begin{equation}
\label{free_state_set}
\f = \left\{ \sigma \in \SH \ | \ \forall \, \rho \in \SH, \exists \, \chn \in \A : \chn(\rho) = \sigma \right\}.
\end{equation}
For example, in entanglement theory the free states are the separable states, in asymmetry theory they are
the ones that commute with the elements of the considered group, and in thermodynamics they are the
thermal states at the background temperature.
\par
An \emph{invariant set} is a set of states that is preserved under action of any allowed operation.
From the definition of free states in Eq.~\eqref{free_state_set}, it is easy to show that $\f$ is an
invariant set, and we write this as $\chn(\f) \subseteq \f$ for all $\chn \in \A$. It is worth
noting that while the set of free states is invariant, the opposite clearly does not need to be true. In particular,
when we study multi-resource theory, we will see that several invariant sets can be found, and still there
might be no free set for the theory. Due to the invariant property of free states, we can
also define the class of allowed operations in a different way. Instead of considering the specific
constraints defining the set of allowed operations $\A$, we can simply assume that this set is a
subset of the bigger class of completely positive and trace preserving (CPTP) maps
\begin{equation}
\label{maps_single_copy}
\tilde{\A} = \left\{ \chn : \BH \rightarrow \BH \ | \ \chn \left( \f \right) \subseteq \f \right\},
\end{equation}
that is, the set of maps for which the free states $\f$ form an invariant set. It is worth noting that
$\A$ is often a proper subset of $\tilde{\A}$. For example, in entanglement theory, we have that
$\A$ might be composed by local operations and classical communication (LOCC), which is a
proper subset of the set of all quantum channels which preserve the separable states. Indeed, the
map that swaps between the local states describing the quantum system is clearly not LOCC, but it
preserves separable states~\cite{bennett_quantum_1999}.
\par
We can also extend the single-resource theory to the case in which we consider $n \in \N$ copies of
the quantum system. The class of allowed operations, which in this case we refer to as $\A^{(n)}$, is
still defined by the same constraints, but now acts on $\SHn{n}$, the state-space of $n$ copies of
the system. For example, in the resource theory of thermodynamics with Thermal Operations we
have that the energy of a closed system needs to be exactly conserved. For a single system, this implies
that the operations need to commute with the Hamiltonian $H^{(1)}$. For $n$ non-interacting copies
of the system, instead, the operations commute with the global Hamiltonian $H_n = \sum_{i=1}^n H^{(1)}_i$.
Within the state-space $\SHn{n}$, we can find the set of free states, $\f^{(n)} \subset \SHn{n}$. It
is worth noting that the set of free states for $n$ copies of the system is such that $\f^{\otimes n}
\subseteq \f^{(n)}$, that is, it contains more states than just the tensor product of $n$ states in
$\f$. This is the case, for example, of entanglement theory, where among the free states for two
copies of the system we can find states that are locally entangled, since each agent is allowed to
entangle the partitions of the system they own. On the contrary, the two sets coincide for any $n \in
\N$ for the resource theory of thermodynamics, where the free state is the Gibbs state of a given
Hamiltonian. Anyway, it is still the case that $\f^{(n)}$ is invariant under the class $\A^{(n)}$, and
therefore we can think of the set of allowed operations acting on $n$ copies of the system as a
subset of the bigger set of CPTP maps
\begin{equation}
\label{maps_n_copies}
\tilde{\A}^{(n)} = \left\{ \chn_n : \BHn \rightarrow \BHn \ | \ \chn_n \left( \f^{(n)} \right) \subseteq \f^{(n)} \right\}.
\end{equation}
Thus, in order to extend a single-resource theory to the many-copy case, we need to
take into account the sequence of all sets of allowed operations $\A^{(n)}$, where $n \in \N$ is the number of
copies of the system the maps are acting on.
\par
It is worth noting that the allowed operations we have introduced keep the number of copies
of the system fixed, see Eq.~\eqref{maps_n_copies}. Indeed, we only consider these maps
because, when the number of input and output systems of a quantum channel changes, the
internal structure of the channel involves the discarding (or the addition) of some of these
systems. However, in a (reversible) resource theory, one can perform such operations only if
the amount of resources is kept constant. This is certainly possible if we are to add or trace out
some free states of the theory (which do not contain any resource), but as we will see in the
next section, multi-resource theory not always have any free states. For this reason, we decide
to only focus on maps that conserve the number of systems, even for single-resource theories.
\par
We can now address the problem of quantifying the amount of resource associated with different states of
the quantum system. In resource theories, a resource quantifier is called \emph{monotone}. This
object is a function $f$ from the state-space $\SH$ to the set of real numbers $\R$, which satisfies the
following property,
\begin{equation}
\label{second_law}
f \left( \chn(\rho) \right) \leq f \left( \rho \right) , \qquad \forall \, \rho \in \SH, \ \forall \, \chn \in \A.
\end{equation} 
The above inequality can be interpreted as a ``second law'' for the resource theory, since there is a
quantity (the monotone) that never increases as we act on the system with allowed operations. In the
thermodynamic case, in fact, we know that the Second Law of Thermodynamics imposes that the entropy
of a closed system can never decrease as time goes by. We can extend the definition of monotones to the
case in which we consider $n$ copies of the system. In this case, the function $f$ maps states in $\SHn{n}$
into $\R$, and an analogous relation to the one of Eq.~\eqref{second_law} holds, this time for states in
$\SHn{n}$ and the set of allowed operations $\A^{(n)}$. Finally, we can also define the \emph{regularisation}
of a monotone $f$ as
\begin{equation}
f^{\infty}(\rho) = \lim_{n \rightarrow \infty} \frac{f \left( \rho^{\otimes n} \right)}{n},
\end{equation}
where $\rho \in \SH$, and $\rho^{\otimes n} \in \SHn{n}$. Notice that, given a generic monotone $f$,
we need the above limit to exist and be finite in order to define its regularisation.
\par
For each resource theory there exists several monotones, and we can always build one out of a
\emph{contractive distance}~\cite{brandao_reversible_2015}. Consider the distance
$C \left( \cdot, \cdot \right) : \SH \times \SH \rightarrow \R$ such that
\begin{equation}
C \left( \chn(\rho) , \chn(\sigma) \right)
\leq
C \left( \rho, \sigma \right) , \qquad \forall \, \rho, \sigma \in \SH, \ \forall \, \chn \ \text{CPTP map}.
\end{equation}
Then, a monotone for the single-resource theory with allowed operations $\A$ and free states $\f$ is
\begin{equation}
\label{monotone_contract}
M_{\f} (\rho) = \inf_{\sigma \in \f} C \left( \rho, \sigma \right),
\end{equation}
where it is easy to show that $M_{\f}$ satisfies the property of Eq.~\eqref{second_law}, which follows
from the fact that $\f$ is invariant under the set of allowed operations $\A$, and from the contractivity
of $C \left( \cdot, \cdot \right)$ under any CPTP map. A specific example of a monotone obtained from a
contractive distance is the relative entropy distance from the set $\f$. Consider two states $\rho, \sigma
\in \SH$, such that $\supp{\rho} \subseteq \supp{\sigma}$. Then, we define the relative entropy between
these two states as
\begin{equation} \label{rel_entr}
\re{\rho}{\sigma} = \tr{\rho \, \left( \log \rho - \log \sigma \right)}.
\end{equation}
The relative entropy is contractive under CPTP maps~\cite{lindblad_completely_1975}, and
even if it does not satisfy all the axioms to be a metric\footnote{The relative entropy is non-negative
for any two inputs, and zero only when the two inputs coincide, but it is not symmetric, nor does it
satisfy the triangular inequality.} over $\SH$, we can still obtain a monotone out of this quantity,
building it as in Eq.~\eqref{monotone_contract}. This monotone is
\begin{equation} \label{rel_entr_dist}
E_{\f}( \rho ) = \inf_{\sigma \in \f} \re{\rho}{\sigma},
\end{equation}
and is known as the relative entropy distance from $\f$. When the separable states form the set $\f$,
for example, the monotone is the relative entropy of entanglement~\cite{vedral_entanglement_1998}.
It is worth noting that, in order for $E_{\f}$ to be well-defined, the set $\f$ has to contain at least one
full-rank state.
\subsection{Multi-resource theory}
\label{multi_resource}
Let us consider the case in which we can identify in the theory a number $m > 1$ of resources,
which can arise from some conservation laws, or from some constraints. We now introduce a multi-resource
theory with these $m$ resources. The quantum system under investigation is the same as in the
previous section, described by the states in the state-space $\SH$. For the $i$-th resource of interest, where
$i = 1, \ldots, m$, we consider the corresponding single-resource theory $\rt_i$, defined by the set of allowed
operations $\A_i$ acting on the state-space $\SH$. We denote the set of free states of this single-resource
theory as $\f_i \subset \SH$, and we recall that any allowed operation in $\A_i$ leaves this set invariant.
Therefore, we can consider the class of allowed operation as a subset of the set of CPTP maps
\begin{equation}
\label{maps_single_copy_i}
\tilde{\A}_i = \left\{ \chn_i : \BH \rightarrow \BH
\ | \
\chn_i \left( \f_i \right) \subseteq \f_i \right\}.
\end{equation}
We can also extend the resource theory $\rt_i$ to the case in which we consider more than one copy
of the system, following the same procedure used in the previous section. Then, the class of allowed
operations $\A_i^{(n)}$ acting on $n$ copies of the system is a subset of the set of operations which
leave $\f_i^{(n)} \subset \SHn{n}$ invariant, see Eq.~\eqref{maps_n_copies}.
\par
Once all the single-resource theories $\rt_i$'s are defined, together with their sets of allowed operations,
we can build the multi-resource theory $\rt_{\text{multi}}$ for the quantum system described by the
Hilbert space $\hil$. The set of allowed operations for this theory is given by the maps contained in
the intersection\footnote{While other multi-resource theory constructions can be imagined, the one we
use in this paper provides the certainty that no resource can be created out of free states.} between
the classes of allowed operations of the $m$ single-resource theories, that is
\begin{equation}
\label{all_ops_multi}
\A_{\text{multi}} = \overset{m}{\underset{i=1}{\cap}} \A_i.
\end{equation}
Notice that, alternatively, one can define the set of allowed operations $\A_{\text{multi}}$ as a subset
of the bigger set $\cap_{i=1}^m \tilde{\A}_i$, where $\tilde{\A}_i$ is the set of all the CPTP maps for
which $\f_i$ is invariant, see Eq.~\eqref{maps_single_copy_i}. When $n$ copies of the system are
considered, the class of allowed operations for the multi-resource theory, $\A_{\text{multi}}^{(n)}$, is
obtained by the intersection between the sets of allowed operations $\A_i^{(n)}$ of the different
single-resource theories, that is, $\A_{\text{multi}}^{(n)} = \cap_{i=1}^m \A_i^{(n)}$.
\begin{figure}[t!]
\center
\includegraphics[width=0.3\textwidth]{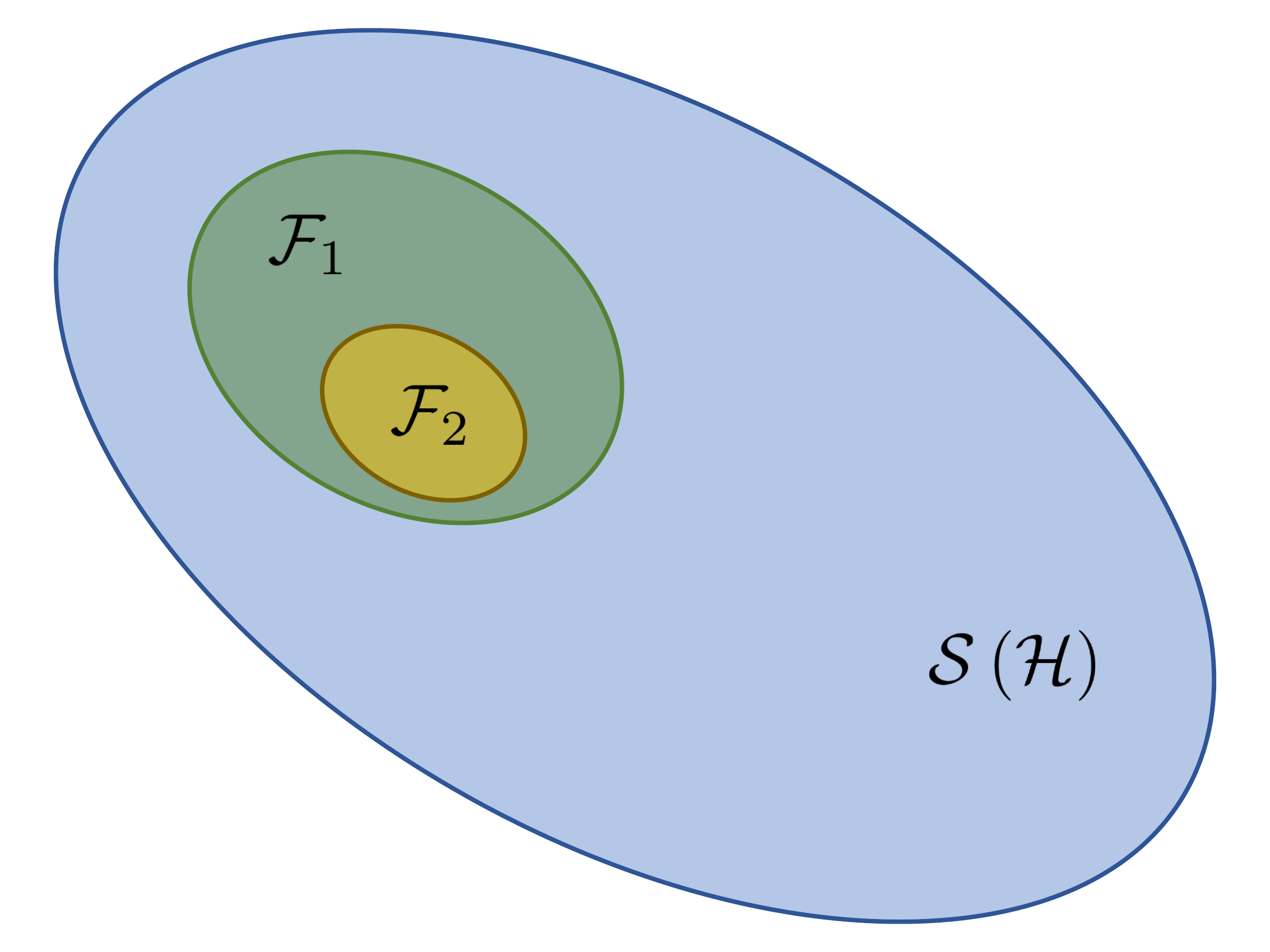}
\includegraphics[width=0.3\textwidth]{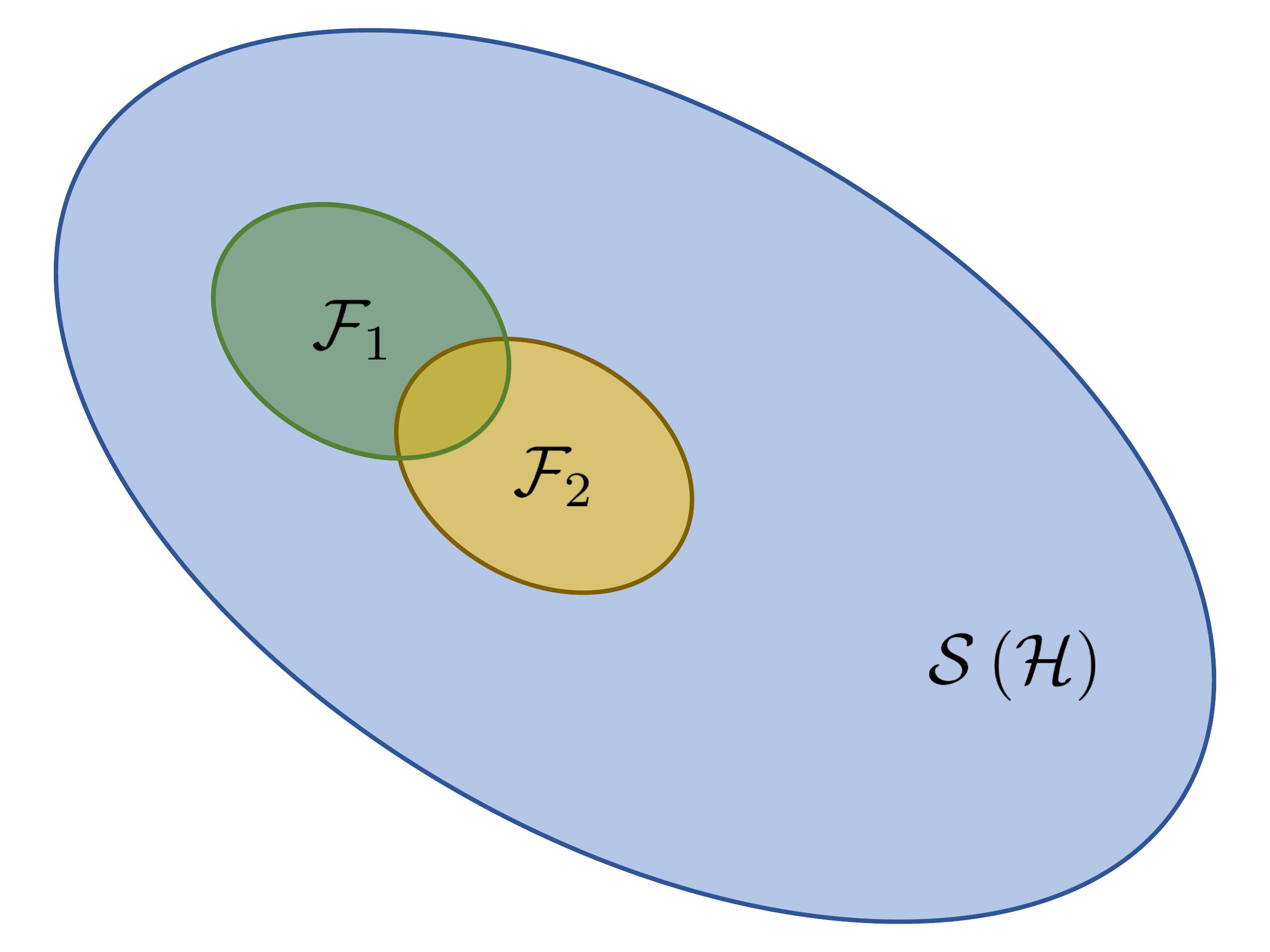}
\includegraphics[width=0.3\textwidth]{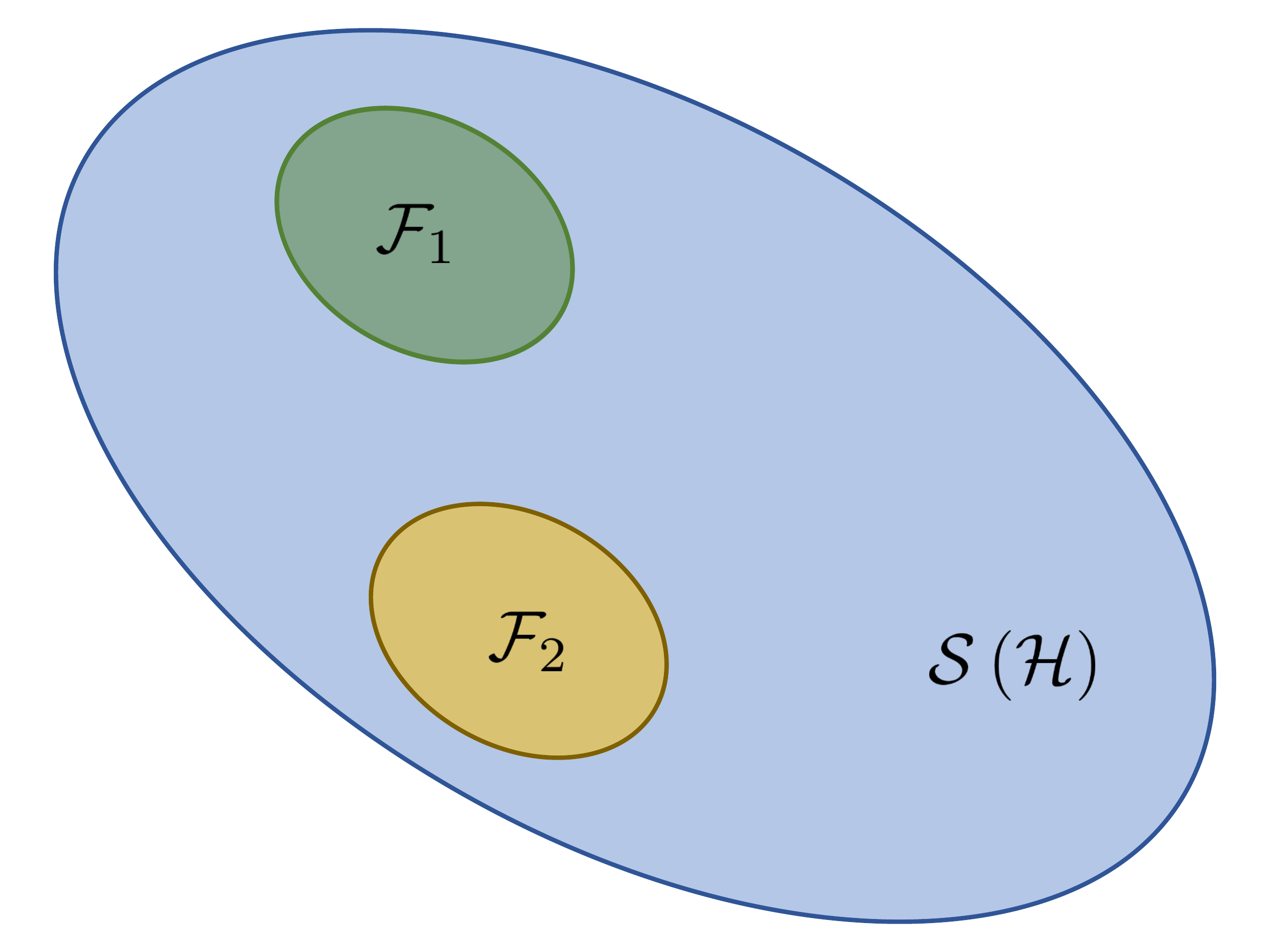}
\caption{The structure of the sets of free states for two single-resource theories
which compose a multi-resource theory. These sets are invariant under the allowed
operations of the resulting multi-resource theory. For theories with $m > 2$ resources,
the structure of the free sets can be obtained by composing the three fundamental scenarios presented here.
{\bf Left.} The invariant set $\f_2$ is a subset of $\f_1$. This multi-resource theory has a set of free
states, which coincides with $\f_2$. An example of such a theory is that of
coherence~\cite{baumgratz_quantifying_2014} and purity~\cite{gour_resource_2015}, where the
invariant sets are incoherent states with respect to a given basis and the maximally-mixed state,
respectively. {\bf Centre.} The two invariant sets intersect each other. This theory has a set of
free states which coincides with the intersection, $\f_1 \cap \f_2$. An example of multi-resource
theory with this structure concerns tripartite entanglement for systems $A$, $B$, and $C$. The
allowed operations of this theory are defined by the intersection of the operations associated
with the theories of bipartite entanglement for systems $AB$ and $C$, systems $AC$ and $B$,
and systems $A$ and $BC$. Notice that this theory does not coincide with the theory of tripartite
LOCC, since some of the free states are entangled~\cite{bennett_unextendible_1999}.
{\bf Right.} The two invariant sets are separated. Consequently, the theory does not
have any free states. In this situation, one can find an interconversion relation between the resources,
as shown in Sec.~\ref{sec:first_law}. An example of a multi-resource theory with this structure is
thermodynamics of closed systems. If the agent does not have perfect control on the reversible
operations they implement, and the closed system is coupled to a sink of energy (an ancillary system
which can only absorb energy), then the allowed operations are given by the intersection between the
set of mixtures of unitary operations, and the set of average-energy-non-increasing maps. In this case, the
maximally-mixed state and the ground state of the Hamiltonian are the two invariant sets of the theory.
Notice that the set of energy-preserving unitary operations, considered in Ref.~\cite{sparaciari_resource_2016},
is a subset of this bigger set.}
\label{fig:invariant_sets_structure}
\end{figure}
\par
We can now consider the invariant sets of this multi-resource theory. Clearly, each set of free states
$\f_i$ associated with the single-resource theory $\rt_i$ is an invariant set for the class of operations
$\A_{\text{multi}}$. However, it is worth noting that the states contained in the $\f_i$'s might not be free
when the multi-resource theory is considered, where a free state is (as we pointed out in the previous
section) a state that does not contain any resource and can be realised using the allowed operations.
Indeed, the states contained in the set $\f_i$ might be resourceful states for the single-resource theory
$\rt_j$, and therefore we would not be able to realise such states with the class of operations
$\A_{\text{multi}}$. In Fig.~\ref{fig:invariant_sets_structure} we show the different configurations for
the invariant sets of a multi-resource theory with two resources. While in the left and central panels the
theory has free states, in the right panel no free states can be found, a noticeable difference from the
framework for single-resource theories.
\par
The multi-resource theory $\rt_{\text{multi}}$ also inherits the monotones of the single-resource theories
that compose it. This follows trivially from the choice we made in defining the class of allowed operations
$\A_{\text{multi}}$, see Eq.~\eqref{all_ops_multi}. Furthermore, other monotones, that are only valid for
the multi-resource theory, can be obtained from the ones inherited from the single-resource theories $\rt_i$'s.
For example, if $f_i$ is a monotone for the single-resource theory $\rt_i$, and $f_j$ is a monotone for the
theory $\rt_j$, their linear combination, where the linear coefficients are positive, is a monotone for the
multi-resource theory $\rt_{\text{multi}}$. Interestingly, in Sec.~\ref{interconv} we will see that a specific
linear combination of the monotones of the different single-resource theories plays an important role in
the interconversion of resources.
\par
Examples of multi-resource theories that can be described within our formalism are already present in the
literature. In Ref.~\cite{streltsov_entanglement_2016}, for instance, the authors study the problem of state-merging
when the parties can only use local operations and classical communication (LOCC), and they restrict
the local operations to be incoherent operations, that is, operations that cannot create coherence (in a
given basis). This theory coincides with the multi-resource theory obtained from combining two single-resource
theories, the one of entanglement, whose set of allowed operations only contains quantum channels
built out of LOCC, and the one of coherence, whose set of allowed operations only contains maps
which do not create coherence. In this case, the structure of the invariant sets is given by the central panel
of Fig.~\ref{fig:invariant_sets_structure}. Another example is the one of Ref.~\cite{sparaciari_resource_2016},
where thermodynamics is obtained as a multi-resource theory whose class of allowed operations is a subset of the
one obtained by taking the intersection of energy-non-increasing maps (operations which do not increase
the average energy of the quantum system, see Sec.~\ref{average_non_increasing}), and mixtures of unitary
operations. In this case the resources are, respectively, average
energy and entropy, and the structure of the invariant sets is given by the right panel of
Fig.~\ref{fig:invariant_sets_structure}, where $\f_1$  coincides with the ground state of the Hamiltonian
(if the Hamiltonian is non-degenerate), while $\f_2$ coincides with the maximally-mixed state.
Other examples of multi-resource theories can be found, and in future work~\cite{multi_resource_paper}
we will present the general properties of multi-resource theories with different invariant sets structures.
\section{Reversible multi-resource theories}
\label{rev_theory_mult}
In this section we study reversibility in the context of multi-resource theories. We first introduce
a property, which we refer to as the \emph{asymptotic equivalence property}, for multi-resource
theories. We then show that, when a resource theory satisfies this property, we can (uniquely)
quantify the amount of resources needed to perform an asymptotic state transformation. This allows
us to introduce the notion of \emph{batteries}, i.e., systems where each individual resource can be stored,
and to keep track of the changes of the resources during a state transformation.
Furthermore, we show that a theory which satisfies the asymptotic equivalence property is also reversible, that is,
the amount of resources exchanged with the batteries during an asymptotic state transformation mapping
$\rho$ into $\sigma$ is equal, with negative sign, to the amount of resources exchanged when mapping
$\sigma$ into $\rho$. Finally, we show that, when the invariant sets of the theory satisfy some general
properties, and the theory satisfies asymptotic equivalence, then the relative entropy distances from the
different invariant sets are the unique measures of the resources. This result is a generalisation of the one
obtained in single-resource theories, see Ref.~\cite{horodecki_are_2002,horodecki_quantumness_2012,brandao_reversible_2015}.
\subsection{Asymptotic equivalence property}
\label{asympt_eqiv_prop}
Let us consider the multi-resource theory $\rt_{\text{multi}}$ introduced in Sec.~\ref{multi_resource}.
This theory has $m$ resources, its set of allowed operations $\A_{\text{multi}}$ is defined in
Eq.~\eqref{all_ops_multi}, and its invariant sets are the $\f_i$'s, that is, the sets of free states of
the different single-resource theories composing it. The multi-resource theory $\rt_{\text{multi}}$ is
\emph{reversible} if the amount of resources spent to perform an asymptotic state transformation is
equal to the amount of resources gained when the inverse state transformation is performed. In this
way, performing a cyclic state transformation over the system (which recovers its initial state at the
end of the transformation) never consumes any of the $m$ resources initially present in the system.
\par
For a single-resource theory, the notions of reversibility and state transformation are usually associated with
the \emph{rates of conversion}. Suppose that we are given $n \gg 1$ copies of a state $\rho \in \SH$, and
we want to find out the maximum number of copies of the state $\sigma \in \SH$ that can be obtained by
acting on the system with the allowed operations. If $k$ is the maximum number of copies of $\sigma$
achievable, then the rate of conversion is defined as $R(\rho \rightarrow \sigma) = \frac{k}{n}$, see
Def.~\ref{def:rate_conversion} in appendix~\ref{rev_theory_sing}. Reversibility is then defined by asking
that, for all $\rho, \sigma \in \SH$, the rates of conversion associated to the forward and backward state
transformations are such that $R(\rho \rightarrow \sigma) R(\sigma \rightarrow \rho) = 1$, see
Def.~\ref{def:reversibility} in the appendix. It is worth noting that, when considering rates of conversion,
one is in general allowed to trace out part of the system, or to add ancillary systems in a free state. For
example, being able to map $n$ copies of $\rho$ into $k$ copies of $\sigma$, with $n < k$, implies that
we have the possibility to add $k-n$ copies in a free state to the initial $n$ copies of $\rho$, and to act
globally to produce $k$ copies of $\sigma$. This is certainly possible for single-resource theories,
where free states always exists, but not always possible for multi-resource theories, see the invariant
set structure of the right panel of Fig.~\ref{fig:invariant_sets_structure}.
\par
Due to the possible absence of free states in a generic multi-resource theory, we first need to
introduce the following definition\footnote{Notice that this definition is analogous to the notion
of ``seed regularisation'' in Ref.~\cite[Sec.~6]{fritz_resource_2015}, although in our case we
are solely focused on reversible transformations and on equalities of monotones.}, which will then allow us to
study reversibility.
\begin{definition}
\label{def:asympt_equivalence_multi}
Consider a multi-resource theory $\rt_{\mathrm{multi}}$. We say that $\rt_{\mathrm{multi}}$ satisfies the
\emph{asymptotic equivalence property} with respect to the set of monotones $\left\{ f_i \right\}_{i=1}^m$,
where each $f_i$ is a monotone for the corresponding single-resource theory $\rt_i$ whose regularisation
is not identically zero, if for all $\rho, \sigma \in \SH$ we have that the following two statements are equivalent,
\begin{itemize}
\item $f_i^{\infty}(\rho) = f_i^{\infty}(\sigma)$ for all $i = 1, \ldots , m$.
\item There exist a sequence of maps $\left\{ \tilde{\chn}_n : \SHn{n} \rightarrow \SHn{n} \right\}_n$
such that 
\begin{equation}
\lim_{n \rightarrow \infty} \left\| \tilde{\chn}_n(\rho^{\otimes n}) - \sigma^{\otimes n} \right\|_1 = 0,
\end{equation}
as well as a sequence of maps performing the reverse process. The maps $\left\{\tilde{\chn}_n\right\}$
are defined as
\begin{equation}
\label{allowed_ancilla}
\tilde{\chn}_n(\cdot) = \Tr{A}{\chn_n(\cdot \otimes \eta^{(A)}_n)},
\end{equation}
where $A$ is an ancilla composed by a sub-linear number $o(n)$ of copies of the system,
and it is described by an arbitrary state $\eta^{(A)}_n \in \SHn{o(n)}$, such that $f_i(\eta^{(A)}_n)
= o(n)$ for all $i = 1, \ldots , m$. The map $\chn_n \in \A_{\mathrm{multi}}^{(n+o(n))}$ is an
allowed operation of the multi-resource theory.
\end{itemize}
Here, $f_i^{\infty}$ is the regularisation of the monotone $f_i$, $\| \cdot \|_1$ is the trace norm,
define as $\| O \|_1 = \tr{\sqrt{O^{\dagger}O}}$ for $O \in \BH$, and we are using the little-o notation,
where $g(n) = o(n)$ means $\lim_{n \rightarrow \infty} \frac{g(n)}{n} = 0$.
\end{definition}
An example of a multi-resource theory that satisfies the above property is thermodynamics
(even in the case in which multiple conserved quantities are present), as shown in
Refs.~\cite{sparaciari_resource_2016,bera_thermodynamics_2017}. In this example the monotones
for which asymptotic equivalence is satisfied are the average energy and the Von Neumann entropy of
the system. Notice that the above property implicitly assumes that the monotones $f_i$'s can be regularised,
that is, that the limit involved in the regularisation is always finite. Furthermore, in this property
we are allowing the agent to act over many copies of the system with more than just the set of allowed
operations; we assume the agent to be able to use a small ancillary system, sub-linear in the number of
copies of the main system. Roughly speaking, the role of this ancilla is to absorb the fluctuations in the
monotones $f^{\infty}_i$'s during the asymptotic state transformation. It is important to notice that this
ancillary system only contributes to the transformation by exchanging a sub-linear amount of
resources. Thus, its contribution per single copy of the system is negligible when $n \gg 1$, 
which justifies the use of this additional tool.
\par
The asymptotic equivalence property
essentially states that the multi-resource theory can reversibly map between any two states with
the same values of the monotones $f_i$'s. In particular, transforming between such two states comes
at no cost, since we can do so by using the allowed operations of the theory, $\A_{\text{multi}}$. It is
worth noting that, when the number of considered resources is $m = 1$, that is, our theory is a single-resource
theory, the notion of asymptotic equivalence given in Def.~\ref{def:asympt_equivalence_multi} corresponds
to the one given in terms of rates of conversion, Def.~\ref{def:reversibility}. We prove this equivalence in
appendix~\ref{rev_theory_sing}, see Thm.~\ref{thm:reversible_asympt_equiv}. The set of monotones
in Def.~\ref{def:asympt_equivalence_multi} is not a priori unique; however, in the following section we identify
the properties that the monotones need to satisfy for this set to be unique, see Thm.~\ref{thm:reversible_multi}.
Finally, notice that the asymptotic equivalence property does not say anything about the state transformations which
involve states with different values of the monotones $f_i$'s. To include these transformations in the theory,
we will have to add a bit more structure to the current framework, by considering some additional
systems that can store a single type of resource each, which we refer to as
\emph{batteries}~\cite{kraemer_currencies_2016}.
\subsection{Quantifying resources with batteries}
\label{quant_res}
When a multi-resource theory satisfies the asymptotic equivalence property of
Def.~\ref{def:asympt_equivalence_multi}, we have that states with the same values of a specific set of
monotones can be inter-converted between each others. In this section, we show that these monotones
actually quantify the amount of resources contained in the system. To do so, we need to introduce some
additional systems, which can only store a single kind of resource each, and can be independently addressed
by the agent. These additional systems are referred to as batteries. Let us suppose that the multi-resource
theory $\rt_{\text{multi}}$ satisfies the asymptotic equivalence property with respect to the set of monotones
$\left\{ f_i \right\}_{i=1}^m$, and that the quantum system over which the theory acts is actually divided
into $m+1$ partitions. The first partition is the main system $S$, and the remaining ones are the batteries
$B_i$'s. Then, the Hilbert space under consideration is $\hil = \hil_S \otimes \hil_{B_1} \otimes \ldots \otimes
\hil_{B_m}$.
\par
Let us now introduce some properties the monotones need to satisfy in order for
the resources to be quantified in a meaningful way. Since each resource is associated
to a different monotone, we can forbid a battery to store more than one resource by
constraining the set of states describing it to those ones with a fixed value
of all but one monotones.
\begin{description}
\item[M1\label{item:M1}] Consider two states $\omega_i, \omega'_i  \in \SHbI$ describing
the battery $B_i$. Then, the value of the regularisation of any monotone $f_{j}$ (where $j \neq i$)
over these two states is fixed,
\begin{equation}
\label{eq:batt_mon}
f^{\infty}_j(\omega'_i) = f^{\infty}_j(\omega_i ) , \quad \forall j \neq i.
\end{equation}
\end{description}
In this way, the battery $B_i$ is only able to store and exchange the resource associated with
the monotone $f_i$. It would be natural to extend the condition of Eq.~\eqref{eq:batt_mon}
to the monotones themselves, rather than to use their regularisations. However, this
condition is not required for deriving our results, and to use it in our proofs we would need an
additional assumption, namely the additivity of the monotones.
\par
In order to address each battery
as an individual system, we ask the value of the monotones over the global system to be given
by the sum of their values over the individual components,
\begin{description}
\item[M2\label{item:M2}] The regularisations of the monotones $f_i$'s can be separated between main
system and batteries,
\begin{equation}
f^{\infty}_i(\rho \otimes \omega_{1} \otimes \ldots \otimes \omega_{m})
=
f^{\infty}_i(\rho) + f^{\infty}_i(\omega_{1}) + \ldots + f^{\infty}_i(\omega_{m}),
\end{equation}
where $\rho \in \SHs$ is the state of the main system, and $\omega_{i} \in \SHbI$ is the
state of the battery $B_i$.
\end{description}
The above property allows us to separate the contribution given by each subsystem to the amount
of $i$-th resource present in the global system. It is important to stress that we are here requiring
additivity for the regularisation of the monotones between system and batteries, and between batteries,
but we are not requiring the regularised monotones to be additive in general.
\par
We then ask the
monotones to satisfy an additional property, so as to simplify the notation. Namely, we ask the zero of each monotone
$f_i$ to coincide with its value over the states in $\f_i$,
\begin{description}
\item[M3\label{item:M3}] For each $n \in \N$ and $i \in \left\{ 1, \ldots , m \right\}$, the monotone $f_i$ is equal to
$0$ when computed over the states of $\f_i^{(n)}$, that is
\begin{equation}
f_i(\gamma_{i,\,n}) = 0 , \qquad \forall \, \gamma_{i,\,n} \in \f_i^{(n)}.
\end{equation}
\end{description}
This property serves as a way to ``normalise'' the monotone, setting its value to $0$ over the states
that were free for the specific single-resource theory the monotone is linked to. Notice that
property~\ref{item:M3} is trivially satisfied by any monotone after a translation. The next
property requires that tracing out part of the system does not increase the value of the
monotones $f_i$'s,
\begin{description}
\item[M4\label{item:M4}] For all $n, k \in \N$ where $k < n$, the monotones $f_i$'s are such that
\begin{equation}
f_i(\Tr{k}{\rho_n}) \leq f_i(\rho_n) , \qquad \forall \, i \in \left\{ 1, \ldots , m \right\}.
\end{equation}
where $\rho_n \in \SHn{n}$ and $\Tr{k}{\rho_n} \in \SHn{n-k}$.
\end{description}
This property implies that the resources contained in a system cannot increase if we
discard/forget part of it.
\par
We require our monotones to satisfy sub-additivity, namely
\begin{description}
\item[M5\label{item:M5}] For all $n, k \in \N$, the monotones $f_i$'s are such that
\begin{equation}
f_i(\rho_n \otimes \rho_k) \leq f_i(\rho_n) + f_i(\rho_k) , \qquad \forall \, i \in \left\{ 1, \ldots , m \right\}.
\end{equation}
where $\rho_n \in \SHn{n}$ and $\rho_k \in \SHn{k}$.
\end{description}
That is, the amount of resources contained in two uncorrelated systems, when measured on the
two systems independently, is bigger or equal to the value measured on the two systems together.
This is the case, for example, of the relative entropy of entanglement~\cite{vollbrecht_entanglement_2001}.
Notice that sub-additivity is here explicity required since, as we stressed before, property~\ref{item:M2}
only requires additivity between system and different batteries, but not between different partitions of the
individual system or battery. Another property we require is for the monotones $f_i$'s to be sub-extensive,
\begin{description}
\item[M6\label{item:M6}] Given any sequence of states $\left\{ \rho_n \in \SHn{n} \right\}$, the
monotones $f_i$'s are such that
\begin{equation}
f_i(\rho_n) = O(n) , \qquad \forall \, i \in \left\{ 1, \ldots , m \right\}.
\end{equation}
where we are using the big-O notation.
\end{description}
This property is satisfied, for example, if the monotones scale extensively, that is, if they scale
linearly in the number of systems considered. In the next section we will encounter a family of
monotones which indeed satisfy this property, namely the relative entropy distance from a given set of free
states, when some fairly generic conditions are satisfied by such set (see Prop.~\ref{thm:properties_rel_ent}).
However, it is worth noting that property~\ref{item:M6} is not equivalent to extensivity, since a
monotone scaling sub-linearly in the number of systems would still satisfy it. We demand that our
monotones satisfy this property so as to be able to regularise them (although their regularisation might be
identically zero on the whole state space). The last property we ask concerns a particular kind of
continuity the monotones need to satisfy,
\begin{description}
\item[M7\label{item:M7}] The monotones $f_i$'s are \emph{asymptotic continuous}, that is, for all sequences
of states $\rho_n, \sigma_n \in \SHn{n}$ such that $\left\| \rho_n - \sigma_n \right\|_1 \rightarrow 0$
for $n \rightarrow \infty$, where $\| \cdot \|_1$ is the trace norm, we have
\begin{equation}
\frac{\left| f_i \left( \rho_n \right) - f_i \left( \sigma_n \right) \right|}{n} \rightarrow 0
\ \text{for} \ n \rightarrow \infty , \qquad \forall \, i \in \left\{ 1, \ldots , m \right\}.
\end{equation}
This notion of asymptotic continuity coincides with condition (C2) given in Ref.~\cite{horodecki_entanglement_2001}.
\end{description}
This property implies that the monotones are physically meaningful, since their values
over sequences of states converge if the sequences of states converge asymptotically.
In Thm.~\ref{thm:reversible_multi} we show that, when the monotones satisfy
asymptotic continuity, they are the unique quantifiers of the amount of resources
contained in the main system.
\par
We can now use this formalism to discuss how resources can be quantified in a multi-resource theory,
and consequently how the asymptotic equivalence property implies that the theory is reversible.
Let us consider any two states $\rho, \sigma \in \SHs$, that do not need to have the same values for
the monotones $f_i$'s. Then, we choose the initial and final states of each battery $B_i$ such that
\begin{equation}
\label{eq:fi_condition_Ri}
f^{\infty}_i \left( \rho \otimes \omega_1 \otimes \ldots \otimes \omega_m \right)
=
f^{\infty}_i \left( \sigma \otimes \omega'_1 \otimes \ldots \otimes \omega'_m \right)
, \qquad
\forall \, i = 1, \ldots, m,
\end{equation}
where $\omega_i$, $\omega'_i \in \SHbI$, for $i = 1, \ldots, m$. Under these conditions, due to the
asymptotic equivalence property of $\rt_{\text{multi}}$, we have that the two global states can be
asymptotically mapped one into the other in a reversible way, using the allowed operations of the
theory, that is
\begin{equation} \label{eq:transform}
\rho \otimes \omega_1 \otimes \ldots \otimes \omega_m
\xleftrightarrow{\text{asympt}}
\sigma \otimes \omega'_1 \otimes \ldots \otimes \omega'_m,
\end{equation}
where the symbol $\xleftrightarrow{\text{asympt}}$ means that there exists two allowed
operations that maps $n \gg 1$ copies of the state on the lhs into the state of the rhs, and viceversa,
while satisfying the condition in the second statement of Def.~\ref{def:asympt_equivalence_multi}.
\par
We can now properly define the notion of resources in this framework. The resource associated with
the monotone $f_i$ is the one exchanged by the battery $B_i$ during a state transformation.
\begin{definition}
Consider a multi-resource theory $\rt_{\text{multi}}$ with $m$ resources, satisfying the asymptotic
equivalence property with respect to the set of monotones $\left\{ f_i \right\}^m_{i=1}$. For a state
transformation of the form given in Eq.~\eqref{eq:transform}, we define the amount of $i$-th resource
exchanged between the system $S$ and the battery $B_i$ as
\begin{equation} \label{work_Ri}
\Delta W_i := f^{\infty}_i(\omega'_i) - f^{\infty}_i(\omega_i),
\end{equation}
where $\omega_i, \omega'_i \in \SHbI$ are, respectively, the initial and final state of the battery $B_i$.
\end{definition}
It is now possible to compute the amount of the $i$-th resource $\Delta W_i$ needed to map the state of
the main system $\rho$ into $\sigma$.
\begin{prop}
Consider a theory $\rt_{\text{multi}}$ with $m$ resources and allowed operations $\A_{\text{multi}}$,
equipped with batteries $B_1$, $\ldots$, $B_m$. If the theory satisfies the asymptotic equivalence
property with respect to the set of monotones $\left\{ f_i \right\}^m_{i=1}$, and these monotones
satisfy the properties~\ref{item:M1} and~\ref{item:M2}, then the amount of $i$-th resource needed
to perform the asymptotic state transformation $\rho \rightarrow \sigma$ is equal to
\begin{equation}
\label{resource_work_i}
\Delta W_i = f_i^{\infty}(\rho) - f_i^{\infty}(\sigma).
\end{equation}
\end{prop}
\begin{proof}
Due to asymptotic equivalence, a transformation mapping the global state $\rho \otimes \omega_1 \otimes
\ldots \otimes \omega_m$ into $\sigma \otimes \omega'_1 \otimes \ldots \otimes \omega'_m$ exists iff the
conditions in Eq.~\eqref{eq:fi_condition_Ri} are satisfied. For a given $i$, using the property~\ref{item:M2} of
the monotone $f_i$, we can re-write the condition as
\begin{equation}
f^{\infty}_i \left( \rho \right) + f^{\infty}_i \left( \omega_1 \right) + \ldots + f^{\infty}_i \left( \omega_m \right)
=
f^{\infty}_i \left( \sigma \right) + f^{\infty}_i \left( \omega'_1 \right) + \ldots + f^{\infty}_i \left( \omega'_m \right).
\end{equation}
Then, we can use the property~\ref{item:M1}, which guarantees that the only systems for which $f_i$
changes are the main system and the battery $B_i$. Thus, we find that
\begin{equation}
f^{\infty}_i \left( \rho \right) + f^{\infty}_i \left( \omega_i \right)
=
f^{\infty}_i \left( \sigma \right) + f^{\infty}_i \left( \omega'_i \right).
\end{equation}
By rearranging the factors in the above equation, and using the definition of $\Delta W_i$ given in
Eq.~\eqref{work_Ri}, we prove the proposition.
\end{proof}
It is now easy to show that, if $\rt_{\text{multi}}$ satisfies the asymptotic equivalence property, any
state transformation on the main system $S$ is reversible. Indeed, from Eq.~\eqref{resource_work_i}
it follows that the amount of resources used to map the state of this system from $\rho$ to $\sigma$
is equal, but with negative sign, to the amount of resources used to perform the reverse transformation,
from $\sigma$ to $\rho$. Therefore, any cyclic state transformation over the main system leaves
the amount of resources contained in the batteries unchanged.
\par
The above formalism also provides us with a way to quantify the amount of resources contained
in the main system. Indeed, if the system is described by the state $\rho \in \SHs$, the amount
of $i$-th resource contained in the system is given by the amount of $i$-th resource exchanged,
$\Delta W_i$, while mapping $\rho$ into a state contained in $\f_i$. Using property~\ref{item:M3} 
and Prop.~\eqref{resource_work_i} it follows that
\begin{coro}
\label{cor:quantifier}
Consider a theory $\rt_{\text{multi}}$ with $m$ resources and allowed operations $\A_{\text{multi}}$,
equipped with batteries $B_1$, $\ldots$, $B_m$. If the theory satisfies the asymptotic equivalence
property with respect to the set of monotones $\left\{ f_i \right\}^m_{i=1}$, and these monotones
satisfy the properties~\ref{item:M1}, \ref{item:M2}, and~\ref{item:M3}, then the amount of the $i$-th
resource contained in the main system, when described by the state $\rho$, is given by $f_i^{\infty}(\rho)$.
\end{coro}
It is worth noting that, in general, one cannot extract all the resources contained in the main system at once.
Indeed, this is only possible when the multi-resource theory contains free states, like for example in the cases
depicted in the left and centre panels of Fig.~\ref{fig:invariant_sets_structure}. Furthermore, the process
of resources extraction is in general non-trivial, since property~\ref{item:M1} forbids each battery from storing
more than one kind of resource. As a result, it is not possible to simply perform a swap operation
which exchanges the state of the system with one of the batteries, see Sec.~\ref{control_theory_ex} for an example
involving the theory of local control.
\begin{figure}[t!]
\center
\includegraphics[width=0.5\textwidth]{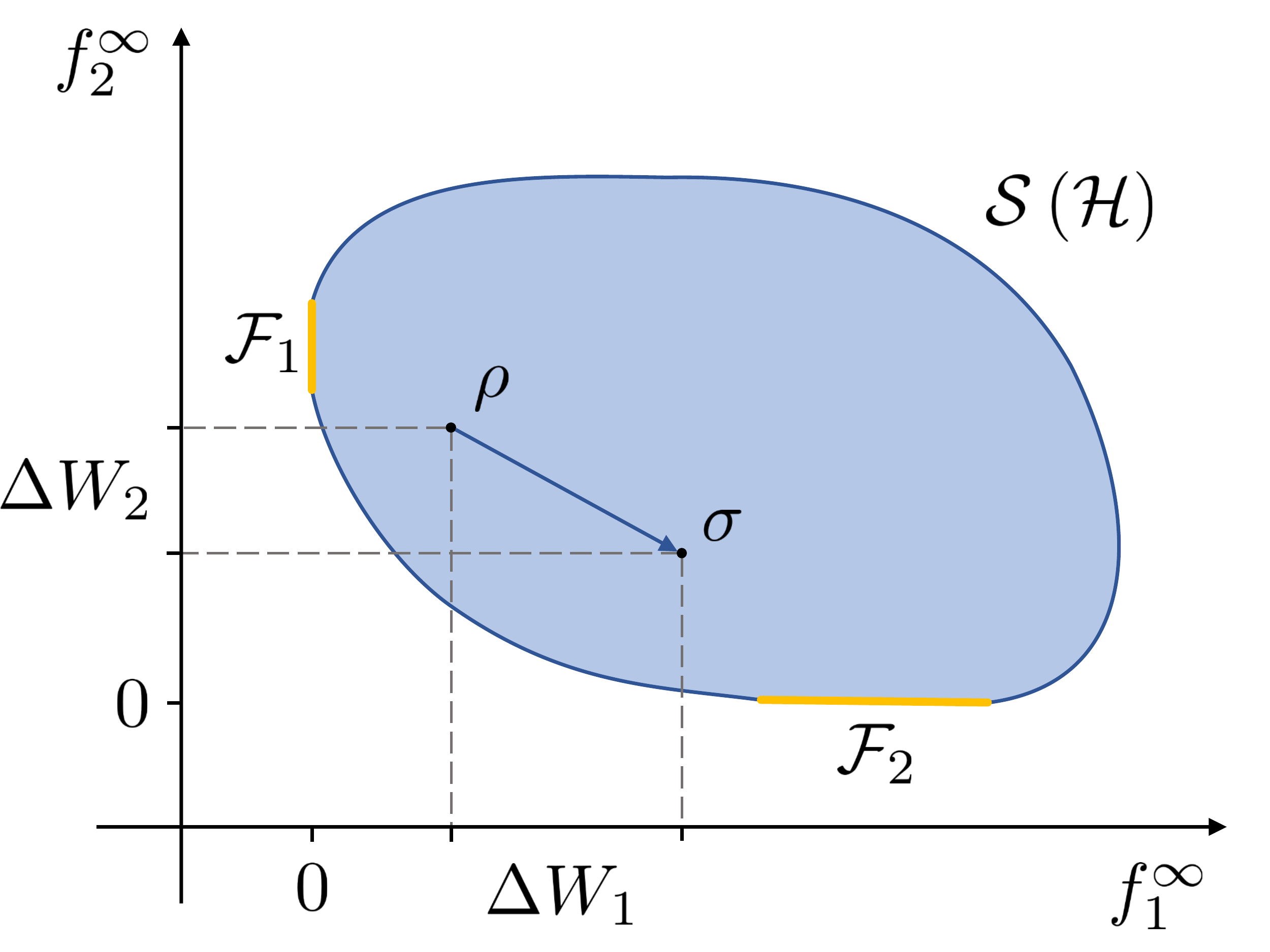}
\caption{In the figure we represent the state-space $\SH$ of a multi-resource theory
$\rt_{\text{multi}}$ with two resources. In order for the diagram to be a meaningful
representation of this state-space, we need the theory to satisfy the asymptotic
equivalence property of Def.~\ref{def:asympt_equivalence_multi} with respect to the
monotones $f_1$ and $f_2$. In fact, when the theory satisfies this property we can divide $\SH$
into equivalence classes of states with the same value of the regularised monotones
$f^{\infty}_1$ and $f^{\infty}_2$, which become the abscissa and ordinate of the diagram.
The state-space of the theory is represented by the blue region, and the yellow
segments are the invariant sets $\f_1$ and $\f_2$. These sets are disjoint, since the two
segments do not intercept each other, and the resource theory $\rt_{\text{multi}}$ thus
corresponds to the one depicted in the right panel of Fig.~\ref{fig:invariant_sets_structure}.
Two equivalence classes, respectively associated to the states $\rho$ and $\sigma$, are
represented in the diagram. The amount of resources that is exchanged when transforming
from one state to the other, Eq.~\eqref{resource_work_i}, is given in the diagram by the difference
between the coordinates of these two points. Notice that, for ease of viewing, we have shifted
the origin of the axis.}
\label{fig:resource_diagram}
\end{figure}
\par
Being able to quantify the amount of resources contained in a given quantum state allows us to represent
the whole state-space of the theory in a \emph{resource diagram}~\cite{fritz_resource_2015,sparaciari_resource_2016}.
In fact, from the definition of asymptotic equivalence it follows that, if two states contain the same
amount of resources, i.e., if they have the same values of the monotones $f^{\infty}_i$'s, then we
can map between them using the allowed operations $\A_{\text{multi}}$. This property implies that
we can divide the entire state-space into equivalence classes, that is, sets of states with same value
of the $m$ monotones (where we recall that $m$ is the number of resources, or batteries, in the theory).
Then, we can represent each equivalence class as a point in a $m$-dimensional diagram, with coordinates
given by the values of the monotones. By considering all the different equivalence classes, we can
finally represent the state-space of the main system in the diagram, see for example
Fig.~\ref{fig:resource_diagram}, where the state-space of a two-resource theory is shown.
\subsection{Reversibility implies a unique measure for each resource}
\label{multi_rev_unique}
We now show that, when a multi-resource theory satisfies the asymptotic equivalence
property with respect to a set of monotones $\left\{ f_i \right\}_{i=1}^m$, and these monotones
satisfy the properties~\ref{item:M1} -- \ref{item:M7}, then there exists a unique quantifier
for each resource contained in the main system. In particular, when the $i$-th resource is
considered, this quantifier coincides with $f_i^{\infty}$, modulo a multiplicative factor which
sets the scale. It is worth noting that such multiplicative factor can be different for each
resource. Indeed, the resources are generally independent of each other, since they are quantified
by different measures. Each measure can have a different unit, which corresponds to an individual
rescaling factor applied to each resource measure independently.
\par
In the previous section, Cor.~\ref{cor:quantifier}, we showed that a quantifier
exists if the monotones satisfy the first three properties~\ref{item:M1}, \ref{item:M2}, and \ref{item:M3}.
However, when these monotones are also asymptotic continuous, property~\ref{item:M7},
we can prove that they \emph{uniquely} quantify the amount of resources contained in the
main system. This means that one cannot find other monotones $g_i$'s that give the
same equivalence classes of the $f_i$'s, but order them in a different way. Asymptotic continuity
was used in Ref.~\cite{horodecki_quantumness_2012} to show that the relative entropy distance
from the set of free states of a reversible single-resource theory is the unique measure of resource.
Thus, the following theorem (whose proof can be found in appendix~\ref{main_results}) can be
understood as a generalisation of the above result to multi-resource theories,
\begin{restatable}{thm}{uniquemeas}
\label{thm:reversible_multi}
Consider the resource theory $\rt_{\text{multi}}$ with $m$ resources, equipped with the batteries $B_i$'s,
where $i = 1, \ldots, m$. Suppose the theory satisfies the asymptotic equivalence property with respect to
the set of monotones $\left\{ f_i \right\}_{i=1}^m$. If these monotones satisfy the properties~\ref{item:M1}
-- \ref{item:M7}, then the amount of $i$-th resource contained in the main system $S$ is uniquely quantified
by the regularisation of the monotone $f_i$ (modulo a multiplicative constant).
\end{restatable}
In particular, we now consider the case of a multi-resource theory $\rt_{\text{multi}}$ that satisfies
the asymptotic equivalence property of Def.~\ref{def:asympt_equivalence_multi} with respect to
the relative entropy distances from the invariant sets $\f_i$'s. We refer to the
relative entropy distance from the set $\f_i$ as $E_{\f_i}$, whose definition can be found in
Eq.~\eqref{rel_entr_dist}. Since the multi-resource theory we consider is equipped with batteries,
and we want to be able to measure the amount of resources they contain independently of the other subsystems,
we ask the invariant sets to be of the form
\begin{equation}
\label{eq:independent_free}
\f_i = \f_{i,S} \otimes \f_{i,B_1} \otimes \ldots \otimes \f_{i,B_m},
\end{equation}
so that the main system $S$ and the batteries $B_i$'s all have they own independent invariant sets.
We now show that, under very general assumptions over the properties of the invariant sets, the regularised
relative entropy distances from these sets are the unique quantifiers of the resources, provided
that these quantities are not identically zero over the whole state space\footnote{An example where
the regularised relative entropy from an invariant set is identically zero for all states in $\SH$ is
the resource theory of asymmetry, see Ref.~\cite{gour_measuring_2009}.}. This result follows from
Thm.~\ref{thm:reversible_multi}, and from the fact that these monotones satisfy the properties~\ref{item:M1},
\ref{item:M2}, \ref{item:M3}, and \ref{item:M7} listed in the previous sections. The properties we are
interested in for the invariant sets $\left\{ \f_i \right\}_{i=1}^m$ of the theory are very general,
and they are satisfied in most of the known resource theories, see
Refs.~\cite{brandao_reversible_2010, brandao_generalization_2010}.
\begin{description}
\item[F1\label{item:F1}] The sets $\f_i$'s are closed sets.
\item[F2\label{item:F2}] The sets $\f_i$'s are convex sets.
\item[F3\label{item:F3}] Each set $\f_i$ contains at least one full-rank state.
\item[F4\label{item:F4}] The sets $\f_i$'s are closed under tensor product, that is,
$\f_i^{(k)} \otimes \f_i^{(n)} \subseteq \f_i^{(n+k)}$ for all $i = 1, \ldots , m$.
\item[F5\label{item:F5}] The sets $\f_i$'s are closed under partial tracing, that is,
$\Tr{k}{\f_i^{(n)}} \subseteq \f_i^{(n-k)}$ for all $i = 1, \ldots , m$.
\end{description}
Let us briefly comment on the above properties. Property~\ref{item:F1} requires that any
converging sequence in the set converges to an element in the set. This property is necessary
for the continuity of the resource theory. Property~\ref{item:F2}, instead, tells us that we are allowed
to forget the exact state describing the system, and therefore we can have mixture of states.
Property~\ref{item:F3} is necessary for the relative entropy distance to be physically natural,
since the quantity $\re{\rho}{\sigma}$, see Eq.~\eqref{rel_entr}, diverges when
$\text{supp}(\rho) \not\subseteq \text{supp}(\sigma)$. Finally, property~\ref{item:F4} implies that
composing two systems that do not contain any amount of $i$-th resource is not going to increase
that resource, and similarly, property~\ref{item:F5} implies that forgetting about part of a
system which does not contain resources will not create resources.
\par
When the invariant sets satisfy the above properties, the relative entropy distances
$E_{\f_i}$'s satisfy the same properties discussed in the previous section,
\begin{restatable}{prop}{relentproperties}
\label{thm:properties_rel_ent}
Consider a resource theory $\rt_{\text{multi}}$ with $m$ resources, equipped with the batteries
$B_i$'s, where $i = 1, \ldots, m$. Suppose the class of allowed operations is $\A_{\text{multi}}$
and the invariant sets are $\left\{ \f_i \right\}_{i=1}^m$. If the invariant set $\f_i$ is of the
form of Eq.~\eqref{eq:independent_free}, and it satisfies the properties~\ref{item:F1} -- \ref{item:F5},
then the relative entropy distances from this set, $E_{\f_i}$, is a regularisable monotone under
the class of allowed operations, and it obeys the properties~\ref{item:M1} -- \ref{item:M7}.
\end{restatable}
This result is known in the literature, see Refs.~\cite{synak-radtke_asymptotic_2006,
brandao_generalization_2010}, but we nevertheless provide a proof in appendix~\ref{additional}
to make the paper self-contained. By virtue of Thm.~\ref{thm:reversible_multi} it then follows that,
if $E^{\infty}_{\f_i}$ has a positive value over the states that are not in $\f_i$, then it is the unique
quantifier of the amount of $i$-th resource contained in the system for a multi-resource theory that
satisfies the asymptotic equivalence property with respect to these monotones. Furthermore, the amount
of $i$-th resource used to map the main system from the state $\rho$ into the state $\sigma$ is then
equal to 
\begin{equation}
\label{resource_work_ent}
\Delta W_i = E_{\f_i}^{\infty}(\rho) - E_{\f_i}^{\infty}(\sigma),
\end{equation}
for all $i = 1, \ldots , m$.
\subsection{Relaxing the conditions on the monotones}
\label{average_non_increasing}
There are situations, when we consider specific resource theories, in which some of the properties of the
set of free states are not satisfied. In particular, we can have that the set of free states does not
contain a full-rank state, that is, property~\ref{item:F3} is not satisfied. An example would be the resource
theory of energy-non-increasing maps for a system with Hamiltonian $H$,
\begin{equation}
\label{non_increasing_energy}
\A_{H} = \left\{ \chn_H \ : \ \BH \rightarrow \BH |
\ \tr{\chn_{H} (\rho) H} \leq \tr{\rho H} \ \forall \rho \in \SH \right\}.
\end{equation}
An example of a subset of $\A_{H}$ are unitary operations which commute with the Hamiltonian $H$ (as in
the resource theory of Thermal Operations). If the Hamiltonian $H$ has a non-degenerate ground
state $\ket{\text{g}}$, then it is easy to show that this state is fixed, that is,
\begin{equation}
\chn_H \left( \ket{\text{g}}\bra{\text{g}} \right) = \ket{\text{g}}\bra{\text{g}}.
\end{equation}
In fact, the operation $\chn_{\text{g}} (\cdot) = \Tr{A}{S ( \cdot \otimes
\ket{\text{g}}\bra{\text{g}}_A ) S^{\dagger}}$, where $S$ is the unitary operation implementing
the swap between the two states, belongs to $\A_{H}$ and maps all states into the ground state.
Thus, the set of free states does not contain a full-rank state, which implies that the relative
entropy distance from this set is ill-defined, and it is not asymptotic continuous. Notice
that the above argument holds even in the case of a degenerate ground state, with the
difference that the invariant set would be composed by any state with support on this degenerate
subspace.
\par
We can introduce a different monotone for this kind of resource theory, that is, the average of the
observable which is not increased by the allowed operations (modulo a constant factor). For the
example we are considering, this monotone would be
\begin{equation} \label{montone_average}
M_H(\rho) = \tr{H \rho} - E_{\text{g}},
\end{equation}
where $H$ is the Hamiltonian of the system, and $E_{\text{g}} = \tr{H \ket{\text{g}}\bra{\text{g}}}$
is the energy of the ground state. When $n$ copies of the system are considered, we define the total
Hamiltonian as $H_n = \sum_{i=1}^{n} H^{(i)}$, where $H^{(i)}$ is the Hamiltonian acting on the $i$-th
copy. In this case, it is easy to show that this quantity is equal to $0$ when evaluated on the fixed
state $\ket{\text{g}}\bra{\text{g}}$, property~\ref{item:M3}, is monotonic under partial tracing,
property~\ref{item:M4}, is additive (and therefore satisfies sub-additivity, property~\ref{item:M5}), and
it scales extensively in the number of copies of the system, thus satisfying property~\ref{item:M6}. Furthermore,
$M_H(\cdot)$ is monotonic under the class of operations (by definition of the class itself), and it is
asymptotic continuous, property~\ref{item:M7}, as shown in Prop.~\ref{average_asymp_cont} in
appendix~\ref{additional}. If batteries are introduced, we can define the operator $H$ is such a
way that properties~\ref{item:M1} and \ref{item:M2} are satisfied, see for example Sec.~\ref{thermo_example}.
\par
Thus, if one (or more) of the monotones of the multi-resource
theory is of the form given in Eq.~\eqref{montone_average}, we have that the results of the previous
section still apply, particularly Thm.~\ref{thm:reversible_multi}. Furthermore, we can quantify the change
in the resource associated with $M_H$ during a state transformation $\rho \rightarrow \sigma$ with
 Eq.~\eqref{resource_work_ent}, where the regularised relative entropy distance $E_{\f_i}^{\infty}$
is replaced with the regularised monotone $M_H^{\infty}$. As a side remark, we notice that the
monotone $M_H$ can be obtained as
\begin{equation}
M_H(\rho) = \lim_{\beta \rightarrow \infty} \frac{1}{\beta} \, \re{\rho}{\tau_{\beta}},
\end{equation}
where $\tau_{\beta} = e^{- \beta H} / Z$ is the Gibbs state of the Hamiltonian $H$, and $Z = \tr{e^{- \beta H}}$
is the partition function of the system.
\section{Bank states, interconversion relations, and the first law}
\label{interconv}
Within certain types of multi-resource theories, it is possible to inter-convert the resources stored in the
batteries, i.e., to exchange one resource for another at a given exchange rate. Examples of resource
interconversion can be found in thermodynamics, where Landauer's principle~\cite{landauer_irreversibility_1961}
tells us that energy can be exchanged for information, while a Maxwell's demon can trade information for
energy~\cite{bennett_thermodynamics_1982}. In these examples, a thermal bath is necessary to perform
the interconversion of resources. Indeed, in the following sections we show that in order to exchange
between resources one always needs an additional system, which we refer to as a \emph{bank}, that
captures the necessary properties of thermal baths in thermodynamics, and abstracts them so that they
can be applied to other resource theories. When such a system exists, we can pay a given amount
of one resource and gain a different amount of another resource, with an exchange rate that only
depends on the state describing the bank, see Thm.~\ref{thm:interconvert_relation}.
Within the thermodynamic examples we are considering, this corresponds to exchanging one
bit of information for one unit of energy, and vice versa. The exchange rate of these processes
is proportional to the temperature of the thermal bath.
\par
During a resource interconversion the state of the bank should not change its main properties, so that
we can keep using it indefinitely. Furthermore, we should always have to invest one resource in
order to gain the other. For these reasons the bank is taken to be of infinite size, and its state
to be \emph{passive}, i.e., to always contain the minimum possible values of the resources. In
fact, in the thermodynamic examples we are considering, the thermal bath has infinite size, and
its state has maximum entropy for fixed energy, or equivalently minimum energy for fixed
entropy~\cite{jaynes_information_1957}. We additionally show that the relative entropy distance from
the set of bank states plays a fundamental role in quantifying the exchange rate at which resources are
inter-converted, see Cor.~\ref{bank_equal_rel_ent}. For instance, in thermodynamics this quantity
is proportional to the Helmholtz free energy $F = E - T S$, which links together the two resources,
internal energy $E$ and information, which is proportional to $-S$. Through this quantity, one can define
the exchange rate between energy and entropy, i.e., the temperature of the thermal bath $T$. Finally,
we introduce a first-law-like relation for multi-resource theories. The first law consists of a single relation
that regulates the state transformation of a system when the agent has access to a bank for exchanging
the resources. In particular, this relation links the change in the relative entropy distance from the set of
bank states over the main system to the amount of resources exchanged by the batteries during the
transformations, see Cor.~\ref{coro:first_law}. In the example we are considering, this relation coincides
with the First Law of thermodynamics, as it connects a change in the Helmholtz free energy $\Delta F$ of
the system with the energy and information exchanged by the batteries,
\begin{equation}
\label{first_law}
\Delta F = \Delta W_E + T \, \Delta W_I,
\end{equation}
where $\Delta W_E$ is the energy exchanged by the first battery, $\Delta W_I$ is the information exchanged by
the second battery, and $T$ is the background temperature, describing the state of the bank.
\par
We now briefly discuss about the value that resources have in the different theories of thermodynamics,
and the role of the first law in connecting these resources together. Let us first consider the single-resource
theory of thermodynamics, where the system is in contact with an infinite thermal reservoir~\cite{brandao_resource_2013}.
To perform a state transformation we need to provide only one kind of resource, known as athermality
($\Delta F$), or work. Since the thermal reservoir is present, it is easy to get close to the free state, i.e.
to the thermal state at temperature $T$, because we can simply thermalise the system with the allowed
operations. However, it is difficult to go in the opposite direction, unless we use part of the athermality
stored in a battery. For this reason, a positive increment in the athermality of the battery is considered
valuable, while a negative change is considered a cost.
\par
Let us now move to the multi-resource theory of thermodynamics, whose allowed operations are
energy-preserving unitary operations~\cite{sparaciari_resource_2016}. In this case, it is easy to
see that negative and positive contributions of energy and information are equally valuable, since these
two quantities are conserved by the set of allowed operations. As a result,  the agent cannot perform
state transformations in any direction without having access to the batteries. If we now allow the
agent to use a thermal bath as a bank, and we keep the system decoupled from it (so that the agent
cannot perform operations that thermalise the system for free), we find that changing a single resource,
either energy or information, is enough to perform a generic state transformation on the system. In fact,
we can always inter-convert one resource for the other with the bank, and then change the state of the
system accordingly. Notice that, however, we still have that negative and positive change in one resource
are equally valuable.
\par
Thus, it seems that the advantage that multi-resource theories provide over single-resource theories
is that they make explicit which resources are used during a state transformation. And the link between
the single resource and the multiple ones is given by the first law. In thermodynamics, for example, we
have that the first law, Eq.~\eqref{first_law}, indicates that the amount of athermality $\Delta F$ needed
to transform a state can be actually divided in two contributions, energy $\Delta W_E$ and information
$\Delta W_I$. Notice that all of these quantities can be understood in terms of the relative entropy
distance to an invariant set of states. Athermality being measured by its relative entropy distance to the
thermal state, information and energy being the relative entropy to the maximally mixed or ground state. As
we will see, the generalised first law given in Eq.~\eqref{eq:first_law} also relates the relative entropy
to the bank state, to the relative entropies to the invariant sets of the single resource theories.
\subsection{Banks and interconversion of resources}
\label{bank_interconvert}
We now introduce the bank system, and show how this additional tool allows us to perform interconversion
between resources. To simplify the notation, we only focus on a theory with two resources. However, the
results we obtain also apply to theories with more resources, since in that case we can just select two
resources and perform interconversion while keeping the others fixed. Thus, in the following we consider a
resource theory $\rt_{\text{multi}}$ with two invariant sets $\f_1$ and $\f_2$ (each of them associated with
one of the resources), and allowed operations $\A_{\text{multi}}$. We assume the theory to satisfy the
asymptotic equivalence property of Def.~\ref{def:asympt_equivalence_multi} with respect to the relative
entropy distances from $\f_1$ and $\f_2$, and we ask the two invariant sets to satisfy the
properties~\ref{item:F1}, \ref{item:F2}, and \ref{item:F3}, while we replace properties~\ref{item:F4}
and ~\ref{item:F5} with the following, more demanding, property
\begin{description}
\item[F5b\label{item:F5b}] The invariant sets $\f_i$'s are such that $\f_i^{(n)} = \f_i^{\otimes n}$,
for all $n \in \N$.
\end{description}
The above properties implies that the relative entropy distances $E_{\f_1}$ and $E_{\f_2}$ are the
unique quantifiers for the two resources of our theory, as we have seen in Sec.~\ref{multi_rev_unique}.
From property~\ref{item:F5b} it follows that these two monotones are additive, i.e.,
$E_{\f_i}(\rho \otimes \sigma) = E_{\f_i}(\rho) + E_{\f_i}(\sigma)$ for $i = 1,2$, and consequently
that their regularisation $E^{\infty}_{\f_i}$ coincides with $E_{\f_i}$. Furthermore, the
properties~\ref{item:F2} and \ref{item:F5b} together imply that the invariant sets are composed by
a single state, i.e., $\f_i = \left\{ \rho_i \right\}$, where $\rho_i \in \SH$, for $i = 1,2$. We make
use of property~\ref{item:F5b} in Lem.~\ref{f_i_inequality}, shown in appendix~\ref{additional},
which itself is used to prove some essential properties of the set of bank states, see Def.~\ref{def:bank_state}.
This property is ultimately used to show that the exchange rate between resources is given by
the relative entropy distance from the set of states describing the bank, see Cor.~\ref{bank_equal_rel_ent}.
\par
It is important to stress that property~\ref{item:F5b} is not satisfied by every multi-resource theory.
For example, this property is satisfied by the multi-resource theory of thermodynamics, but it is violated
by other theories, like entanglement theory, where the set of free states is composed of separable states.
We are currently working to weaken this property, following the ideas presented in
Ref.~\cite{brandao_generalization_2010}, by requiring the invariant sets to be closed under permutations
of copies. This less demanding property should allow us to use the approximate de Finetti's
theorems~\cite{renner_symmetry_2007}, and to obtain similar conditions to those obtained with \ref{item:F5b}.
To study the interconversion of entanglement with some other resource, however, one can think of restricting
the state space of the theory in a way in which the resulting subset of separable states satisfies property~\ref{item:F5b},
see the example in Sec.~\ref{control_theory_ex}. Finally, it is worth noting that all the results we obtain in this
section also apply if one of the monotones, or both, is of the form shown in Eq.~\eqref{montone_average}.
Indeed, these monotones satisfy the same properties of the relative entropy distances, with the difference that
the corresponding invariant set can be composed by multiple states, and these states do not need to have full rank.
\par
Let us now consider an example of resource interconversion which will highlight the properties that we are
searching for in a bank system. Suppose we have a certain amount of euros and pounds in our wallet, and we
want to convert one into the other, for example, from pounds to euros. In order to convert these two currencies
we need to go to the bank, that we would expect to satisfy the following properties. The first property could
be referred to as \emph{passivity} of the bank, and it is represented by the fact that if we do not provide some
pounds, we cannot receive any euros (and vice versa). Second is the existence of an exchange rate, that is, the
bank will convert the two currency at a certain exchange rate, and this rate can be different depending on the
bank we use. The last property concerns the catalytic nature of the bank, since we would like a bank not to
change the exchange rate between pounds and euros as a consequence of our transaction (this last property
is approximately satisfied by real banks, at least for the amount exchanged by average costumers).
\par
The previous example shows that, in order to achieve resource interconversion, we need to introduce in our
framework an additional system, the bank, with some specific properties. Within our formalism, we consider
the same multi-partite system introduced in Sec.~\ref{quant_res}, with the main system $S$, and two
batteries $B_1$ and $B_2$. The system $S$ is now used as a bank, which has to satisfy the three essential
properties (passivity, existence of a rate, catalytic behaviour) that we have informally described in the
previous paragraph, and that we are going to formalise in the following. First of all, we need the states describing the bank to be \emph{passive}, meaning
that we should not be able to extract from this system both resources at the same time, since we always need
to pay one resource to gain another one.
Thus, the set of \emph{bank states} is defined as
\begin{definition} \label{def:bank_state}
Consider a multi-resource theory $\rt_{\text{multi}}$ satisfying the asymptotic equivalence property
with respect to the monotones $E_{\f_1}$ and $E_{\f_2}$. The \emph{set of bank states} of the theory
is a subset of the state space $\SH$ defined as,
\begin{align}
\label{set_f3}
\f_{\text{bank}} = \big\{ \rho \in \SH \ | \ \forall \, \sigma \in \SH , \
&E_{\f_1}(\sigma) > E_{\f_1}(\rho) \ \text{or} \nonumber \\
&E_{\f_2}(\sigma) > E_{\f_2}(\rho) \ \text{or} \nonumber \\
&E_{\f_1}(\sigma) = E_{\f_1}(\rho) \, \text{and} \, E_{\f_2}(\sigma) = E_{\f_2}(\rho) \, \big\}.
\end{align}
Within the set $\f_{\text{bank}}$ we can find different subsets of bank states with a fixed value of
$E_{\f_1}$ and $E_{\f_2}$. We define each of these subsets as
\begin{equation}
\label{subset_bank}
\f_{\text{bank}}\left(\bar{E}_{\f_1},\bar{E}_{\f_2}\right) =
\left\{ \rho \in \f_{\text{bank}} \ | \ E_{\f_1}(\rho) =
\bar{E}_{\f_1} \, \text{and} \, \ E_{\f_2}(\rho) = \bar{E}_{\f_2} \right\}.
\end{equation}
\end{definition}
Notice that Eq.~\eqref{set_f3} implies that no state can be found with smaller values of both
monotones $E_{\f_i}$'s. In this way, the agent is not able to transform the state of the bank in a
way in which both resources are extracted from it and stored in the batteries. Instead, they always
need to trade resources. The set of bank states $\f_{\mathrm{bank}}$ can be visualised in
the resource diagram of the theory, see Fig.~\ref{fig:bank_subset}. This set is represented by a curve
on the boundary of the state space, connecting the points associated with $\f_1$ to those associated
with $\f_2$. In appendix~\ref{convex_bound} we show that, under the current assumptions, this
curve is always convex, and in the following we focus our attention to those segments where the curve
is strictly convex.
\par
The subsets $\f_{\mathrm{bank}} \left(\bar{E}_{\f_1},\bar{E}_{\f_2}\right)$'s represent individual
points in the resource diagram describing the multi-resource theory, and they obey many of the properties
satisfied by the invariant sets $\f_i$'s. Indeed, one can show that
\begin{itemize}
\item For all $n \in \N$, we have that each subset of bank states is such that
\begin{equation}
\label{eq:add_bank}
\f^{(n)}_{\mathrm{bank}} \left(\bar{E}_{\f_1},\bar{E}_{\f_2}\right)
=
\f^{\otimes n}_{\mathrm{bank}} \left(\bar{E}_{\f_1},\bar{E}_{\f_2} \right),
\end{equation}
that is, these subsets satisfy property~\ref{item:F5b}. This equality is proved in
Prop.~\ref{additive_f3} of appendix~\ref{additional}.
\item Every subset $\f_{\mathrm{bank}}\left(\bar{E}_{\f_1},\bar{E}_{\f_2}\right)$ is convex,
property~\ref{item:F2}, as shown in Prop.~\ref{convex_f3} in appendix~\ref{additional}.
\item Every subset $\f_{\mathrm{bank}}\left(\bar{E}_{\f_1},\bar{E}_{\f_2}\right)$, and its
extensions to the many-copy case, is invariant under the class of allowed operations
$\A_{\mathrm{multi}}$ of the multi-resource theory, as shown in Lem.~\ref{lem:inv_f3} in
appendix~\ref{additional}.
\end{itemize}
\begin{figure}[t!]
\center
\includegraphics[width=0.5\textwidth]{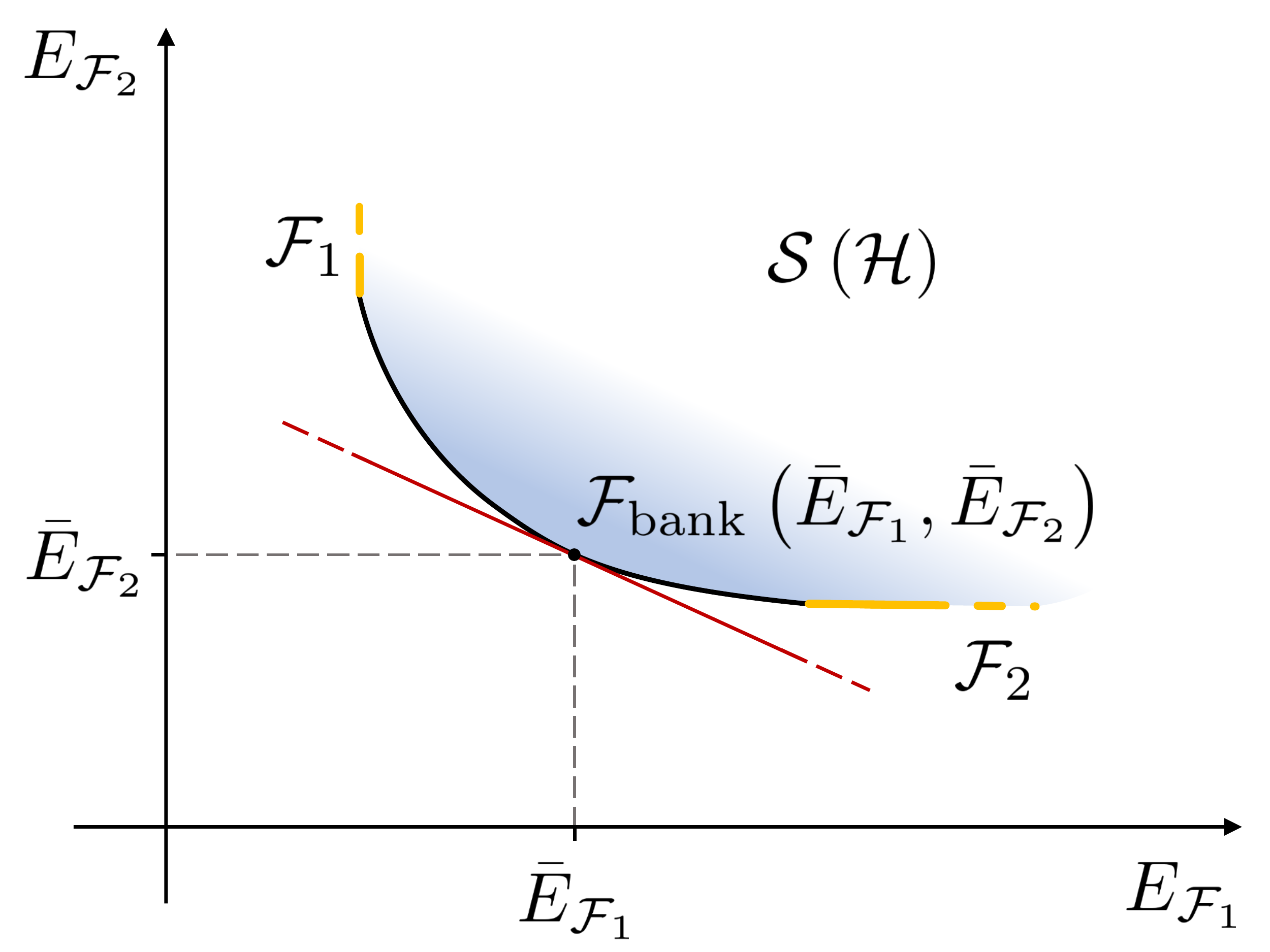}
\caption{The set of bank states introduced in Eq.~\eqref{set_f3} is represented in the
$E_{\f_1}$--$E_{\f_2}$ diagram. Only part of the state-space $\SH$ is shown, highlighted
by the blue gradient region, together with the invariant sets of the theory $\f_1$ and $\f_2$, the
two yellow segments. The black curve connecting these segments is the set of all the bank states
of the theory $\f_{\text{bank}}$. It is worth noting that this set does not include any states
contained in the invariant sets $\f_1$ and $\f_2$. Indeed, a bank state needs to have a non-zero value
of both resources in order to allow for general resource interconversion, and the states in $\f_1$ and
$\f_2$ do not contain any amount of their associated resource. A specific subset of bank states,
labelled by $\f_{\text{bank}}\left(\bar{E}_{\f_1},\bar{E}_{\f_2}\right)$, is shown on the
curve, see Eq.~\eqref{subset_bank}. Notice that, graphically, a bank state is one for
which there exists no other state in the region immediately below and left. The red line,
which is tangent to the set of bank states and passes through the point
$\f_{\text{bank}}\left(\bar{E}_{\f_1},\bar{E}_{\f_2}\right)$, is parametrised by
$f_{\text{bank}}^{\bar{E}_{\f_1},\bar{E}_{\f_2}} = 0$, see Eq.~\eqref{f3_monotone}.}
\label{fig:bank_subset}
\end{figure}
\par
The second essential property for a bank
is that the exchange rate needs only to depend on which state of the bank we choose to use. In our
framework, it is the choice of the values $\bar{E}_{\f_1}$ and $\bar{E}_{\f_2}$, defining the subset
$\f_{\text{bank}} \left(\bar{E}_{\f_1},\bar{E}_{\f_2}\right)$, that determines the exchange rate at which the
resources are converted. In order to obtain this exchange rate we introduce the following function, which
quantifies how much the properties of the bank change during a transformation, and generalises the
Helmholtz free energy used in thermodynamics. Given the subset of bank states $\f_{\text{bank}}
\left(\bar{E}_{\f_1},\bar{E}_{\f_2}\right)$, this function is defined as
\begin{equation}
\label{f3_monotone}
f_{\text{bank}}^{\bar{E}_{\f_1},\bar{E}_{\f_2}}(\rho) := \alpha \, E_{\f_1}(\rho) + \beta \, E_{\f_2}(\rho) - \gamma,
\end{equation}
where $\alpha$, $\beta$, and $\gamma$ are non-negative constant factors, which depend on the
subset of bank states we have chosen. In order
to define the linear coefficients, we impose the following two properties for this function,
\begin{description}
\item[B1\label{item:B1}] The function $f_{\text{bank}}^{\bar{E}_{\f_1},\bar{E}_{\f_2}}$ is equal to
zero over the subset $\f_{\text{bank}}\left(\bar{E}_{\f_1},\bar{E}_{\f_2}\right)$.
\item[B2\label{item:B2}] The value of this function on the states contained in the subset
$\f_{\text{bank}}\left(\bar{E}_{\f_1}, \bar{E}_{\f_2}\right)$ is minimum.
\end{description}
Notice that property~\ref{item:B1} is there to set the zero of the function, and implies that
\begin{equation}
\gamma = \alpha \, \bar{E}_{\f_1} + \beta \, \bar{E}_{\f_2}.
\end{equation}
Property~\ref{item:B2}, instead, fixes the ratio between
the constants $\alpha$ and $\beta$. This condition can be visualised in the resource diagram, and is
equivalent to the request that, in such a diagram, the bank monotone is tangent to the state space,
so that
\begin{equation}
\label{eq:tangent_f3}
f_{\mathrm{bank}}^{\bar{E}_{\f_1},\bar{E}_{\f_2}}(\rho) \geq f_{\mathrm{bank}}^{\bar{E}_{\f_1},\bar{E}_{\f_2}}(\sigma)
, \quad
\forall \, \rho \in \SH , \, \forall \, \sigma \in \f_{\mathrm{bank}}\left(\bar{E}_{\f_1},\bar{E}_{\f_2}\right).
\end{equation}
The above property is always satisfied under our working assumptions, since the curve
of bank states is convex, see Fig.~\ref{fig:bank_subset}. We refer to this function
as the \emph{bank monotone}.
\par
The bank monotone can be easily extended to the state space of $n$ copies of the system. The main
difference is that, when we consider states in $\SHn{n}$, the coefficient $\gamma$ is proportional to
the number of copies $n$, and we write $\gamma = n \left( \alpha \, \bar{E}_{\f_1} + \beta \, \bar{E}_{\f_2}
\right)$. This follows from property~\ref{item:B1}, together with the fact that the subset $\f_{\mathrm{bank}}
\left(\bar{E}_{\f_1},\bar{E}_{\f_2}\right)$ satisfies property~\ref{item:F5b}, see Eq.~\eqref{eq:add_bank}.
Since the function in Eq.~\eqref{f3_monotone} is a linear combination of the monotones
$E_{\f_1}$ and $E_{\f_2}$, it is easy to show (see also appendix~\ref{additional}) that it
satisfies the properties listed in the following proposition
\begin{restatable}{prop}{bankproperties}
\label{prop:properties_bank_mon}
Consider a resource theory $\rt_{\text{multi}}$ with allowed operations $\A_{\text{multi}}$, satisfying
asymptotic equivalence with respect to the monotones $E_{\f_1}$ and $E_{\f_2}$, i.e.~the relative
entropy distances from the invariant sets of the theory. Suppose that these sets satisfy the
properties~\ref{item:F1}, \ref{item:F2}, \ref{item:F3}, and \ref{item:F5b}. Then, the function
$f_{\text{bank}}^{\bar{E}_{\f_1},\bar{E}_{\f_2}}$ introduced in Eq.~\eqref{f3_monotone}
satisfies the following properties.
\begin{description}
\item[B3\label{item:B3}] The function $f_{\text{bank}}^{\bar{E}_{\f_1},\bar{E}_{\f_2}}$ is additive.
\item[B4\label{item:B4}] The function $f_{\text{bank}}^{\bar{E}_{\f_1},\bar{E}_{\f_2}}$ is monotonic
under partial tracing.
\item[B5\label{item:B5}] The function $f_{\text{bank}}^{\bar{E}_{\f_1},\bar{E}_{\f_2}}$ is sub-extensive,
i.e., this function scales at most linearly in the number of systems considered. More precisely, for any sequence of states
$\left\{ \rho_n \in \SHn{n} \right\}$, we have that $f_{\text{bank}}^{\bar{E}_{\f_1},\bar{E}_{\f_2}}(\rho_n) = O(n)$.
\item[B6\label{item:B6}] The function $f_{\text{bank}}^{\bar{E}_{\f_1},\bar{E}_{\f_2}}$ is asymptotic continuous.
\item[B7\label{item:B7}] The function $f_{\text{bank}}^{\bar{E}_{\f_1},\bar{E}_{\f_2}}$ is monotonic under the set
of allowed operations $\A_{\text{multi}}$, since $\alpha$ and $\beta$ are non-negative.
\end{description}
\end{restatable}
\par
The third and last property we demand from a bank concerns the back-reaction it experiences
during interconversion of resources. We want that, after the transformation, the state of the
bank only changes infinitesimally with respect to the bank monotone associated with it. If this
is the case, we can show that the exchange rate only changes infinitesimally, and therefore
we can keep using the bank to inter-convert between resources at the same exchange rate.
More concretely, we now consider a tripartite system composed by a bank $S$ and and two
batteries, $B_1$ and $B_2$. Each of these subsystems is composed by many copies of the
same fundamental system described by $\hil$, for which we defined the notion of bank states.
Thus, the bank $S$ is described by $\hil_S = \hil^{\otimes n}$, with $n \in \N$, and its initial state
is given by $n$ copies of the bank state $\rho \in \f_{\mathrm{bank}}\left(\bar{E}_{\f_1},
\bar{E}_{\f_2}\right)$. The batteries are described by $\hil_{B_i} = \hil^{\otimes m_i}$,
$m_i \in \N$, where $i = 1,2$. The states describing the batteries are $\omega_1 \in \SHbi$,
and $\omega_2 \in \SHbii$, respectively.
\par
A \emph{resource interconversion} is an asymptotically reversible transformation
\begin{equation}
\label{interconv_trasf}
\rho^{\otimes n} \otimes \omega_1 \otimes \omega_2
\xleftrightarrow{\text{asympt}}
\tilde{\rho}^{\otimes n} \otimes \omega'_1 \otimes \omega'_2,
\end{equation}
where $\tilde{\rho} \in \SH$, $\omega'_1 \in \SHbi$, and $\omega'_2 \in \SHbii$,
satisfying the following property, see also Fig.~\ref{fig:tangent_monotone},
\begin{description}
\item[X1\label{item:X1}] The state of the bank changes infinitesimally during the resource interconversion.\\
If $\rho \in \f_{\text{bank}}\left(\bar{E}_{\f_1},\bar{E}_{\f_2}\right) \subset \SH$, then
the state $\tilde{\rho} \in \SH$ is such that
\begin{equation}
\label{condition_x1}
f_{\text{bank}}^{\bar{E}_{\f_1},\bar{E}_{\f_2}}(\tilde{\rho}^{\otimes n})
=
f_{\text{bank}}^{\bar{E}_{\f_1},\bar{E}_{\f_2}}(\rho^{\otimes n}) + \delta_n,
\end{equation}
where $\delta_n > 0$ is such that $\delta_n \rightarrow 0$ as $n \rightarrow \infty$.
\end{description}
It is worth noting that, according to the above definition, the bank is here acting as a \emph{catalyst},
allowing for resource interconversion. Catalysts are used in resource theories to allow for state
transformations which are otherwise impossible~\cite{gour_resource_2015,chitambar_quantum_2018}.
These systems are generally described by resourceful states, and therefore are subject to strict
constraints, for example the requirement that their initial state needs to be perfectly (or approximately)
re-obtained at the end of the transformation~\cite{ng_limits_2015}. These constraints are required since,
otherwise, one might act on the catalyst and extract resources from it, thus trivializing the theory~\cite{van_dam_universal_2003}.
It is interesting to notice that our bank is similarly constrained, specifically by Eq.~\ref{condition_x1}. As
we see in the following theorem, this constrain is enough to allow for resource interconversion, but also to
ensure a non-trivial behaviour of the theory (no resource is extracted for free).
\par
We are now ready to introduce the interconversion relation which links the different amounts of
resources exchanged, weighted by the exchange rate given by the bank. The theorem is proved
in appendix~\ref{main_results}.
\begin{restatable}{thm}{interconvert}
\label{thm:interconvert_relation}
Consider a resource theory $\rt_{\text{multi}}$ with two resources, equipped with the batteries
$B_1$ and $B_2$. Suppose the theory satisfies asymptotic equivalence with respect to the
monotones $E_{\f_1}$ and $E_{\f_2}$, i.e.~the relative entropy distances from the invariant sets of the
theory, and that these sets satisfy the properties~\ref{item:F1}, \ref{item:F2}, \ref{item:F3}, and \ref{item:F5b}.
Then, the resource interconversion of Eq.~\eqref{interconv_trasf}, where the bank has to
transform in accord to condition~\ref{item:X1}, is solely regulated by the following relation,
\begin{equation}
\label{finite_interconversion}
\alpha \, \Delta W_1 = - \beta \, \Delta W_2  + \delta_n.
\end{equation}
Furthermore, when the number of copies of the bank system $n$ is sent to infinity, we have that the
above equation reduces to the following one, which we refer to as the \emph{interconversion relation},
\begin{equation}
\label{interconversion}
\Delta W_1 = - \frac{\beta}{\alpha} \, \Delta W_2,
\end{equation}
where the amount of resources exchanged $\Delta W_i$ is non-zero.
\end{restatable}
Let us highlight some properties that
a bank state needs to satisfy in order to allow for interconversion of resources from one battery to another,
and vice versa. We show that to interconvert between the resources in both directions we need a bank state containing
a non-zero amount of both resources. First notice that, since both parameters $\alpha$ and $\beta$ are
non-negative, whenever we exchange between resources, we increase the amount
contained in one of the batteries (for example, $\Delta W_1 > 0$) while decreasing the amount
contained in the other ($\Delta W_2 < 0$). However, the change in these two resources also
depends on the transformation of the bank state, see Eq.~\eqref{resource_work_i}. Therefore,
one has to consider the bank state used for interconversion, and the amount of resources
contained in it. When the bank state $\rho$ is such that $E_{\f_1}(\rho) > 0$ and $E_{\f_2}(\rho) > 0$,
then interconversion can be achieved (in both directions) between $\Delta W_1$ and $\Delta W_2$,
at the rate specified by Eq.~\eqref{interconversion}. Moreover, as far as the amount of resources in
the bank is non-zero, we can exchange any amount of one resource for the other (since we can
take the number of copies of the bank to be infinite). This is the case of thermodynamics, where
thermal states indeed contain a positive amount of both energy and entropy, the two resources of
the theory, and Eq.~\eqref{interconversion} gives the conversion rate for Landauer's erasure.
\begin{figure}[t!]
\center
\includegraphics[width=0.5\textwidth]{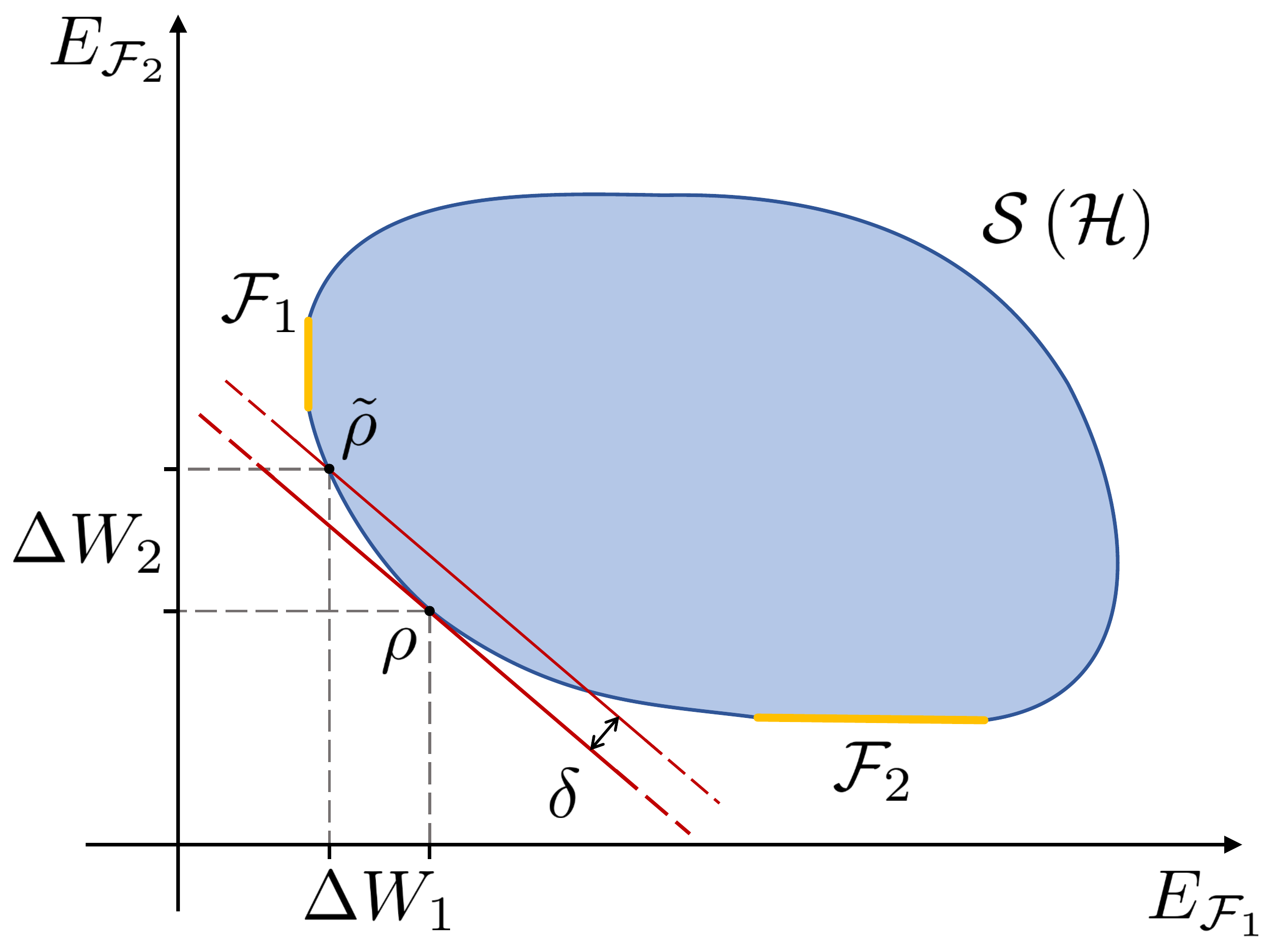}
\caption{The state-space of the theory $\rt_{\text{multi}}$ is represented in the $E_{\f_1}$--$E_{\f_2}$
diagram. The invariant sets of the theory, $\f_1$ and $\f_2$, are represented by the two yellow
segments. The set of bank states $\f_{\text{bank}}$ lies on the boundary of the state-space,
and is represented by the curve connecting the two invariant sets, see appendix~\ref{convex_bound}.
The subset of bank states $\f_{\text{bank}}\left(\bar{E}_{\f_1},\bar{E}_{\f_2}\right)$, where $\rho$ is contained,
is represented by a point in the diagram. The red line which is tangent to the state-space and passes by
the point associated to $\rho$ represents the set of states with value of the monotone $f_{\text{bank}}
^{\bar{E}_{\f_1},\bar{E}_{\f_2}}$ equal to $0$. The other line is given by all those states with a value
$\delta > 0$ of this monotone. We see that, by mapping $\rho$ into $\tilde{\rho}$, we can extract an
amount $\Delta W_1$ of the first resource, while paying an amount $\Delta W_2$ of the second resource.
Furthermore, one can show that when $\delta \rightarrow 0$, these two quantities tend to $0$ as
$\delta^{\frac{1}{2}}$, i.e., with a slower rate. It is then possible to keep the $\Delta W_i$'s finite if we
take $n \propto \delta^{-1}$ copies of the bank states, see the proof of Thm.~\ref{thm:interconvert_relation},
in appendix~\ref{main_results}. Thus, in the limit $n \rightarrow \infty$, the overall back-action on bank
states associated with the conversion of resources can be made arbitrarily small.}
\label{fig:tangent_monotone}
\end{figure}
\par
Finally, let us consider what would happen if we were to allow the states in $\f_1$ or $\f_2$ (or in their
intersection) to describe the bank. If the bank state were such that $E_{\f_1}(\rho) > 0$ and $E_{\f_2}(\rho) = 0$
(or vice versa), then we could only exchange in one direction, since we could gain the first resource while paying
the second resource (or vice versa). If the bank state did not contain any amount of resources, $E_{\f_1}(\rho) = 0$
and $E_{\f_2}(\rho) = 0$, then we could not perform interconversion at all, because we would have to reduce the
amount of one of them within the bank. However, this would not be possible since the amount of resource stored
in a (bank) state cannot be negative. As a result, the multi-resource theories in which an interesting
interconversion relation can be found are the ones in which the invariant sets of the theory do not intercept,
see the right panel of Fig.~\ref{fig:invariant_sets_structure}.
\subsection{Bank monotones and the relative entropy distance}
\label{bank_monotone}
We start this section with an example concerning different models to describe thermodynamics,
and the connection between these models. In the last part of Sec.~\ref{multi_resource}, we have
introduced a multi-resource theory whose resources are energy and entropy (or, information). For this theory,
the bank states are thermal states at a given temperature $T$. We can move from this description
of thermodynamics to a different one, based on a single-resource theory, by enlarging the class of
operations in such a way that the agent can freely add ancillary systems in a thermal state with
temperature $T$. This corresponds to the physical situation in which the system is put in contact
with an infinite thermal bath. The single-resource theory we obtain is analogous to the one of
Thermal Operations~\cite{brandao_resource_2013, horodecki_fundamental_2013}, and its
resource quantifier is unique. In fact, we can show that the bank monotone of the multi-resource
theory and the resource quantifier of the single resource theory both coincides (modulo a multiplicative
factor) with $F - F_{\beta}$, where $F$ is the Helmholtz free energy of the state whose resource
we are quantifying, and $F_{\beta}$ is the Helmholtz free energy of the thermal state with temperature
$T = \beta^{-1}$.
\par
In the following we study the connection between a general multi-resource theory and the single-resource
theory obtained by enlarging the allowed operations with the possibility of adding ancillary systems
described by bank states in $\f_{\text{bank}}\left(\bar{E}_{\f_1},\bar{E}_{\f_2}\right)$. We find that
the bank monotone of Eq.~\eqref{f3_monotone}, $f_{\text{bank}} ^{\bar{E}_{\f_1},\bar{E}_{\f_2}}$,
coincides with the unique measure of resource for the obtained single-resource theory. As a result,
we find that property~\ref{item:X1}, which regulates the exchange
of resources in the multi-resource theory, can be understood as the condition that the resource
characterising the bank does not increase during the transformation. Furthermore, we show that,
when the subset of bank states $\f_{\text{bank}}\left(\bar{E}_{\f_1},\bar{E}_{\f_2}\right)$ contains
a full-rank state, the monotone $f_{\text{bank}}^{\bar{E}_{\f_1}, \bar{E}_{\f_2}}$ is proportional to
the relative entropy distance from this subset. Let us now introduce the single-resource theory
which can be derived from $\rt_{\text{multi}}$ by allowing the possibility of adding ancillary systems
described by specific bank states.
\begin{definition}
\label{def:sing_res_constr}
Consider the two-resource theory $\rt_{\text{multi}}$ with allowed operations $\A_{\text{multi}}$
and invariant sets $\f_1$ and $\f_2$ which satisfy the properties~\ref{item:F1}, \ref{item:F2}, \ref{item:F3},
and \ref{item:F5b}. Consider the bank set $\f_{\text{bank}}\left(\bar{E}_{\f_1},\bar{E}_{\f_2}\right)$
introduced in Eq.~\eqref{subset_bank}. We define the single-resource theory $\rt_{\text{single}}$
as that theory whose class of allowed operations $\A_{\text{single}}$ is composed by the following three
fundamental operations,
\begin{enumerate}
\item Add an ancillary system described by $n \in \N$ copies of a bank state
$\rho_P \in \f_{\text{bank}}\left(\bar{E}_{\f_1},\bar{E}_{\f_2}\right)$.
\item Apply any operation $\chn \in \A_{\text{multi}}$ to system and ancilla.
\item Trace out the ancillary systems.
\end{enumerate}
The most general operation in $\A_{\text{single}}$ which does not change
the number of systems between its input and output is
\begin{equation}
\label{sin_res_map}
\chn^{\text{(s)}}(\rho) = \Tr{P^{(n)}}{\chn \left( \rho \otimes \rho_P^{\otimes n} \right)},
\end{equation}
where we are partial tracing over the degrees of freedom $P^{(n)}$, that is, over the ancillary
system initially in $\rho_P^{\otimes n}$.
\end{definition}
The bank monotone associated with the bank set $\f_{\mathrm{bank}}\left(
\bar{E}_{\f_1}, \bar{E}_{\f_2}\right)$, see Eq.~\eqref{f3_monotone}, is the unique quantifier
for the single-resource theory $\rt_{\mathrm{single}}$. In order to show the uniqueness of this
monotone, we first have to show that the single-resource theory satisfies asymptotic equivalence.
\begin{restatable}{thm}{singres}
\label{unique_f3}
Consider the two-resource theory $\rt_{\text{multi}}$ with allowed operations $\A_{\text{multi}}$, and invariant
sets $\f_1$ and $\f_2$ which satisfy the properties~\ref{item:F1}, \ref{item:F2}, \ref{item:F3}, and \ref{item:F5b}.
Suppose the theory satisfies the asymptotic equivalence property with respect to the monotones $E_{\f_1}$ and
$E_{\f_2}$. Then, given the subset of bank states $\f_{\text{bank}}\left(\bar{E}_{\f_1}, \bar{E}_{\f_2}\right)$, the
single-resource theory $\rt_{\text{single}}$ with allowed operations $\A_{\text{single}}$ satisfies the
asymptotic equivalence property with respect to $f_{\text{bank}}^{\bar{E}_{\f_1},\bar{E}_{\f_2}}$.
\end{restatable}
The proof of this theorem can be found in appendix~\ref{main_results}, and we provide a
geometric sketch of it in Fig.~\ref{fig:geometric_proof}. As a side remark, notice that the functions
$E_{\f_1}$ and $E_{\f_2}$ are not monotonic under the set of allowed operations $\A_{\text{single}}$.
This follows from the fact that we can now replace any state of the system with a state in
$\f_{\text{bank}}\left(\bar{E}_{\f_1}, \bar{E}_{\f_2}\right)$, since we are free to add an ancillary
system in such state, and to perform a swap between main system and ancilla (since this operation
belongs to $\A_{\text{multi}}$). Then, if the bank state contains a non-zero amount of resources,
meaning that $\bar{E}_{\f_i} > 0$ for $i=1,2$, we can always find a state in $\SH$ with lower
value of either $E_{\f_1}$ or $E_{\f_2}$ (but not both at the same time), and therefore the above
transformation would increase the value of this monotone.
\begin{figure}[t!]
\center
\includegraphics[width=0.45\textwidth]{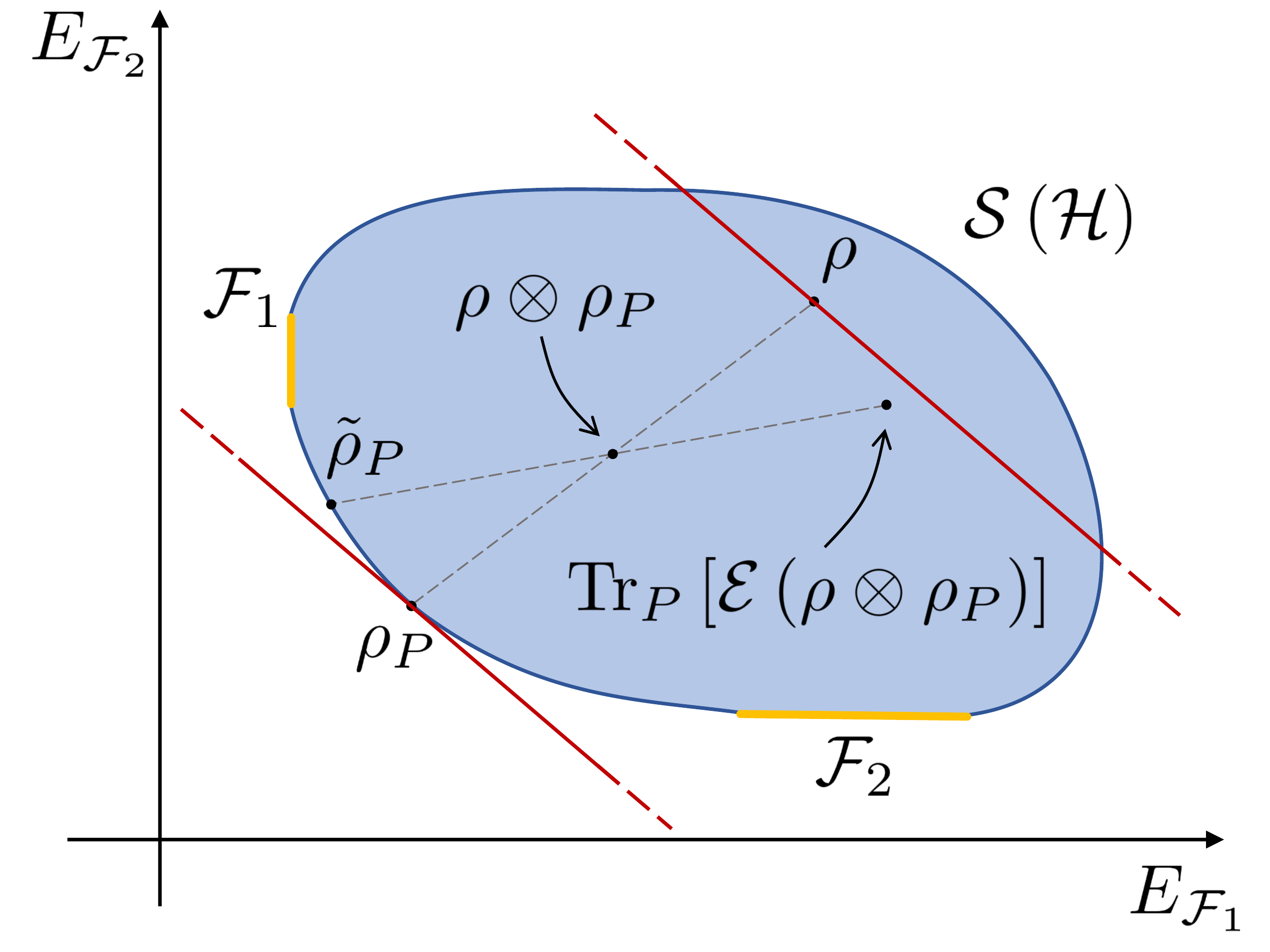}
\includegraphics[width=0.45\textwidth]{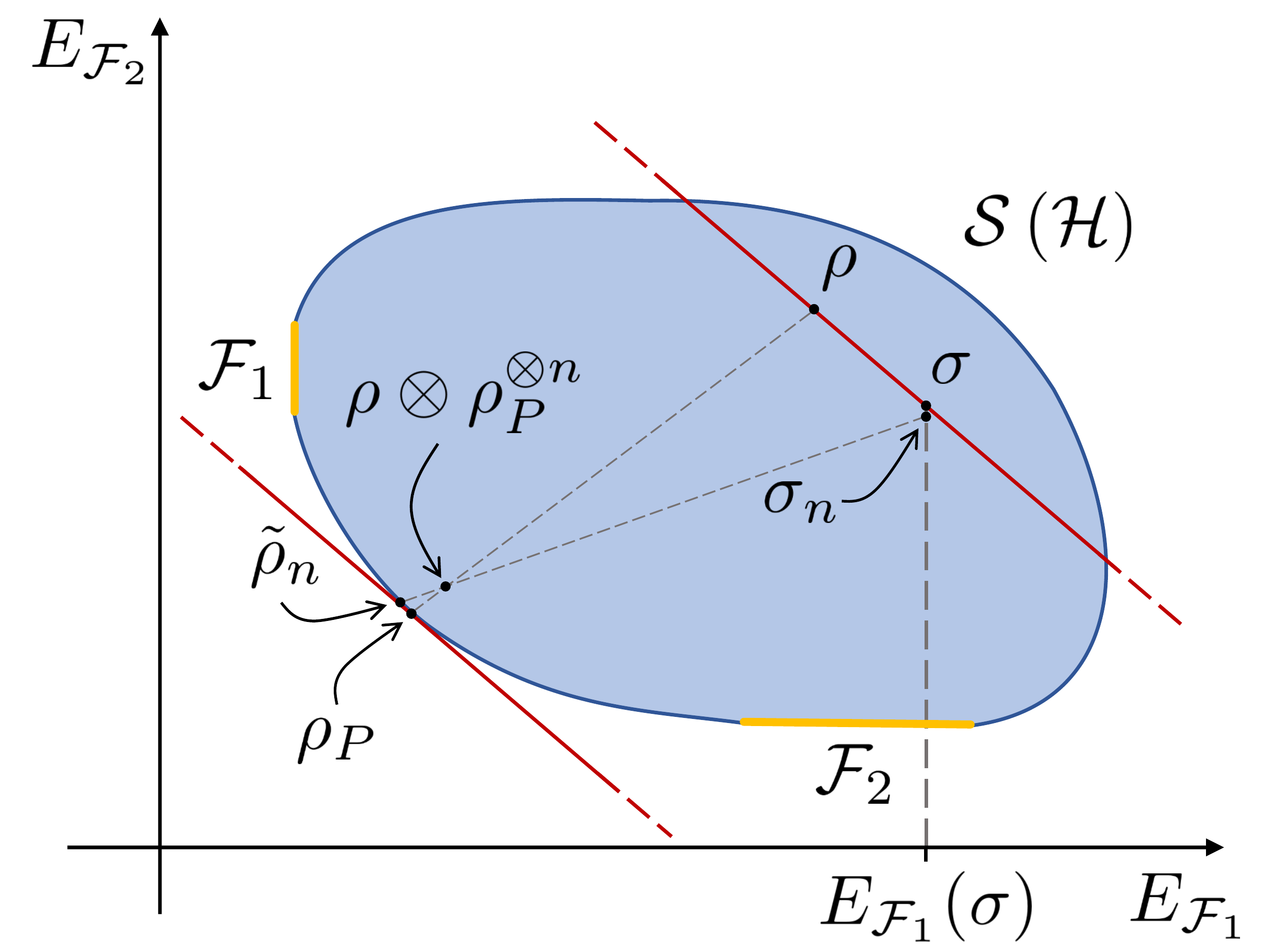}
\caption{We sketch a geometric proof of Thm.~\ref{unique_f3} using the $E_{\f_1}$--$E_{\f_2}$ diagram.
The blue region is the state-space $\SH$, the yellow segments are the invariant sets $\f_1$
and $\f_2$, and the red lines highlight the states with same value of monotone
$f_{\text{bank}}^{\bar{E}_{\f_1}, \bar{E}_{\f_2}}$. Notice that in this figure we are using the fact
that, when $\f_i$ satisfies property~\ref{item:F5b}, the monotone $E_{\f_i}$ is such that $E_{\f_i}(\rho \otimes
\sigma) = E_{\f_i}(\rho) + E_{\f_i}(\sigma)$ for any two states $\rho$ and $\sigma$ in $\SH$, see
Lem.~\ref{f_i_inequality}. To represent the state $\rho \otimes \sigma$ in the diagram, we renormalise its
values of the $E_{\f_i}$'s by dividing them by the number of copies considered, in this case by two.
{\bf Left.} We first sketch why the function $f_{\text{bank}}^{\bar{E}_{\f_1}, \bar{E}_{\f_2}}$ is
monotonic under the set of allowed operations $\A_{\text{single}}$. Consider a system initially described
by the state $\rho$, and add to it an ancillary system described by the bank state $\rho_P \in
\f_{\text{bank}}\left(\bar{E}_{\f_1},\bar{E}_{\f_2}\right)$. The global system is then represented by
a point in the middle of the segment connecting $\rho$ and $\rho_P$. We can transform the global
state with the help of a sub-linear ancilla and of the operation $\chn \in \A_{\text{multi}}$,
mapping it into the state $\sigma \otimes \tilde{\rho}_P$ with same value of $E_{\f_1}$ and $E_{\f_2}$.
If we take $\tilde{\rho}_P$ to be on the boundary of the state-space, we can easily see that $\sigma
\equiv \Tr{P}{\chn (\rho \otimes \rho_P)}$ always lies below the red line passing through $\rho$,
i.e., its value of $f_{\text{bank}}^{\bar{E}_{\f_1}, \bar{E}_{\f_2}}$ is smaller than the one for $\rho$.
{\bf Right.} We now sketch how to map between states with the same value of the monotone
$f_{\text{bank}}^{\bar{E}_{\f_1}, \bar{E}_{\f_2}}$, using the set of operations $\A_{\text{single}}$.
In this case, we compose the main system, initially described by $\rho$, with an additional one
described by $n$ copies of $\rho_P$. We then use an operation $\chn \in \A_{\text{multi}}$, together
with a sub-linear ancilla, and we ask the final state of the system, $\sigma_n = \Tr{P^{(n)}}{\chn(\rho
\otimes \rho_P^{\otimes n})}$ to have the same value of $E_{\f_1}$ of the target state $\sigma$. It is
then easy to show that, as $n \rightarrow \infty$, the state $\sigma_n$ tends to $\sigma$, while the
$n$ copies of the final state of the ancilla, $\tilde{\rho}_n$, tends to the bank state $\rho_P$.}
\label{fig:geometric_proof}
\end{figure}
\par
From the above theorem it follows an interesting link between the bank monotone $f_{\text{bank}}
^{\bar{E}_{\f_1}, \bar{E}_{\f_2}}$, defined in Eq.~\eqref{f3_monotone}, and the relative entropy distance
from the set of bank states $\f_{\text{bank}}\left(\bar{E}_{\f_1},\bar{E}_{\f_2}\right)$. Indeed, when this
set of states contains at least one full-rank state, we can prove that these two functions have to coincide,
modulo a multiplicative factor. This is a consequence of the fact that $\rt_{\text{single}}$ satisfies
asymptotic equivalence, which implies the uniqueness of the resource measure, and of the fact that
both the bank monotone and the relative entropy distance from the bank set satisfy the same properties,
in particular monotonicity under the operations in $\A_{\text{single}}$ and asymptotic continuity.
We can express this fact in the following corollary, whose proof can be found in appendix~\ref{main_results}.
\begin{restatable}{coro}{bankmonrelent}
\label{bank_equal_rel_ent}
Consider the two-resource theory $\rt_{\text{multi}}$ with allowed operations $\A_{\text{multi}}$, and invariant
sets $\f_1$ and $\f_2$ which satisfy the properties~\ref{item:F1}, \ref{item:F2}, \ref{item:F3}, and \ref{item:F5b}.
Suppose the theory satisfies the asymptotic equivalence property with respect to the monotones $E_{\f_1}$ and
$E_{\f_2}$. If the subset of bank states $\f_{\text{bank}}\left(\bar{E}_{\f_1}, \bar{E}_{\f_2}\right)$ contains a
full-rank state, then the bank monotone $f_{\text{bank}}^{\bar{E}_{\f_1},\bar{E}_{\f_2}}$ coincides with the
relative entropy distance from this subset of states, modulo a multiplicative constant.
\end{restatable}
\par
We close the section with the remark that, in the currently known scenarios, the bank subsets
always contain at least a full-rank state, and in fact we find that, for these theories, the above
correspondence between the bank monotone of Eq.~\eqref{f3_monotone} and the relative entropy
distance is satisfied. An example is the multi-resource theory of thermodynamics, in which the
relative entropy distance from a thermal state at a given temperature is indeed equal to the linear
combination of the average energy and the entropy of a system. Other examples can be found in
Sec.~\ref{examples}.
\subsection{First law for multi-resource theories}
\label{sec:first_law}
We can now introduce a general first law for multi-resource theories with disjoint invariant sets,
see the right panel of Fig.~\ref{fig:invariant_sets_structure}.
In order for this law to be valid, we need access to the batteries, the bank, and the main system.
Within this setting, the first law consists of a single relation which links the different amount of
resources exchanged with the batteries, the $\Delta W_i$'s, with the change in bank monotone
over the state of the main system. The idea is that, contrary to what seen in Sec.~\ref{quant_res},
a state transformation over the main system is possible, when a bank is present, if this single relation
is satisfied. Indeed, we do not need to use a fixed amount of each resource, since they are
inter-convertible using the bank system.
\par
In more detail, we consider a theory $\rt_{\text{multi}}$ that, for simplicity, has just two resources. The invariant
sets are $\f_1$ and $\f_2$, they satisfy the properties~\ref{item:F1}, \ref{item:F2}, \ref{item:F3} and~\ref{item:F5b},
and the theory satisfies the asymptotic equivalence property with respect to the monotones $E_{\f_1}$
and $E_{\f_2}$. The global system is divided into four partitions, the main system $S$, the bank
$P$, and the batteries $B_1$ and $B_2$. We assume the bank to be initially described by a state
$\rho_P \in \f_{\text{bank}}\left(\bar{E}_{\f_1}, \bar{E}_{\f_2}\right)$, where this subset contains at
least one full-rank state. The relevant monotone for the interconversion of resources is then the
relative entropy distance from the subset $\f_{\text{bank}}\left(\bar{E}_{\f_1}, \bar{E}_{\f_2}\right)$,
as shown in Cor.~\ref{bank_equal_rel_ent}.
\par
Suppose that the main system is initially described by the state $\rho \in \SHs$, and we want to
map it into the state $\sigma \in \SHs$, with possibly a different value of $E_{\f_1}$ and  $E_{\f_2}$.
If we do not have access to the bank, then the amount of resources we need to exchange is given
by the difference of the monotones $E_{\f_i}$'s between the initial and final state of the main system,
see Eq.~\eqref{resource_work_ent} in Sec.~\ref{multi_rev_unique}. But since we have
access to the battery, we can exchange between the resources, and we are not obliged any more
to provide a fixed amount for each resource. In order to show this, consider the global initial state
$\rho \otimes \rho_P \otimes \omega_1 \otimes \omega_2$, describing the main system, the bank,
and the two batteries $B_1$ and $B_2$. Then, we (asymptotically) map this global state, using the
allowed operations $\A_{\text{multi}}$, into the final state $\sigma \otimes \tilde{\rho}_P \otimes
\omega'_1 \otimes \omega'_2$, where the final state of the bank is $\tilde{\rho}_P$, and the
batteries $B_1$ and $B_2$ have final state $\omega'_1$ and $\omega'_2$, respectively. Due to
asymptotic equivalence, this state transformation is possible only if the monotones $E_{\f_i}$'s
are preserved. However, the final state of the bank only has to satisfy property~\ref{item:X1}, and we
have shown in Sec.~\ref{bank_interconvert} that such constraint still allows us to exchange an
arbitrary amount of resources, see Thm.~\ref{thm:interconvert_relation}. As a result, there is
a single relation that regulates the state transformation over the main system,
\begin{coro}
\label{coro:first_law}
Consider the two-resource theory $\rt_{\text{multi}}$ with allowed operations $\A_{\text{multi}}$,
and invariant sets $\f_1$ and $\f_2$ which satisfy the properties~\ref{item:F1}, \ref{item:F2}, \ref{item:F3},
and \ref{item:F5b}. Suppose the theory satisfies the asymptotic equivalence property with
respect to the monotones $E_{\f_1}$ and $E_{\f_2}$, and that the global system is divided
into a main system $S$, a bank described by the set of states $\f_{\text{bank}}\left(\bar{E}_{\f_1},
\bar{E}_{\f_2}\right)$ (which contains at least one full-rank state), and two batteries $B_1$ and $B_2$.
Then, a transformation which maps the state of the main system from $\rho$ into $\sigma$, where
these states are completely general, only has to satisfy the following relation
\begin{equation}
\label{eq:first_law}
\alpha \, \Delta W_1 + \beta \, \Delta W_2
=
E_{\f_{\text{bank}}\left(\bar{E}_{\f_1}, \bar{E}_{\f_2}\right)}(\rho)
-
E_{\f_{\text{bank}}\left(\bar{E}_{\f_1}, \bar{E}_{\f_2}\right)}(\sigma),
\end{equation}
where each $\Delta W_i$ is defined as the difference in the monotone $E_{\f_i}$ over the final and
initial state of the battery $B_i$, see Eq.~\eqref{work_Ri}, and $E_{\f_{\text{bank}}\left(\bar{E}_{\f_1},
\bar{E}_{\f_2}\right)}$ is the relative entropy distance from the set of states describing the bank.
\end{coro}
\par
We refer to Eq.~\eqref{eq:first_law} as the first law of multi-resource theories. Indeed, for the resource
theory of thermodynamics, where energy and entropy are the two resources, and the bank is given by an
infinite thermal reservoir with a given temperature $T$, this equation corresponds to the First Law of
Thermodynamics. In fact, in the thermodynamic scenario we have that $\Delta W_1 = - \Delta U$, where
$U$ is the internal energy of the system, while $\Delta W_2 = \Delta S$ is the change in entropy in the
system. The change in relative entropy distance on the main system is proportional to the change in
Helmholtz free-energy, which in turn is equal to the work extracted from the system, $W$. The linear
coefficients in the equation can be computed from Eq.~\eqref{f3_monotone}, knowing that the bank
monotone is equal to the relative entropy distance from the thermal state with temperature $T$. It is
easy to show that $\alpha = T^{-1}$ and $\beta =1$ . If we re-arrange the equation, and we define
$Q = T \, \Delta S$ as the amount of heat absorbed by the system, we obtain $\Delta U = Q - W$, that
is, the First Law of Thermodynamics.
\section{Examples}
\label{examples}
In this section we present two examples of multi-resource theories where an interconversion
relation can be derived. The first one is thermodynamics for multiple conserved quantities (even
non-commuting ones), while the second one concerns local control under energetic restrictions.
In both examples we describe the state-space (and we represent it with a resource diagram), we
find the bank states of the theory, and we derive an interconversion relation for the different resources.
Furthermore, in both cases we find that the bank monotone is proportional to the relative entropy
distance from the given set of bank states, as expected from Cor.~\ref{bank_equal_rel_ent}.
\begin{figure}[t!]
    \centering

\begin{tikzpicture}
	\node[box] (initial) {\small {\bf I}. Does the theory satisfy the asymptotic equivalence property of Def.~\ref{def:asympt_equivalence_multi}?};
	\node[initbox, above = 1cm of initial] (preamble) {{\bf Instructions}
	\begin{enumerate}
	\itemsep-.3em
	\item Identify constraints and conservation laws of the theory.
	\item Define allowed operations $\A_{\text{multi}}$ and identify invariant sets $\left\{\f_i\right\}$.
	\item Use flowchart to identify  the properties of monotones and invariant sets.
	\end{enumerate}};
	\node[box, below left = 2cm of initial] (entropy) {\small {\bf II}. Is the theory asymptotically equivalent with respect to the relative entropy distances from the invariant sets, or equivalently to the monotones introduced in Sec.~\ref{average_non_increasing}?};
	\node[box, below right = 2cm of initial] (noasymteq) {\small {\bf III}. Modify the class of allowed operations $\A_{\text{multi}}$, if possible.};
	\node[box, below = 1.414cm of initial] (propF123) {\small {\bf IV}. Are the invariant sets closed (\ref{item:F1}) and convex (\ref{item:F2}),
	and do they contain a full-rank state (\ref{item:F3})$^{\text{a}}$? Can we find batteries where resources are stored individually (\ref{item:M1}) with respect to these monotones?};
	\node[box, below = 1.414cm  of entropy] (trivial) {\small {\bf V}. Can we find candidate battery subsystems
	such that each resource is stored (\ref{item:M1}) and accessed (\ref{item:M2}) individually, and the
	resource quantifiers are non-negative (\ref{item:M3})?};
	\node[box, below = .6cm  of propF123] (asymtcont) {\small {\bf VII}. Are the resource measures monotonic under partial trace
	(\ref{item:M4}), sub-additive (\ref{item:M5}), sub-extensive (\ref{item:M6}), and asymptotic continuous (\ref{item:M7})?};
	\node[box, below = 1cm  of noasymteq] (propF4orF5) {\small {\bf VI}. Which additional properties do the invariant sets satisfy?
	Are they closed under tensor product (\ref{item:F4}) and partial trace (\ref{item:F5}), or are they composed by a single state
	(\ref{item:F5b})?};
	\node[finbox, below = 1cm  of trivial] (batteries) {\small The resource theory does not allow for the use of batteries, see Sec.~\ref{quant_res}.};
	\node[finbox, below = 1.414cm  of propF4orF5] (unique) {\small The state-space of the theory can be represented in a resource diagram, and the representation is unique, see Thm.~\ref{thm:reversible_multi}.};
	\node[finbox, below = 1cm  of asymtcont] (resdiag) {\small The state-space of the theory can be represented in a resource diagram, see Fig.~\ref{fig:resource_diagram}, but the representation is not unique.};
	\node[box, below = 1cm  of unique] (intercept) {\small {\bf VIII}. The state-space of the theory can be represented in a resource diagram, and the representation is unique. Are the invariant sets disjoint, see right panel of Fig~\ref{fig:invariant_sets_structure}?};
	\node[longbox, below = 1cm of batteries.east,xshift=0.8cm] (nointerconv) {\small Resources cannot be exchanged, as the bank does not contain one or more of them.};
	\node[longbox, below = .5cm of nointerconv] (interconv) {\small It is possible to exchange one resource for another, see Thm.~\ref{thm:interconvert_relation}, and the theory admits a first law, see Cor.~\ref{coro:first_law}};
	\draw [arrow] (initial) -- node[left = .5cm] {\small YES} (entropy);
	\draw [arrow] (initial) -- node[right = .5cm] {\small NO} (noasymteq);
	\draw [arrow] (noasymteq) |- (initial);
	\draw [arrow] (entropy) -- node[below = .3cm] {\small YES} (propF123);
	\draw [arrow] (entropy) -- node[left = .2cm] {\small NO} (trivial);
	\draw [arrow] (trivial) -- node[left = .2cm] {\small NO} (batteries);
	\draw [arrow] (trivial) -- node[below = .3cm] {\small YES} (asymtcont);
	\draw [arrow] (asymtcont) -- node[below = .2cm] {\small YES} (unique);
	\draw [arrow] (asymtcont) -- node[left = .2cm] {\small NO} (resdiag);
	\draw [arrow] (propF123) -- node[above = .2cm] {\small YES} (propF4orF5);
	\draw [arrow] (propF123) -- node[above = .2cm] {NO} (noasymteq);
	\draw [arrow] (propF4orF5) -- node[right = .2cm] {\small \ref{item:F4} $\land$ \ref{item:F5}} (unique);
	\draw [arrow] (propF4orF5) |- ([xshift=0.5cm]propF4orF5.east) node[right = .2cm,yshift=-1.37cm] {\small \ref{item:F5b}} |- (intercept);
	\draw [arrow] (intercept) -- node[above = .2cm] {\small NO} (nointerconv);
	\draw [arrow] (intercept) -- node[below = .2cm] {\small YES} (interconv);
	\node[inner sep=0pt,font=\footnotesize] at ([xshift=-3.7cm,yshift=-0.5cm]current bounding box.south) (ftnta) {$^{\text{a}}$For the monotones introduced in Sec.~\ref{average_non_increasing}, property~\ref{item:F3} is not relevant.};
\end{tikzpicture}
    \caption{How to apply the results of this paper to an arbitrary resource theory.}
    \label{fig:flowchart}
\end{figure}
\par
Before we introduce the examples, we provide a flowchart~\ref{fig:flowchart} that should help the
reader in building a multi-resource theory. In particular, the flowchart clarifies in which situations
each of the results we obtain hold for a specific theory. This tool should be used as follows,
\begin{itemize}
\item The fundamental constraints and conservation laws of the task under consideration should be
identified, and together with them the resources composing the theory.
\item Given the set of resources for the theory, we define the class of allowed operations $\A_{\text{multi}}$
as in Eq.~\eqref{all_ops_multi}, and we identify the invariant sets of the theory $\left\{ \f_i \right\}_{i=1}^m$.
\item Checking whether asymptotic equivalence holds for the multi-resource theory is the first step
of the flowchart (box {\bf I} in Fig.~\ref{fig:flowchart}). To show that the theory satisfies this property, we need to find a
protocol which maps between states with same values of a given set of monotones.
\item If the theory satisfies asymptotic equivalence, we can focus on the properties of the monotones and
of the invariant sets. Following the flowchart, we can then easily identify which properties and features hold for
the theory under consideration.
\end{itemize}
The flowchart here introduced is used in the first example to clarify how to characterise a multi-resource theory.
\subsection{Thermodynamics of multiple-conserved quantities}
\label{thermo_example}
In this example we consider the resource theory of thermodynamics in the presence of
multiple conserved quantities (even in the case in which these quantities do not 
commute)~\cite{guryanova_thermodynamics_2016,yunger_halpern_microcanonical_2016,
lostaglio_thermodynamic_2017}. Our system is a $d$-level quantum system, and for simplicity,
we only consider two conserved quantities $A$ and $B$. The allowed operations are Thermal
Operations~\cite{brandao_resource_2013, horodecki_fundamental_2013}, composed by unitary
operators which commute with both $A$ and $B$. This set of maps can be obtained as a proper
subset of the intersection between the allowed operations of the following single-resource theories,
\begin{itemize}
\item The resource theory of the quantity $A$. The allowed operations are all the average-$A$-non-increasing
maps, whose invariant set is composed by a single state, $\f_A = \left\{ \ket{a_0}\bra{a_0} \right\}$, the
eigenstate of $A$ associated with its minimum eigenvalue $a_0$ (for simplicity, we here assume it to be
non-degenerate). From Sec.~\ref{average_non_increasing} it follows that this theory has a monotone of the
form $M_A(\rho) = \tr{A \rho} - a_0$.
\item The resource theory of the quantity $B$. The allowed operations are all the average-$B$-non-increasing
maps, whose invariant set is composed by a single state, $\f_B = \left\{ \ket{b_0}\bra{b_0} \right\}$, the
eigenstate of $B$ associated with its minimum eigenvalue $b_0$ (for simplicity, we here assume it to be
non-degenerate). From Sec.~\ref{average_non_increasing} it follows that this theory has a monotone of the
form $M_B(\rho) = \tr{B \rho} - b_0$.
\item The resource theory of purity, where the allowed operations are all the maps whose fix point is the
maximally-mixed state $\f_S = \left\{ \frac{\Id}{d} \right\}$ (unital maps). One monotone of the theory is
the relative entropy distance from $\frac{\Id}{d}$, that is, $E_{\f_S}(\rho) = \log d - S(\rho)$ where
$S(\cdot)$ is the von Neumann entropy. 
\end{itemize}
\par
We can now make use of the flowchart to characterise the multi-resource theory. Box {\bf I} in the
flowchart asks whether or not the considered multi-resource theory satisfies asymptotic equivalence. In
Refs.~\cite{bera_thermodynamics_2017,sparaciari_resource_2016} it has been shown that, indeed, a
resource theory of this kind does satisfy the asymptotic equivalence property of Def.~\ref{def:asympt_equivalence_multi}
with respect to the monotones $M_A$, $M_B$ and $E_{\f_S}$. Furthermore, it is easy to see that these
monotones are either relative entropy distances from the set of invariant states, or that they are of the form
given in Eq.~\eqref{montone_average}. This implies that we can answer positively to box {\bf II} in
the flowchart.
\par
We now need to consider the properties of the invariant sets of the theory, which in turn determine the
properties of the monotones. It is easy to show that these sets are closed (property~\ref{item:F1}) and
convex (property~\ref{item:F2}). Furthermore, $\f_S$ contains a full-rank state (property~\ref{item:F3}),
that implies asymptotic continuity of the associated monotone, see Refs.~\cite{synak-radtke_asymptotic_2006,
brandao_generalization_2010}. The fact that the other sets do not contain a full-rank state is not problematic
since we are considering monotones of the form of Eq.~\eqref{montone_average}, that are nevertheless
asymptotic continuous, see Prop.~\ref{average_asymp_cont}. Thus, the invariant sets satisfy all the
properties required in box {\bf IV}, and we now need to construct batteries able to store the different resources
separately (property~\ref{item:M1}).
\par
For the first kind of resource, this can be achieved by selecting two pure states with different
average values of $A$, and same average values of $B$. The battery $B_A$, storing the first kind of resource,
is then composed by a certain number of copies of these two states, where the number varies when we
extract/store the resource. A similar construction can be done for the other battery $B_B$. For the purity battery,
we can take a system with degenerate $A$ and $B$, and take states with a certain number of copies of a pure
state and mixed state. If this construction is possible, then we can answer positively to box {\bf IV} in the
flowchart.
\par
We can now study the properties of the invariant sets, specifically their closure with respect to tensor
product (property~\ref{item:F4}) and partial trace (property~\ref{item:F5}). Since each invariant set is composed
by a single state, we find that both these properties and property~\ref{item:F5b} are satisfied. Thus, from box
{\bf VI} we can  move to box {\bf VIII}, and therefore the theory can be studied with a resource diagram, see
Fig.~\ref{fig:example_monotones_thermo} and the representation is unique.
\par
Let us now consider a reversible transformation, described by the following equation
\begin{equation}
\rho^{\otimes n} \otimes \omega_A \otimes \omega_B \otimes \omega_S
\xleftrightarrow{\text{asympt}}
\sigma^{\otimes n} \otimes \omega'_A \otimes \omega'_B \otimes \omega'_S,
\end{equation}
where the $n$ copies of $\rho$ and $\rho'$ describe the main system at the beginning and the end of
the transformation, and the states $\omega_i$ and $\omega_i'$ are the initial and final states of
the battery $B_i$, for $i = A,B,S$. According to asymptotic equivalence, the transformation is possible if
\begin{align}
\Delta W_A &= M_A^{\infty}(\rho) - M_A^{\infty}(\sigma) = \tr{A \left( \rho - \sigma \right)}, \\
\Delta W_B &= M_B^{\infty}(\rho) - M_B^{\infty}(\sigma) = \tr{B \left( \rho - \sigma \right)}, \\
\Delta W_S &= E_{\f_S}^{\infty}(\rho) - E_{\f_S}^{\infty}(\sigma) = S(\sigma) - S(\rho).
\end{align}
\par
To answer the last box of the flowchart, box {\bf VIII}, we need to focus on the resources contained in
the bank states. Indeed, in order to get an interconversion relation and a first law we need the bank states to
contain a non-zero amount of each resource. This has to be the case for the current resource theory, since the
invariant sets do not intercept each other. Therefore, this theory admits a first law, as we are going to show.
It can be easily shown, using Jaynes principle~\cite{jaynes_information_1957}, that the bank states are
of the following form
\begin{equation}
\label{GGS}
\tau_{\beta_1,\beta_2} = \frac{e^{- \beta_1 A - \beta_2 B}}{Z},
\end{equation}
where the parameters $\beta_1, \beta_2 \in [0, \infty)$, and $Z = \tr{e^{- \beta_1 A - \beta_2 B}}$
is the partition function of the system. These states are known in thermodynamics as the
grand-canonical ensemble. Each $\tau_{\beta_1,\beta_2}$ is a bank state with a different value
of resource $A$, resource $B$, and purity. The value of these three resources only depends
on the parameters $\beta_1$ and $\beta_2$. In order to find the interconversion relation
we need to construct the bank monotone
\begin{equation}
f_{\text{bank}}^{\bar{\beta}_1, \bar{\beta}_2}(\rho) =
\alpha_{\bar{\beta}_1, \bar{\beta}_2} M_A(\rho) +
\gamma_{\bar{\beta}_1, \bar{\beta}_2} M_B(\rho) +
\delta_{\bar{\beta}_1, \bar{\beta}_2} E_{\f_S}(\rho) -
\xi_{\bar{\beta}_1, \bar{\beta}_2}
\end{equation}
which is equal to zero over the bank state $\tau_{\bar{\beta}_1, \bar{\beta}_2}$. Properties~\ref{item:B1}
and~\ref{item:B2} provide a geometrical way of building the monotone. If we represent the state space
in a three-dimensional diagram (where the axes are given by $M_A$, $M_B$, and $E_{\f_S}$), then
the hyperplane defined by the equation $f_{\text{bank}}^{\bar{\beta}_1, \bar{\beta}_2} = 0$ is tangent
to the state space and only intercepts it in $\tau_{\bar{\beta}_1, \bar{\beta}_2}$, see
Fig.~\ref{fig:example_monotones_thermo} for an example.
\begin{figure}[t!]
\center
\includegraphics[width=0.5\textwidth]{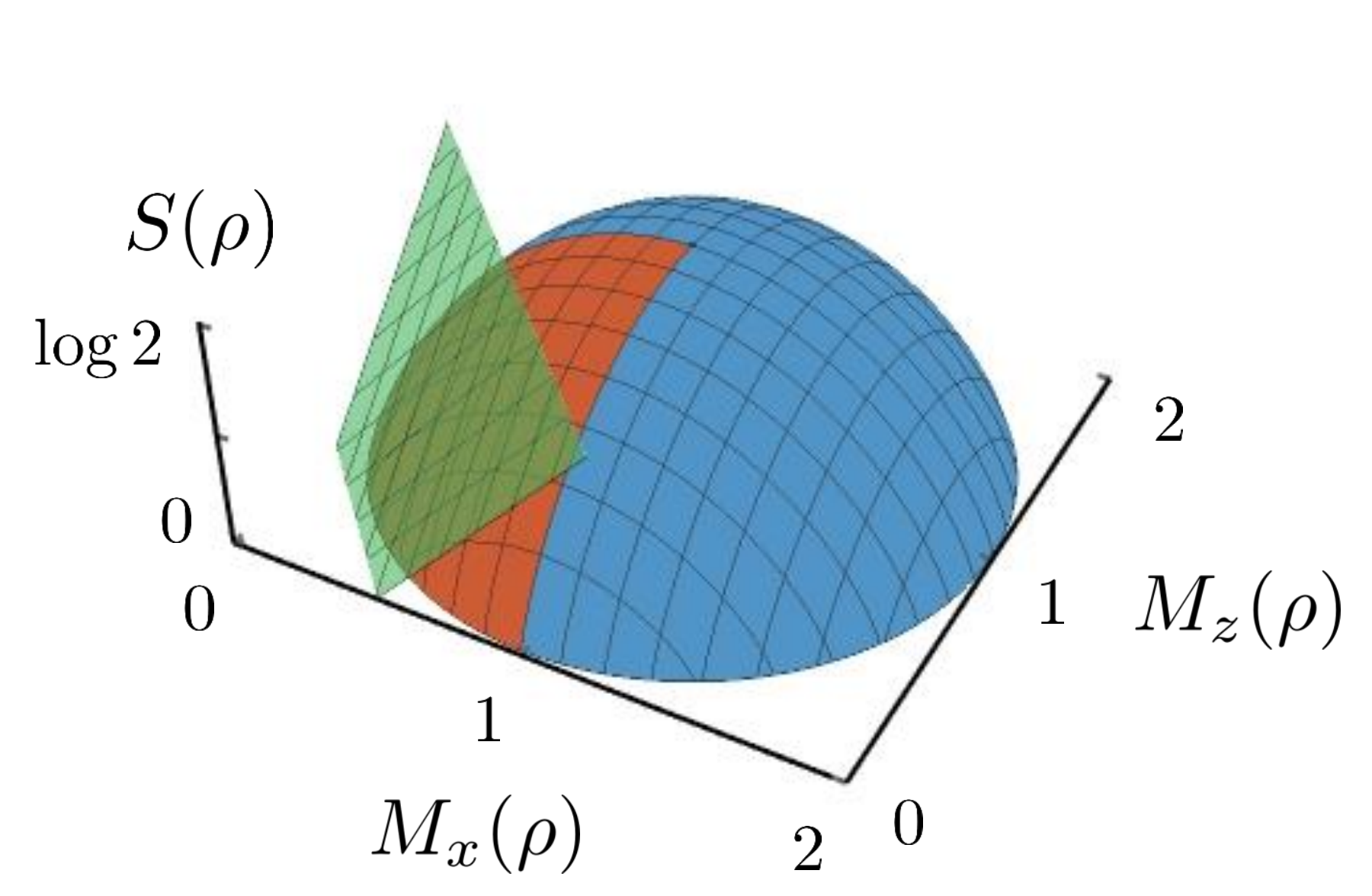}
\caption{The state space of the multi-resource theory of thermodynamics and conserved
angular momenta (along the $x$ and $z$ axes). On the surface we find the states
$\tau_{\beta_1,\beta_2}$ defined in Eq.~\eqref{GGS}, where $\beta_1$ and $\beta_2$
take values in $\R$. The red surface is the set of bank states, with $\beta_1$ and $\beta_2$
are both non-negative, and the green plane is tangent to the state space in the point
associated with $\tau_{\bar{\beta}_1, \bar{\beta}_2}$. The equation of the plane gives
the monotone $f_{\text{bank}}^{\bar{\beta}_1, \bar{\beta}_2}$.}
\label{fig:example_monotones_thermo}
\end{figure}
\par
The hyperplane defined by $f_{\text{bank}}^{\bar{\beta}_1, \bar{\beta}_2} = 0$ is identified by the
normal vector
\begin{equation}
\hat{n} = \hat{r}_1 \times \hat{r}_2 , \quad \text{where} \
\hat{r}_i =
\left(
\frac{\partial M_A(\tau_{\bar{\beta}_1, \bar{\beta}_2})}{\partial \beta_i} ;
\frac{\partial M_B(\tau_{\bar{\beta}_1, \bar{\beta}_2})}{\partial \beta_i} ;
\frac{\partial E_{\f_S}(\tau_{\bar{\beta}_1, \bar{\beta}_2})}{\partial \beta_i}
\right)^T \
\text{for} \ i = 1, 2.
\end{equation}
The parametric equation of the hyperplane then gives us the expression of the monotone,
\begin{equation}
f_{\text{bank}}^{\bar{\beta}_1, \bar{\beta}_2}(\rho) =
n_1 \left( M_A(\rho) - M_A(\tau_{\bar{\beta}_1, \bar{\beta}_2}) \right) +
n_2 \left( M_B(\rho) - M_B(\tau_{\bar{\beta}_1, \bar{\beta}_2}) \right) +
n_3 \left( E_{\f_S}(\rho) - E_{\f_S}(\tau_{\bar{\beta}_1, \bar{\beta}_2}) \right),
\end{equation}
where $n_i$ is the $i$-th component of the normal vector $\hat{n}$. By evaluating
the monotones $M_A$, $M_B$, $E_{\f_S}$, and their derivatives we find that
$f_{\text{bank}}^{\bar{\beta}_1, \bar{\beta}_2}$ is equal (modulo a positive multiplicative factor
depending on the parameters $\bar{\beta}_1$ and $\bar{\beta}_2$)
to the relative entropy distance from $\tau_{\bar{\beta}_1, \bar{\beta}_2}$,
\begin{equation}
f_{\text{bank}}^{\bar{\beta}_1, \bar{\beta}_2}(\rho) \propto E_{\tau_{\bar{\beta}_1, \bar{\beta}_2}}(\rho)
= \bar{\beta}_1 \, \tr{\rho A} + \bar{\beta}_2 \, \tr{\rho B} - S(\rho) + \log Z.
\end{equation}
Thus, the bank state $\tau_{\bar{\beta}_1, \bar{\beta}_2}$ allows us to obtain the following
interconversion relation between the three resources,
\begin{equation}
\bar{\beta}_1 \, \Delta W_A + \bar{\beta}_2 \, \Delta W_B = \Delta W_S,
\end{equation}
while the state of the bank only changes by an infinitesimal amount in terms of
$E_{\tau_{\bar{\beta}_1, \bar{\beta}_2}}$.
\subsection{Local control theory under energetic restrictions}
\label{control_theory_ex}
We now introduce a multi-resource theory describing local control under energetic
restrictions. Specifically, we consider the situation in which a quantum system
is divided into two well-defined partitions $A$ and $B$, and we can only act on the
individual partitions with non-entangling operations, which furthermore need to not
increase the energy of the overall system. This kind of simultaneous restrictions on
locality and thermodynamics has also been considered in other previous works, see for example
Refs.~\cite{hovhannisyan_entanglement_2013,huber_thermodynamic_2015,wilming_second_2016,
beny_energy_2017,lekscha_quantum_2018}. The multi-resource theory is obtained by considering
two single-resource theories, the one of entanglement and the one of energy. While this is
a well-defined multi-resource theory, it is not straightforward to prove that it is also a
reversible theory. Therefore, to provide a first law in this setting, we have to restrict
the state-space to a subset of all bipartite density operators.
\subsubsection{Set-up}
\label{setup_en_ent}
Let us consider a bipartite system, whose partitions are labelled as $A$ and $B$, with a non-local
Hamiltonian $H_{AB}$ (that is, the two partitions interact with each other, and the ground state of the
system is an entangled state). The set of allowed operations of this multi-resource theory is
obtained from the intersection of the allowed operations of the following single-resource theories,
\begin{itemize}
\item The resource theory of energy. The allowed operations are all the average-energy-non-increasing
maps, defined in Sec.~\ref{average_non_increasing}. When the Hamiltonian has non-degenerate
ground state $\ket{\text{g}}$, the fixed state of the maps is $\f_H = \ket{\text{g}}\bra{\text{g}}$.
The monotone of this resource theory is $M_H(\rho) = \tr{H \rho} - E_{\text{g}}$, where
$E_{\text{g}}$ is the eigenvalue associated with the ground state $\ket{\text{g}}$.
\item The resource theory of entanglement. The allowed operations are the asymptotically non-entangling
maps~\cite{brandao_reversible_2010}. These maps are relevant to us for two reasons. Firstly, all our results
hold in the asymptotic limit, and therefore it is reasonable to consider the set of maps which do not create
entanglement in this limit. Secondly, this is the only set of operations which provides a reversible theory for
entanglement. The monotone is $E_{\f_{\text{sep}}}(\cdot)$, where $\f_{\text{sep}}$ is the set of
separable states, invariant under the class of operations.
\end{itemize}
While the current multi-resource theory is well-defined and meaningful, it is not straightforward to prove whether
it is reversible in the sense given in Def.~\ref{def:asympt_equivalence_multi}. Furthermore, it is known that
the relative entropy of entanglement, $E_{\f_{\text{sep}}}$, is not additive (or even extensive) for all bipartite
density operator. Therefore, if we want to study interconversion of resources in this setting, we need to consider
a subset of the state-space (as well as of the invariant set $\f_{\text{sep}}$).
\par
In the following we will focus on the simplest example of a multi-resource theory of this kind. The bipartite system
is composed by two qubits, so that its Hilbert space is $\hil_{AB} = \C^2 \otimes \C^2$. The Hamiltonian
of the system is
\begin{equation}
\label{non_loc_ham}
H_{AB} = E_0 \ket{\Psi_\text{singlet}}\bra{\Psi_\text{singlet}} + E_1 \, \Pi_{\text{triplet}},
\end{equation}
where $E_0 < E_1$, the ground state is the singlet state,
\begin{equation}
\ket{\Psi_\text{singlet}} = \frac{1}{\sqrt{2}} \left( \ket{01} - \ket{10} \right),
\end{equation}
and $\Pi_{\text{triplet}} = \sum_{i=1}^3 \ket{\Psi_\text{triplet}^{(i)}}\bra{\Psi_\text{triplet}^{(i)}}$
is the projector on the triplet subspace, where
\begin{align}
\ket{\Psi_\text{triplet}^{(1)}} &= \frac{1}{\sqrt{2}} \left( \ket{01} + \ket{10} \right), \\
\ket{\Psi_\text{triplet}^{(2)}} &= \frac{1}{\sqrt{2}} \left( \ket{00} - \ket{11} \right), \\
\ket{\Psi_\text{triplet}^{(3)}} &= \frac{1}{\sqrt{2}} \left( \ket{00} + \ket{11} \right).
\end{align}
In order to get a reversible multi-resource theory, and therefore to be able to define
the interconversion relations, we consider a restricted state-space, given by the
following subset of bipartite density operators,
\begin{equation}
\s_1 = \left\{ \rho \in \s(\hil_{AB}) \, | \, \rho = \p_0 \ket{\Psi_\text{singlet}}\bra{\Psi_\text{singlet}}
+ \sum_{i=1}^{3} \p_i \ket{\Psi_\text{triplet}^{(i)}}\bra{\Psi_\text{triplet}^{(i)}}
\ , \ \text{with} \ \p_0 \geq \frac{1}{2} \right\}.
\end{equation}
There are two additional reasons why we are interested in this set of states. First of all, because the
relative entropy of entanglement $E_{\f_{\text{sep}}}$ has an analytical expression for states
which are diagonal in the Bell basis~\cite{vedral_quantifying_1997, audenaert_asymptotic_2001,
miranowicz_closed_2008} (that here coincides with the energy eigenbasis). Secondly, because it is
easy to show, see Eq.~\eqref{set_f3}, that $\s_1$ contains the bank states of the theory, that are
the interesting ones when it comes to study interconversion. Finally, it is worth noting that the
state-space $\s_1$ contains all the Gibbs states of the non-local Hamiltonian $H_{AB}$ with positive
temperatures. Within this restricted state-space we find the following subset of separable states, 
\begin{equation}
\label{css_invariant}
\f_{\text{css}} = \left\{ \rho = \frac{1}{2} \ket{\Psi_\text{singlet}}\bra{\Psi_\text{singlet}}
+ \sum_{i=1}^{3} \p_i \ket{\Psi_\text{triplet}^{(i)}}\bra{\Psi_\text{triplet}^{(i)}} \right\}.
\end{equation}
It is worth noticing that the above subset $\f_{\text{css}}$ contains all the closest-separable states
to the entangled states in our restricted state-space $\s_1$ (see Ref.~\cite{miranowicz_closed_2008}).
As a result, for any state $\rho \in \s_1$ we have that
\begin{equation}
E_{\f_{\text{sep}}}(\rho) = E_{\f_{\text{css}}}(\rho) =
1 - \h\left( \bra{\Psi_\text{singlet}}\rho\ket{\Psi_\text{singlet}} \right),
\end{equation}
where $\h(\cdot)$ is the binary entropy function. Since our focus is restricted to the sole states in
the subset $\s_1$, we will now re-define\footnote{The modified set of allowed operations makes it easier
for us to find a protocol for inter-converting resources. However, we do not exclude the possibility of being
able to perform interconversion with the original set of allowed operations, that preserve all separable
states. However, finding this protocol might be non-trivial, and could be material of future work.}
the set of allowed operations of the multi-resource theory as those energy-non-increasing maps which
only preserve the subset of separable states $\f_{\text{css}} = \f_{\text{sep}} \cap \s_1$. We can
define this class of operation as
\begin{equation}
\label{ent_en_all_ops}
\A_{\text{multi}} =
\left\{
\chn :  \s(\hil_{AB}) \rightarrow \s(\hil_{AB})
\ | \
\chn(\f_{\text{css}}) \subseteq \f_{\text{css}}
\ \text{and} \
\tr{ \chn(\rho) H_{AB}} \leq \tr{ \rho \, H_{AB}} \ \forall \, \rho \in \s(\hil_{AB})
\right\},
\end{equation}
where each $\chn \in \A_{\text{multi}}$ is a completely positive and trace preserving map.
\begin{figure}[t!]
\center
\includegraphics[width=0.35\textwidth]{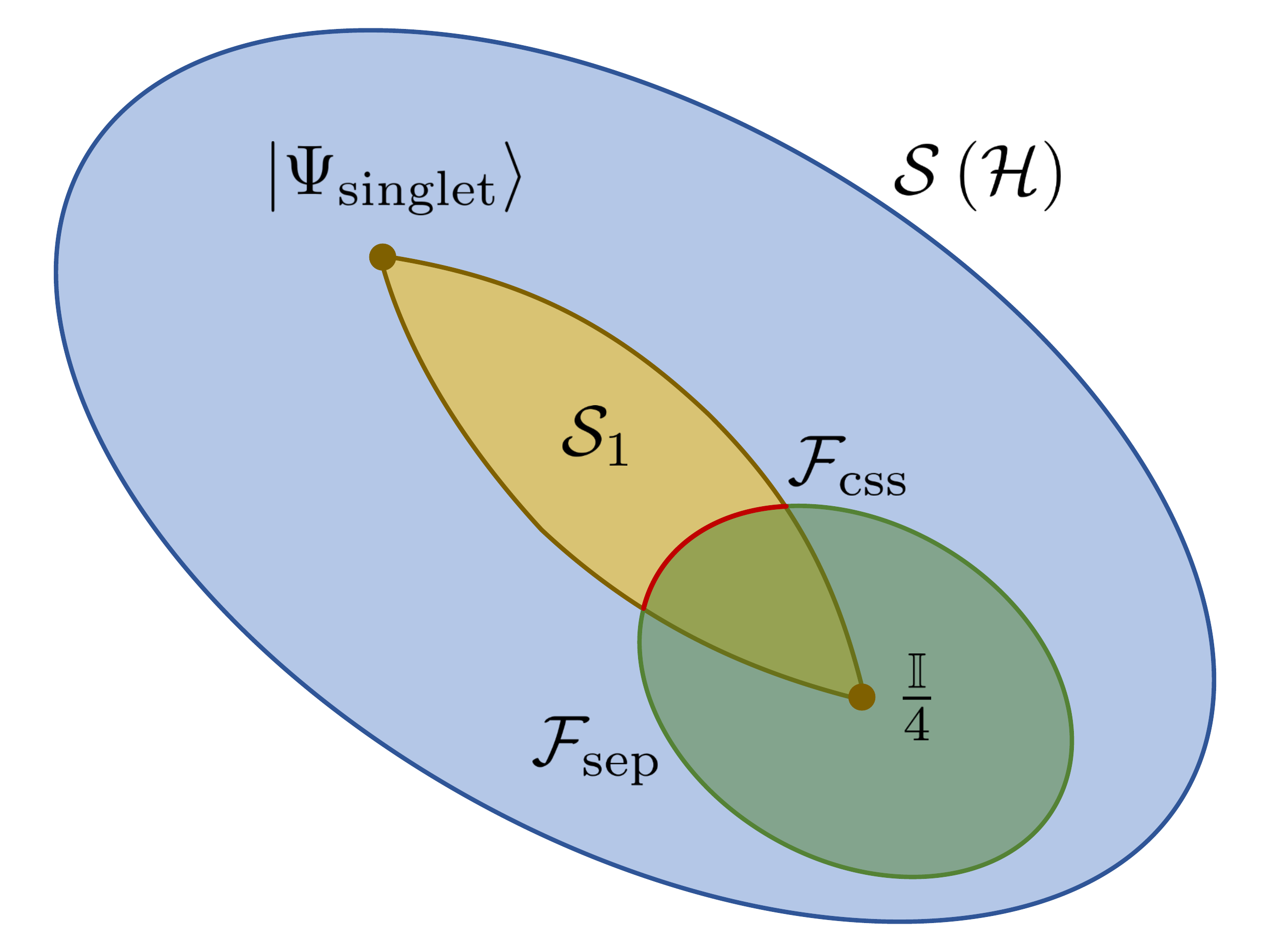}
\hspace{1cm}
\includegraphics[width=0.45\textwidth]{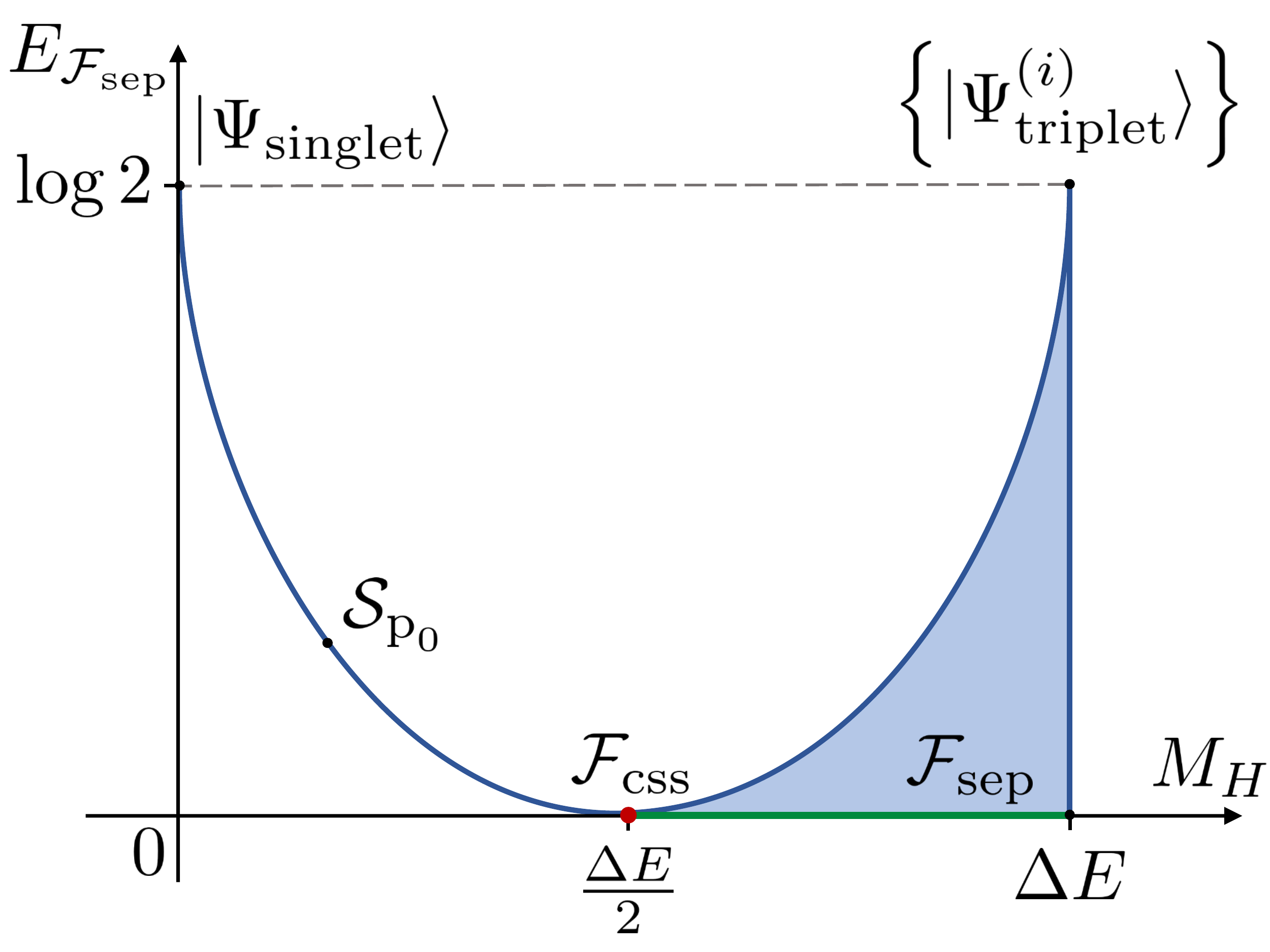}
\caption{The state-space of the multi-resource theory of
local control under energy restrictions. We consider a bipartite system composed by two qubits,
with a non-local Hamiltonian given in Eq.~\eqref{non_loc_ham}. {\bf Left.} The state-space
is represented by the blue region $\SH$, while the green region is the set of all separable states
$\f_{\text{sep}}$, and the red set on its boundary is a subset of separable states $\f_{\text{css}}$,
defined in Eq.~\eqref{css_invariant}. The yellow region contains every state diagonal in the
energy eigenbasis (Bell’s basis). In our simplified example, we restrict the state-space to a portion
of the yellow region, that is, to a subset of diagonal states. Specifically, we consider the set $\s_1$,
containing any states between the singlet and the separable subset $\f_{\text{css}}$.
{\bf Right.} The diagram represents
the set of states diagonal in the energy eigenbasis (blue region). On the left part of the diagram
we find the curve associated with the set $\s_1$, whose extreme points are the singlet and the invariant
set $\f_{\text{css}}$. On the right hand side, we have all those diagonal states with $\p_0 < \frac{1}{2}$.
The green line is the set of all separable states $\f_{\text{sep}}$. It is easy to see that, although the set of
diagonal states in the Bell's basis is convex, its representation in the diagram has not to be convex,
see comments after Lem.~\ref{lem:closed_int}, in the appendix.}
\label{state_space_energy_entanglement}
\end{figure}
\par
The two batteries we use in the theory store, respectively, energy and entanglement.
One can imagine different kinds of energy batteries. For example, we could have that only
Alice (or Bob) has access to the battery, which would imply that only one of them can change
the energy of the non-local system. However, we prefer to consider a symmetric situation in
which both Alice and Bob can interact with the battery. Moreover, we chose the battery to be
non-local, so that they are effectively using the same battery, and not two local batteries. Thus,
the battery $B_W$ is composed by $m$ copies of a two-qubit system with the same
Hamiltonian of the main system, that is,
\begin{equation}
H_W = E_0 \ket{\Psi_\text{singlet}}\bra{\Psi_\text{singlet}} + E_1 \, \Pi_{\text{triplet}}.
\end{equation}
The state of the battery is
\begin{equation}
\label{energy_batt}
\omega_W(k) = \ket{\Psi_\text{singlet}}\bra{\Psi_\text{singlet}}^{\otimes k} \otimes
\ket{\Psi_\text{triplet}^{(1)}}\bra{\Psi_\text{triplet}^{(1)}}^{\otimes m - k}, 
\end{equation}
where the excited state $\ket{\Psi_\text{triplet}^{(1)}}$ could be replaced by any other triplet
state. Notice that, in order to store/provide energy, we have to change the number of triplet
and singlet states contained in the battery, and this can be done locally by both Alice and Bob.
Moreover, even if we are changing the energy of the battery, we are not modifying its entanglement,
in accord with property~\ref{item:M1}.
\par
The second battery $B_E$ is composed by $\ell$ copies of a two-qubit system with trivial Hamiltonian
$H_E \propto \Id$ (so as to be able to exchange entanglement while preserving the energy of the battery).
We choose the state of the battery to be
\begin{equation}
\label{entang_batt}
\omega_E(h) = \ket{\Psi_\text{singlet}}\bra{\Psi_\text{singlet}}^{\otimes h} \otimes
\sigma_{\text{mm}}^{\otimes \ell - h},
\end{equation}
where the state $\sigma_{\text{mm}} \in \f_{\text{css}}$, and we take it to be the maximally-mixed state on
the subspace spanned by $\ket{\Psi_\text{singlet}}$ and $\ket{\Psi_\text{triplet}^{(1)}}$, that is
\begin{equation}
\label{max_mix_sin_trip}
\sigma_{\text{mm}} = \frac{1}{2} \ket{\Psi_\text{singlet}}\bra{\Psi_\text{singlet}}
+ \frac{1}{2} \ket{\Psi_\text{triplet}^{(1)}}\bra{\Psi_\text{triplet}^{(1)}}.
\end{equation}
The change in entanglement is measured by the change in the number of singlet states $h$.
\subsubsection{Reversibility and the interconversion relation}
In order for the present multi-resource theory to admit an interconversion relation, we first
need to show that the asymptotic equivalence property of Def.~\ref{def:asympt_equivalence_multi}
is satisfied. Let us consider the subset of states $\s_{\p_0} \subset \s_1$, where $\p_0 > \frac{1}{2}$,
defined as
\begin{equation}
\label{bank_subset_en_ent}
\s_{\p_0} = \left\{ \rho \in \s_1 \ | \ \bra{\Psi_\text{singlet}}\rho\ket{\Psi_\text{singlet}} = \p_0 \right\}.
\end{equation}
It is easy to show that all the states in this subset have the same value of the energy and entanglement
monotones, which we label $\bar{M}_H$ and $\bar{E}_{\f_{\text{css}}}$ respectively. Furthermore, for any
two states in this set, we can find an allowed operation in $\A_{\text{multi}}$, see Eq.~\eqref{ent_en_all_ops},
which maps one into the other. Indeed, consider an ancillary qutrit system described by the state $\eta =
\sum_{i=1}^3 \q_i \ket{\theta_i}\bra{\theta_i}$, and the global unitary operation $U$ acting on main
system and ancilla. The unitary operation maps $\ket{\Psi_\text{triplet}^{(i)}} \ket{\theta_j}$ into
$\ket{\Psi_\text{triplet}^{(j)}} \ket{\theta_i}$, for $i, j \in \{1,2,3\}$, and acts trivially on the remaining
basis states. Then, the operation $\chn_{\eta}(\cdot) = \Tr{A}{U \left( \cdot \otimes \eta_A
\right) U^{\dagger}} \in \A_{\text{multi}}$ maps any state $\rho \in \s_{\p_0}$ into the state
\begin{equation}
\label{rev_ent_en}
\chn_{\eta}(\rho) = \p_0 \ket{\Psi_\text{singlet}}\bra{\Psi_\text{singlet}}
+ \left( 1 - \p_0 \right) \sum_{i=1}^{3} \q_i \ket{\Psi_\text{triplet}^{(i)}}\bra{\Psi_\text{triplet}^{(i)}},
\end{equation}
where the probability distribution $\left\{ \q_i \right\}_{i=1}^3$ is defined by $\eta$. By choosing different
ancillary states $\eta$, we can reach different states in $\s_{\p_0}$, proving in this way that the
resource theory satisfies asymptotic equivalence\footnote{The operation $\chn_{\eta}(\cdot)$
we introduce is allowed since we restricted the invariant set $\f_{\text{sep}}$ to $\f_{\text{css}}$.
Indeed, the above map would not leave invariant the set of separable states $\f_{\text{sep}}$.}.
\par
We can now consider the interconversion of energy and entanglement. Together with the two batteries
$B_W$ and $B_E$, one for energy and the other for entropy, we need to use a bank system. One can
show that, when diagonal states in the energy eigenbasis are considered, bank states belongs to the set
$\s_1$ introduced in the previous section. Thus, we describe the bank system using $n \gg 1$ copies
of a state $\rho_{\text{in}} \in \s_{\p_0}$, where $\p_0 > \frac{1}{2}$ (the actual form of the state is not
relevant, since we can use the allowed operation $\chn_{\eta}$ to freely select any state in this
set). In order to obtain an interconversion relation, we need to find an allowed operation in $\A_{\text{multi}}$,
acting on the global state of bank and batteries, which modifies the state of the batteries (by exchanging
resources) while leaving the state of the bank almost unchanged with respect to the relative entropy
distance from $\s_{\p_0}$.
\par
In appendix~\ref{protocol_example} we provide a protocol which performs the following
resource interconversion using an allowed operation $\A_{\text{multi}}$, 
\begin{equation}
\rho_{\text{in}}^{\otimes n} \otimes \omega_W(k) \otimes \omega_E(h)
\xleftrightarrow{\text{asympt}}
\rho_{\text{fin}}^{\otimes n} \otimes \omega_W(k') \otimes \omega_E(h').
\end{equation}
In the above transformation, the initial state of the bank $\rho_{\text{in}}$ is mapped into a state
$\rho_{\text{fin}} \in \s_{\p'_0}$, where $\p'_0 = \p_0 + O(n^{-1})$. The energy battery $B_W$
is mapped from the initial state $\omega_W(k)$, containing $k$ copies of the ground state of
$H_{AB}$, into the final state $\omega_W(k')$ with $k' = k + \Delta k$ copies of this ground state,
where $\Delta k > 0$ is arbitrary big. Likewise, the entanglement battery $B_E$ changes from
the initial state $\omega_E(h)$, containing $h$ singlets, to the final state $\omega_E(h')$
containing $h' = h - \log \frac{\p_0}{1-\p_0} \, \Delta k$ singlets. From the above transformation
one is able to derive an interconversion relation between energy and entanglement,
\begin{equation}
\label{int_rel_ent_ene}
\Delta W_W = - \frac{\Delta E}{\log \frac{\p_0}{1-\p_0}} \, \Delta W_E,
\end{equation}
where $\Delta W_W = M_H \left( \omega_W(k') \right) - M_H \left( \omega_W(k) \right)$ is the
amount of energy exchanged, $\Delta W_E = E_{\f_{\text{css}}} \left( \omega_E(h') \right) -
E_{\f_{\text{css}}} \left( \omega_E(h) \right)$ is the amount of entanglement exchanged, and
$\Delta E = E_1 - E_0$ is the energy gap of the Hamiltonian $H_{AB}$. Additionally, we find
that the change in monotone $E_{\s_{\p_0}}$ between the initial and final global state of the
bank is negligible (for $n \rightarrow \infty$), in accord with property~\ref{item:X1}.
\section{Conclusions}
\label{end}
{\bf From multiple constraints to a resource theory.}
With the present work we set the mathematical ground for the development of
resource theories with multiple resources able to describe new physical scenarios. Our
construction of multi-resource theories is based on the definition of their class of allowed
operations. First, we pinpoint the resources that compose the theory, and we introduce the
corresponding single-resource theories. Then, we define the set of allowed operations for
the multi-resource theory as the one composed by the maps in the intersection of the different
classes of allowed operations of each single-resource theory, Eq.~\eqref{all_ops_multi}.
This construction leaves the theory with multiple invariant sets, some of which are the sets of
free states of the relevant single-resource theories. It is worth remarking again that, in
multi-constraint theories, there is a difference between the set of free states and the invariant
sets (in contrast with the case of single-resource theories), and a multi-resource theory can
have multiple invariant sets and no free states,  Fig.~\ref{fig:invariant_sets_structure}.
\par
{\bf Reversibility.}
Together with the introduction of a general framework for multi-resource theories, we have studied
the properties of these reversible theories. In particular, to analyse reversibility when multiple
resources are present, we have first introduced the asymptotic equivalence property, see
Def.~\ref{def:asympt_equivalence_multi}. This property implies that a unique monotone can be
used to quantify each resource. Furthermore, in the case of single-resource theories, it coincides with the usual
notion of reversible rates of conversion. We know of multi-resource theories that satisfy this property,
see the two examples provided in Sec.~\ref{examples}. However, it would be interesting to study
which of the other, already existing, multi-resource theories satisfy the property of
Def.~\ref{def:asympt_equivalence_multi}. Ultimately, one would hope to find some general condition
according to which a multi-resource theory is reversible, similarly to what has been found in
Ref.~\cite{brandao_reversible_2015}. 
\par
{\bf The role of batteries.}
A crucial feature of our framework is the presence of batteries, used to store and quantify the
resources exchanged during a state transformation over the main system. While batteries can be
defined for single-resource theories as well, they do not seem to play the same fundamental
role in that case, since one can quantify the amount of resource contained in a system using
the conversion rate, see Def.~\ref{def:rate_conversion} in appendix~\ref{rev_theory_sing}.
However, the conversion rate is linked to a change in the number of copies, for example
$\rho^{\otimes n} \to \sigma^{\otimes k}$, where it is implicitly assumed that the remaining
$|n-k|$ copies of the system are in a free state. Since the framework allows us to model
theories with no free states, we cannot change the number of systems with the allowed operations,
and therefore we need to use batteries to quantify the amount of resources. We have seen in this
paper what are the main properties for these batteries, primarily property~\ref{item:M1}, which
requires each battery to store one and only one of the resource. It would be interesting to study
these systems more carefully, possibly linking them to the kind of batteries used for fluctuation
theorems~\cite{alhambra_entanglement_2017,alhambra_fluctuating_2016-1,renes_relative_2016,
morris_quantum_2018}, which are described by states in a big superposition, so as to always remain
uncorrelated from the main system during a state transformation~\cite{van_dam_universal_2003,
harrow_entanglement_2010}. A different line of research in this direction could involve the
study of correlated and entangled batteries, already explored in the setting of the single-resource
theory of thermodynamics~\cite{alicki_entanglement_2013,hovhannisyan_entanglement_2013}.
\par
{\bf Interconversion and further examples.}
We have studied the interconversion of resources and we have introduced a first law for multi-resource
theories, Eq.~\eqref{eq:first_law}, valid when the theories are reversible and the invariant sets are
disjoint. We have provided two examples of theories with a first law, one related to thermodynamics,
and the other concerning a theory of local control under energy restriction. In this latter example, we
have studied an extremely simplified case, due to the fact that reversibility has not been proved in general
for this theory. Due to the high importance of both non-locality and thermodynamics in the field of quantum
technology and many-body physics, we believe that a complete analysis of this multi-resource theory would
be useful. Furthermore, it would be interesting to know which other multi-resource theories allow for an
interconversion relation, and whether it is possible to define interconversion for theories with a different
structure of invariant sets, by for instance relaxing the assumptions made on the bank. For example, one
could consider bank states from which both resources could in principle be extracted, and forbid such
extraction by further constraining the class of allowed operations.
\par
{\bf Multiple ways to build a multiple-resource theory.}
In general, there could be different ways to intersect constraints in order to obtain the same final
resource theory, and some of these constructions are a better fit for the analysis presented here than others.
For example, the resource theory of thermodynamics equipped with Thermal Operations can be built as the
intersection of either (1) the resource theories of information and energy, as we have done in
Sec.~\ref{bank_monotone}, or (2) the resource theories of athermality and coherence~\cite{aberg_catalytic_2014,
lostaglio_description_2015, kwon_clock--work_2018}. However, the most convenient setting for the study of this
latter construction is the single-copy regime, since in the many-copy scenario coherence is lost, as this
quantity scales sub-linearly in the number of copies of the system considered.
\par
{\bf Beyond the asymptotic limit.}
The concrete results presented here for reversibility and interconversion of resources are only
valid in the asymptotic limit where many independent and identically distributed copies of a system
are considered. However, the general framework we introduced to describe resource theories with multiple
resources can also be applied to scenarios with a single system. Understanding how resources
can be exchanged in the single-copy regime, and studying the corrections to the first law in such a regime
are worthwhile questions to pursue. We believe that extending the notion of batteries to the single-shot
regime should be the first step toward the definition of a complete framework for multi-resource theories.
However, we anticipate that this will be a highly non-trivial task, since in the single-copy case a resource is
not generally quantified by a single measure, which complicates the definition of batteries, currently given
through property~\ref{item:M1}. These difficulties are exemplified by the single-shot version of the resource
theory of thermodynamics, where the $\alpha$-R{\'e}nyi divergences from the thermal state are all valid
resource measures. A possible way forward in this setting might be the definition of an arbitrary notion of
resource, for instance in terms of the number of resourceful states contained in the battery. Otherwise, the
fluctuation relations for arbitrary resources~\cite{alhambra_entanglement_2017} and their connection to
majorization could be useful conecpts for quantifying resources in the single-shot regime.
\tocless{\section*}{Acknowledgement}
We thank the anonymous TQC referees for feedbacks, and Tobias Fritz for detailed comments on
a previous version of this manuscript, CS is supported by the EPSRC (grant number {EP/L015242/1}).
LdR acknowledges support from the Swiss National Science Foundation through SNSF project
No.~{200020$\_$165843} and through the National Centre of Competence in Research \emph{Quantum
Science and Technology} (QSIT), and from  the FQXi grant \emph{Physics of the observer}. CMS is
supported by the Engineering and Physical Sciences Research Council (EPSRC) through the doctoral
training grant {1652538}, and by  Oxford-Google DeepMind graduate scholarship. CMS would like to
thank the Department of Physics and Astronomy at UCL  for their hospitality. PhF acknowledges support from the Swiss
National Science Foundation (SNSF) through the Early PostDoc.Mobility Fellowship No. {P2EZP2$\_$165239} hosted
by the Institute for Quantum Information and Matter (IQIM) at Caltech, from the IQIM which is a National Science
Foundation (NSF) Physics Frontiers Center (NSF Grant {PHY-1733907}), from the Department of Energy
Award {DE-SC0018407}, as well as from the Deutsche Forschungsgemeinschaft (DFG) Research Unit {FOR 2724}.
JO is supported by the Royal Society, and by an EPSRC Established Career Fellowship. We thank the COST Network
{MP1209} in Quantum Thermodynamics.
\tocless{\subsection*}{Author contributions}
All authors contributed significantly to the ideas behind this work and to the development of the
general framework (Sec.~\ref{multi_res_framework}). CS, LdR and JO developed the results on
batteries, bank states and the first law (Secs.~\ref{rev_theory_mult}, \ref{interconv}, \ref{examples}).
CS wrote the proofs and initial draft.
\newpage
\appendix
{\Large{\sc{Appendix}}}
\section{Reversibility and asymptotic equivalence for single-resource theories}
\label{rev_theory_sing}
In this section we show that, for a single-resource theory, the asymptotic equivalence
property of Def.~\ref{def:asympt_equivalence_multi} is equivalent to the notion of reversibility
given in terms of rates of conversion. Let us first introduce the concept of rate of conversion
for a single-resource theory, see Ref.~\cite{horodecki_quantumness_2012}. The definition of
rate we use coincides with the one used in the literature, with the difference that we are
making explicit use of the partial trace and of the addition of free states. In fact, we prefer
not to include these operations within the set $\A$, as we want the allowed operations to
preserve the number of copies of the system they act over (with the exception of sub-linear
ancillae).
\begin{definition}
\label{def:rate_conversion}
Consider a single-resource theory with allowed operations $\A$ and free states $\f$, and
two states $\rho, \sigma \in \SH$. We define the \emph{rate of conversion} from $\rho$ to
$\sigma$ as
\begin{align}
R(\rho \rightarrow \sigma) = \sup \bigg\{ \frac{k_n}{n} \ | \
&\text{{\rm either}} \
\lim_{n \rightarrow \infty}
\left(
\min_{\tilde{\chn}_n}
\left\| \Tr{n - k_n}{ \tilde{\chn}_n (\rho^{\otimes n}) } - \sigma^{\otimes k_n} \right\|_1
\right) = 0 \nonumber \\
&\text{{\rm or}} \
\lim_{n \rightarrow \infty}
\left(
\min_{\tilde{\chn}_{k_n}}
\left\| \tilde{\chn}_{k_n} (\rho^{\otimes n} \otimes \gamma_{k_n - n}) - \sigma^{\otimes k_n} \right\|_1
\right) = 0 \ , \nonumber \\
&\text{{\rm where}} \ \gamma_{k_n - n} \in \f^{(k_n - n)}
\bigg\}.
\end{align}
where the maps $\tilde{\chn}_n$ have been defined in Eq.~\eqref{allowed_ancilla}, and they
are of the form $\tilde{\chn}_n (\cdot) = \Tr{A}{\chn_n ( \cdot \otimes \eta^{(A)}_n ) }$,
with $\chn_n \in \A^{(n+o(n))}$ and $\eta^{(A)}_n \in \SHn{o(n)}$.
\end{definition}
\par
Now that the notion of rate is defined, we introduce the concept of \emph{reversible} single-resource
theory,
\begin{definition} \label{def:reversibility}
A single-resource theory with allowed operations $\A$ and free states $\f$ is \emph{reversible} if,
given any non-free states $\rho, \sigma \in \SH$, the rate of conversion from $\rho$ to $\sigma$
is such that $R(\rho \rightarrow \sigma) \in (0 , \infty)$, and $R(\rho \rightarrow \sigma) R(\sigma
\rightarrow \rho) = 1$.
\end{definition}
The above notion of reversibility is based on the rates of conversion between two resourceful states.
However, it is not clear how to extend Def.~\ref{def:rate_conversion} to the case of multiple resources,
since the set of free states might be empty for multi-resource theories. For this reason, we have introduced
the property of asymptotic equivalence in Sec.~\ref{asympt_eqiv_prop}. This property also
apply to the single-resource theory case, when $m = 1$.
\par
Now we want to show that Defs.~\ref{def:asympt_equivalence_multi} and \ref{def:reversibility}, for a
single-resource theory, coincide. First, let us introduce a function $f : \SHn{n} \rightarrow \R$ (more
formally, a family of functions) with the following properties,
\begin{description}
\item[SM1\label{item:SM1}] For each $n \in \N$, the function $f$ is monotonic under the set of
allowed operations $\A^{(n)}$, that is
\begin{equation}
f \left( \chn_n \left( \rho_n \right) \right) \leq f \left( \rho_n \right) , \qquad 
\forall \, \rho_n \in \SHn{n} \ , \
\forall \, \chn_n \in \A^{(n)}.
\end{equation}
\item[SM2\label{item:SM2}] For each $n \in \N$, the function $f$ is equal to $0$ for all states
$\gamma_n \in \f^{(n)}$, that is
\begin{equation}
f \left( \gamma_n \right) = 0 , \qquad  \forall \, \gamma_n \in \SHn{n}.
\end{equation}
\item[SM3\label{item:SM3}] The function $f$ is asymptotic continuous.
\item[SM4\label{item:SM4}] The function $f$ is monotonic under partial tracing, that is
\begin{equation}
f \left( \Tr{k}{\rho_n} \right) \leq f \left( \rho_n \right)
, \qquad
\forall \, n, k \in \N \ , \ k < n
\ , \
\forall \, \rho_n \in \SHn{n}.
\end{equation}
\item[SM5\label{item:SM5}] For each $n,k \in \N$, the function $f$ is sub-additive, that is
\begin{equation}
f \left( \rho_n \otimes \rho_k \right) \leq f \left( \rho_n \right) + f \left( \rho_k \right)
, \qquad
\forall \, \rho_n \in \SHn{n}
\ , \
\forall \, \rho_k \in \SHn{k}.
\end{equation}
\item[SM6\label{item:SM6}] For any given sequence of states $\left\{ \rho_n \in \SHn{n} \right\}$,
the function $f$ scales sub-extensively, that is, $f \left( \rho_n \right) = O(n)$.
\end{description}
Notice that property~\ref{item:SM6} implies that the function $f$ is regularisable. Furthermore, the
value of $f$ is preserved if we add free states, that is,
\begin{equation}
\label{free_zero}
f(\rho_n \otimes \gamma_k) = f(\rho_n)
, \qquad \forall \, \rho_n \in \SHn{n} \ , \ \forall \, \gamma_k \in \f^{(k)},
\end{equation}
which follows from properties~\ref{item:SM2}, \ref{item:SM4}, and \ref{item:SM5}.
\par
The first lemma we introduce show that the rate of conversion of a reversible single-resource theory
is linked to the function $f$ satisfying the above properties. Notice that this proof is analogous to
the one of Ref.~\cite{horodecki_entanglement_2001}, with the difference that we are allowing for the
presence of a sub-linear ancilla in the definition of rate, following the notion of ``seed regularisation''
introduced in Ref.~\cite[Sec.~9]{fritz_resource_2015}.
\begin{lem} \label{lem:reversible_rate}
Consider a reversible resource theory with allowed operations $\A$ and free states $\f$, and the function
$f$ satisfying~\ref{item:SM1} -- \ref{item:SM6}. Then, for all non-free states $\rho, \sigma \in \SH$, we have that
\begin{equation}
R(\rho \rightarrow \sigma) = \frac{f^{\infty}(\rho)}{f^{\infty}(\sigma)}
\end{equation}
\end{lem}
\begin{proof}
Let consider $\rho$ and $\sigma$ such that $R(\rho \rightarrow \sigma) \leq 1$
(the proof of the other case is equivalent). Then, there exists a sequence
of operations $\left\{ \tilde{\chn}_n \right\}$ of the form given in
Eq.~\eqref{allowed_ancilla} such that
\begin{equation}
\lim_{n \rightarrow \infty}
\left\| \Tr{n - k_n}{ \tilde{\chn}_n (\rho^{\otimes n}) } - \sigma^{\otimes k_n}
\right\|_1 = 0
\end{equation}
where $\lim_{n \rightarrow \infty} \frac{k_n}{n} = R( \rho \rightarrow \sigma )$.
If we use the asymptotic continuity of the function $f$, property~\ref{item:SM3}, we obtain
\begin{equation}
f \left( \Tr{n - k_n}{ \tilde{\chn}_n (\rho^{\otimes n}) } \right) =
f \left( \sigma^{\otimes k_n} \right) + o ( k_n ).
\end{equation}
Let us now consider the lhs of the above equation. Using the properties of the monotone
$f$, together with the definition of $\tilde{\chn}_n$ in terms of sub-linear ancillae and
allowed operations, we can prove the following chain of inequalities,
\begin{align}
f \left( \Tr{n - k_n}{ \tilde{\chn}_n (\rho^{\otimes n}) } \right)
&\leq f \left( \tilde{\chn}_n (\rho^{\otimes n}) \right)
= f \left( \Tr{A}{\chn_n ( \rho^{\otimes n} \otimes \eta^{(A)}_n ) } \right)
\leq f \left( \chn_n ( \rho^{\otimes n} \otimes \eta^{(A)}_n ) \right) \nonumber \\
&\leq f \left( \rho^{\otimes n} \otimes \eta^{(A)}_n \right) 
\leq f \left( \rho^{\otimes n} \right) + f \left( \eta^{(A)}_n \right)
\leq f \left( \rho^{\otimes n} \right) + o(n)
\end{align}
where the first and second inequalities follow from property~\ref{item:SM4}, the equality
follows from the definition of $\tilde{\chn}_n$, see Eq.~\eqref{allowed_ancilla}, the third
inequality follows from monotonicity under allowed operations, property~\ref{item:SM1}, the
forth inequality from sub-additivity, property~\ref{item:SM5}, and the last one from the fact
that the ancillary system is sub-linear in $n$ together with property~\ref{item:SM6}. Thus,
combining the last two equations, we get
\begin{equation}
f \left( \rho^{\otimes n} \right) \geq f \left( \sigma^{\otimes k_n} \right) + o ( n ).
\end{equation}
We can now divide the left and right hand side of the above equation by $n$, obtaining
\begin{equation}
\frac{1}{n} \, f \left( \rho^{\otimes n} \right) \geq
\frac{k_n}{n} \, \frac{1}{k_n} \, f \left( \sigma^{\otimes k_n} \right) + o ( 1 ).
\end{equation}
By taking the limit of $n \rightarrow \infty$, and using the fact that $f$ is regularisable
(which follows from property~\ref{item:SM6}) together with the definition of rate, we get
\begin{equation}
f^{\infty} \left( \rho \right) \geq R(\rho \rightarrow \sigma) \, f^{\infty} \left( \sigma \right).
\end{equation}
We can also consider the reverse transformation, mapping $n$ copies of the state $\sigma$
into $k'_n$ copies of $\rho$. Using the same steps used above, together with the fact that
the monotone $f$ is equal to zero over free states, property~\ref{item:SM2}, we can show that
\begin{equation}
f^{\infty} \left( \sigma \right) \geq R(\sigma \rightarrow \rho) \, f^{\infty} \left( \rho \right).
\end{equation}
If we now use the reversibility property, which implies $R(\sigma \rightarrow \rho) = \frac{1}{R(\rho \rightarrow \sigma)}$,
we find that
\begin{equation}
\frac{f^{\infty} \left( \rho \right)}{f^{\infty} \left( \sigma \right)}
\geq
R(\rho \rightarrow \sigma)
\geq 
\frac{f^{\infty} \left( \rho \right)}{f^{\infty} \left( \sigma \right)}
\end{equation}
which proves the lemma.
\end{proof}
Furthermore, we introduce a second small lemma, that can be found in
Ref.~\cite[Prop.~13]{donald_uniqueness_2002}, 
\begin{lem} \label{lem:add_regularised_mon}
Given a regularisable function $f : \SHn{n} \rightarrow \R$, the regularised version is extensive,
\begin{equation}
f^{\infty}(\rho^{\otimes k}) = k \, f^{\infty}(\rho) \ , \ \forall \, \rho \in \SH \ , \ \forall \, k \in \N.
\end{equation}
\end{lem}
\begin{proof}
Consider a function $h : \R \rightarrow \R$, such that $\lim_{n \rightarrow \infty} h(n) = L < \infty$.
This is equivalent to say that
\begin{equation} \label{limit_def}
\forall \, \epsilon > 0, \exists \, c \in \R \ : \ | h(n) - L | < \epsilon, \ \forall \, n > c.
\end{equation}
Let us now consider an invertible function $g : \R \rightarrow \R$, and consider $m \in \R$ such
that $n = g(m)$. Then, we can rewrite Eq.~\eqref{limit_def} as
\begin{equation}
\forall \, \epsilon > 0, \exists \, c \in \R \ : \ | h(g(m)) - L | < \epsilon, \ \forall \, g(m) > c,
\end{equation}
and by defining $\tilde{c} = g^{-1}(c)$, we get
\begin{equation}
\forall \, \epsilon > 0, \exists \, \tilde{c} \in \R \ : \ | h(g(m)) - L | < \epsilon, \ \forall \, m > \tilde{c}.
\end{equation}
Therefore, we have $\lim_{m \rightarrow \infty} h(g(m)) = L $.
\par
If we choose $h(n) = \frac{1}{n} f(\rho^{\otimes n})$, whose limit is $L = f^{\infty} (\rho)$,
and we use the reversible function $g(m) = k \cdot m$ where $k \in \N$ is fixed, we get
\begin{equation}
f^{\infty} (\rho) = \lim_{m \rightarrow \infty} \frac{1}{k \cdot m} f(\rho^{\otimes k \cdot m}) =
\frac{1}{k} \lim_{m \rightarrow \infty} \frac{1}{m} f( ( \rho^{\otimes k})^{\otimes m} ) =
\frac{1}{k} f^{\infty} (\rho^{\otimes k}),
\end{equation}
which proves the lemma.
\end{proof}
We can now show that a single-resource theory which is reversible also satisfies the
asymptotic equivalence property, and vice versa.
\begin{thm} \label{thm:reversible_asympt_equiv}
Consider the resource theory with allowed operations $\A$ and free states $\f$. If the theory is
reversible, then it satisfies the asymptotic equivalence property with respect to a function $f$
satisfying the properties~\ref{item:SM1} -- \ref{item:SM6}, and viceversa.
\end{thm}
\begin{proof}
{\bf (a)} Let us first assume that the theory is reversible. Then, if we consider two non-free states $\rho,
\sigma \in \SH$ such that $f^{\infty}(\rho) = f^{\infty}(\sigma)$, and we use Lem.~\ref{lem:reversible_rate},
we find that the rate of conversion is $R(\rho \rightarrow \sigma) = 1$. Then, there exists a sequence
of operations  $\left\{ \tilde{\chn}_n \right\}$ that approach this limit in one of two ways. In one case,
we have
\begin{equation}
\left\| \Tr{n - k_n}{\tilde{\chn}_n (\rho^{\otimes n})} - \sigma^{\otimes k_n} \right\|_1 \rightarrow 0.
\end{equation}
Notice that, since we have $\frac{k_n}{n} \rightarrow 1$, it follows that $n - k_n = o(n)$. Then, the above
equation coincides with the second part of Def.~\ref{def:asympt_equivalence_multi}, where we are
mapping $\rho^{\otimes k_n}$ into $\sigma^{\otimes k_n}$, and the sub-linear ancilla is $\eta'^{(A)}_n
= \eta^{(A)}_n \otimes \rho^{n - k_n}$, where $\eta^{(A)}_n$ is completely arbitrary, and come from the
definition of $\tilde{\chn}_n$. Alternatively, we can have that the sequence of maps is such that
\begin{equation}
\left\|
\tilde{\chn}_{k_n} (\rho^{\otimes n} \otimes \gamma_{k_n - n} )
-
\sigma^{\otimes k_n}
\right\|_1 \rightarrow 0.
\end{equation}
We now use the monotonicity of the trace distance under discarding subsystems to obtain
\begin{equation}
\left\| \Tr{k_n - n}{\tilde{\chn}_{k_n} (\rho^{\otimes n} \otimes \gamma_{k_n - n} )} - \sigma^{\otimes n} \right\|_1 \rightarrow 0.
\end{equation}
Again, the above equation coincides with the second part of Def.~\ref{def:asympt_equivalence_multi}, where
we are mapping $\rho^{\otimes n}$ into $\sigma^{\otimes n}$, and the sub-linear ancilla is $\eta'^{(A)}_n
= \eta^{(A)}_n \otimes \gamma_{k_n - n}$. This proves the validity of one direction of the asymptotic
equivalence property. To prove the other direction (existence of a sequence of maps implies same value
of the monotone on the two states), we can use the fact that, if there exists a sequence of maps
$\left\{ \tilde{\chn}_n \right\}$ sending $\rho^{\otimes n}$ into $\sigma^{\otimes n}$, then the rate
of conversion is $R(\rho \rightarrow \sigma) = 1$. Then, with the help of Lem.~\ref{lem:reversible_rate},
which is valid for reversible theories, we obtain that $f^{\infty}(\rho) = f^{\infty}(\sigma)$. This proves the
other direction of the asymptotic equivalence property.
\par
{\bf (b)} Let now assume that the theory satisfies the asymptotic equivalence property. Consider any two
non-free states $\rho, \sigma \in \SH$, and suppose that $f^{\infty}(\rho) \leq f^{\infty}(\sigma)$ (in
the other case, the proof would follow analogously to the one we are presenting). Take $n,k \in \N$
such that $n \, f^{\infty}(\rho) = k \, f^{\infty}(\sigma)$, and let us use the extensivity of $f^{\infty}$,
Lem.~\ref{lem:add_regularised_mon}. Then, we have $f^{\infty}(\rho^{\otimes n}) = f^{\infty}
(\sigma^{\otimes k})$. Using the property of the function $f$ shown in Eq.~\eqref{free_zero},
we find that
\begin{equation}
f^{\infty}(\rho^{\otimes n}) = f^{\infty}(\sigma^{\otimes k} \otimes \gamma_{n - k}),
\end{equation}
where we add the free state $\gamma_{n - k} \in \f^{(n - k)}$ to the right hand side since $n \geq k$.
Then, we can use the asymptotic equivalence property, which implies the existence of a sequence of
maps $\left\{ \tilde{\chn}_{m \cdot n} \right\}_m$, see Eq.~\ref{allowed_ancilla}, such that
\begin{equation}
\lim_{m \rightarrow \infty}
\left\| \chn_{m \cdot n} (\rho^{\otimes m \cdot n}) -
\sigma^{\otimes m \cdot k} \otimes \gamma_{n - k}^{\otimes m} \right\|_1 = 0.
\end{equation}
If we use the monotonicity of the trace distance under partial tracing, we find that
\begin{equation}
\lim_{m \rightarrow \infty}
\left\| \Tr{m \cdot (n-k)}{\chn_{m \cdot n} (\rho^{\otimes m \cdot n})} -
\sigma^{\otimes m \cdot k} \right\|_1 = 0.
\end{equation}
The existence of this sequence of maps implies that the rate of conversion $R(\rho \rightarrow \sigma)
\geq \frac{k}{n}$. At the same time, we can use asymptotic equivalence to find a sequence of maps
$\left\{ \tilde{\chn}'_{m \cdot n} \right\}_m$ performing the reverse process. Using a similar argument
to the one presented above, we find that $R(\sigma \rightarrow \rho) \geq \frac{n}{k}$. As a result, we
find that the product of the forward and reverse rates of conversion is $R(\rho \rightarrow \sigma)
R(\sigma \rightarrow \rho) \geq 1$. However, this product cannot be higher than one, as otherwise
we would be able to perform a cyclic transformation turning free states intro resourceful one, which
is forbidden under allowed operations, see also Ref.~\cite{popescu_thermodynamics_1997}.
Therefore, we find that $R(\rho \rightarrow \sigma) R(\sigma \rightarrow \rho) = 1$, which closes the
proof.
\end{proof}
\section{Convex boundary and bank states}
\label{convex_bound}
In the following, we consider the case of a two-resource theory $\rt_{\text{multi}}$ defined on the Hilbert
space $\hil$. The set of allowed operations is $\A_{\text{multi}} = \A_1 \cap \A_2$, where each $\A_i$ is
a subset of the set of all CPTP maps that leave the set of states $\f_i$ invariant, $i = 1, 2$. We ask the
resource theory $\rt_{\text{multi}}$ to satisfy the asymptotic equivalence property with respect to the
monotones $E_{\f_1}$ and $E_{\f_2}$. Furthermore, we assume that the two invariant sets satisfy the
properties~\ref{item:F1}--\ref{item:F5}. Thus, it follows
from Thm.~\ref{thm:reversible_multi} and Prop.~\ref{thm:properties_rel_ent} that the two monotones
$E^{\infty}_{\f_1}$ and $E^{\infty}_{\f_2}$ uniquely quantify the resources in this theory. As a result,
we can represent the state-space of $\rt_{\text{multi}}$ in a two-dimensional diagram, as shown in
Fig.~\ref{fig:resource_diagram}.
\par
We choose the two invariant sets of the theory to be disjoints, i.e., $\f_1 \cap \f_2 = \emptyset$, and
we focus on the set of bank states $\f_{\text{bank}} \subset \SH$. Since in this section we are
not making any assumption on the additivity (or extensivity) of the monotones $E_{\f_i}$'s, we have
that the set of bank states is here defined as
\begin{align}
\label{eq:bank_reg}
\f_{\text{bank}} = \big\{ \rho \in \SH \ | \ \forall \, \sigma \in \SH , \
&E^{\infty}_{\f_1}(\sigma) > E^{\infty}_{\f_1}(\rho) \ \text{or} \nonumber \\
&E^{\infty}_{\f_2}(\sigma) > E^{\infty}_{\f_2}(\rho) \ \text{or} \nonumber \\
&E^{\infty}_{\f_1}(\sigma) = E^{\infty}_{\f_1}(\rho) \, \text{and} \, E^{\infty}_{\f_2}(\sigma) = E^{\infty}_{\f_2}(\rho) \, \big\}.
\end{align}
Notice that this set coincides with the one of Eq.~\eqref{set_f3} when property~\ref{item:F5b} is
satisfied, and therefore the results we obtain in this appendix apply to Sec.~\ref{interconv} as well.
It is easy to show that $E^{\infty}_{\f_2}(\rho) > E^{\infty}_{\f_2}(\f_2) = 0 \ \forall \, \rho \in \f_1$,
and similarly $E^{\infty}_{\f_1}(\rho) > E^{\infty}_{\f_1}(\f_1) = 0 \ \forall \, \rho \in \f_2$. Moreover,
inside both invariant sets $\f_1$ and $\f_2$ we can find a subset of states with minimum value of the monotones
$E^{\infty}_{\f_2}$ and $E^{\infty}_{\f_1}$, respectively. We define these sets as
\begin{align}
\label{f_1_min}
\f_{1, \min} &= \left\{ \sigma \in \f_1 \, | \, E^{\infty}_{\f_2}(\sigma) = \min_{\rho \in \f_1} E^{\infty}_{\f_2}(\rho) \right\}
\subseteq \f_1, \\
\label{f_2_min}
\f_{2, \min} &= \left\{ \sigma \in \f_2 \, | \, E^{\infty}_{\f_1}(\sigma) = \min_{\rho \in \f_2} E^{\infty}_{\f_1}(\rho) \right\}
\subseteq \f_2.
\end{align}
Given these two subsets, we can then define the following real intervals,
\begin{align}
\label{interval_1}
I_1 &= \left[ E^{\infty}_{\f_1}(\f_1) = 0 \, ; \, E^{\infty}_{\f_1}(\f_{2, \min}) \right], \\
\label{interval_2}
I_2 &= \left[ E^{\infty}_{\f_2}(\f_2) = 0 \, ; \, E^{\infty}_{\f_2}(\f_{1, \min}) \right].
\end{align}
In what follows, we make use of the following two properties of the monotones $E^{\infty}_{\f_i}$'s,
\begin{itemize}
\item \emph{Asymptotic continuity}, which follows from the assumptions~\ref{item:F1}--\ref{item:F5} over the sets $\f_i$'s,
as shown in Refs.~\cite{christandl_structure_2006,brandao_reversible_2015}.
\item \emph{Convexity}, which follows from the assumptions~\ref{item:F2} and \ref{item:F4} over the sets $\f_i$'s,
as shown in Ref.~\cite{donald_uniqueness_2002}, Prop.~13.
\end{itemize}
We can now state the following lemma, concerning the value of the monotones $E^{\infty}_{\f_i}$'s for bank states.
\par
\begin{lem}
\label{lem:closed_int}
Consider the multi-resource theory $\rt_{\text{multi}}$ with allowed operations $\A_{\text{multi}}$, and
invariant sets $\f_1$ and $\f_2$ which satisfy properties~\ref{item:F1}--\ref{item:F5}, and $\f_1 \cap \f_2
= \emptyset$. If the theory satisfies the asymptotic equivalence property with respect to the monotones
$E_{\f_1}$ and $E_{\f_2}$, then for all bank states $\rho \in \f_{\text{bank}}$ we have that
$E^{\infty}_{\f_1}(\rho) \in I_1$ and $E^{\infty}_{\f_2}(\rho) \in I_2$.
\end{lem}
\begin{proof}
Suppose, for example, that there exists a bank state $\rho \in \f_{\text{bank}}$ such that
$E^{\infty}_{\f_1}(\rho) \notin I_1$, that is, $\exists \, \sigma \in \f_{2, \min}$ such that
$E^{\infty}_{\f_1}(\sigma) < E^{\infty}_{\f_1}(\rho)$. By definition of $\f_2$ we also have
that $E^{\infty}_{\f_2}(\sigma) \leq E^{\infty}_{\f_2}(\rho)$. These two inequalities, however,
contradict the fact that $\rho$ is a bank state, see Eq.~\eqref{eq:bank_reg}, and conclude
the proof.
\end{proof}
It is easy to show that for all $\bar{E}_{\f_1} \in I_1$ there exists (at least) one state $\rho \in \SH$
such that $E^{\infty}_{\f_1}(\rho) = \bar{E}_{\f_1}$, and the same holds for $I_2$. The proof that
$\forall \, \bar{E}_{\f_1} \in I_1, \ \exists \, \rho \in \SH \, : \, E^{\infty}_{\f_1}(\rho) = \bar{E}_{\f_1}$
follows from two facts: ({\it i}) $\SH$ is a compact and path-connected set, and therefore its image
under the (asymptotic) continuous function $E^{\infty}_{\f_1}$ is a compact and path-connected set
in $\R$, that is, a closed and bounded interval $I_{1,\SH}$, and ({\it ii}) $I_1 \subseteq I_{1,\SH}$.
\par
Let us now define, in the $E^{\infty}_{\f_1}$--$E^{\infty}_{\f_2}$ diagram, the curve of bank states,
which lies on part of the boundary of the state-space, as per definition in Eq.~\eqref{eq:bank_reg}.
The curve is defined as
\begin{equation}
\label{curve_bank}
\gamma_{\text{bank}} = \left\{ \left( E^{\infty}_{\f_1}(\rho) , E^{\infty}_{\f_2}(\rho) \right) \ | \ \rho \in \f_{\text{bank}} \right\},
\end{equation}
where $\f_{\text{bank}}$ is the set of bank states of the theory. It is easy to see that this curve is
completely contained within the subset of $\R^2$ given by $I_1 \times I_2$. Together with this curve,
we can introduce the real-valued function $c_{\text{bank}} \, : \, I_1 \rightarrow I_2$, defined as
\begin{equation}
\label{bank_function}
c_{\text{bank}}(E_{\f_1}) =
\text{if} \, \left( \exists \, P \in \gamma_{\text{bank}} \ \text{such that} \ P[0] = E_{\f_1} \right) \ \text{return} \ P[1].
\end{equation}
Essentially, this function checks the first element of the tuples in $\gamma_{\text{bank}}$,
and returns the second element of the tuple whose first element is equal to $E_{\f_1}$.
Since $I_1$ is a closed interval in $\R$, we have that for all $E_{\f_1} \in I_1$, the
function $c_{\text{bank}}$ is well-defined. See Fig.~\ref{fig:bank_curve} for the
representation of the above curve of bank states in the resource diagram of the theory.
\begin{figure}[t!]
\center
\includegraphics[width=0.5\textwidth]{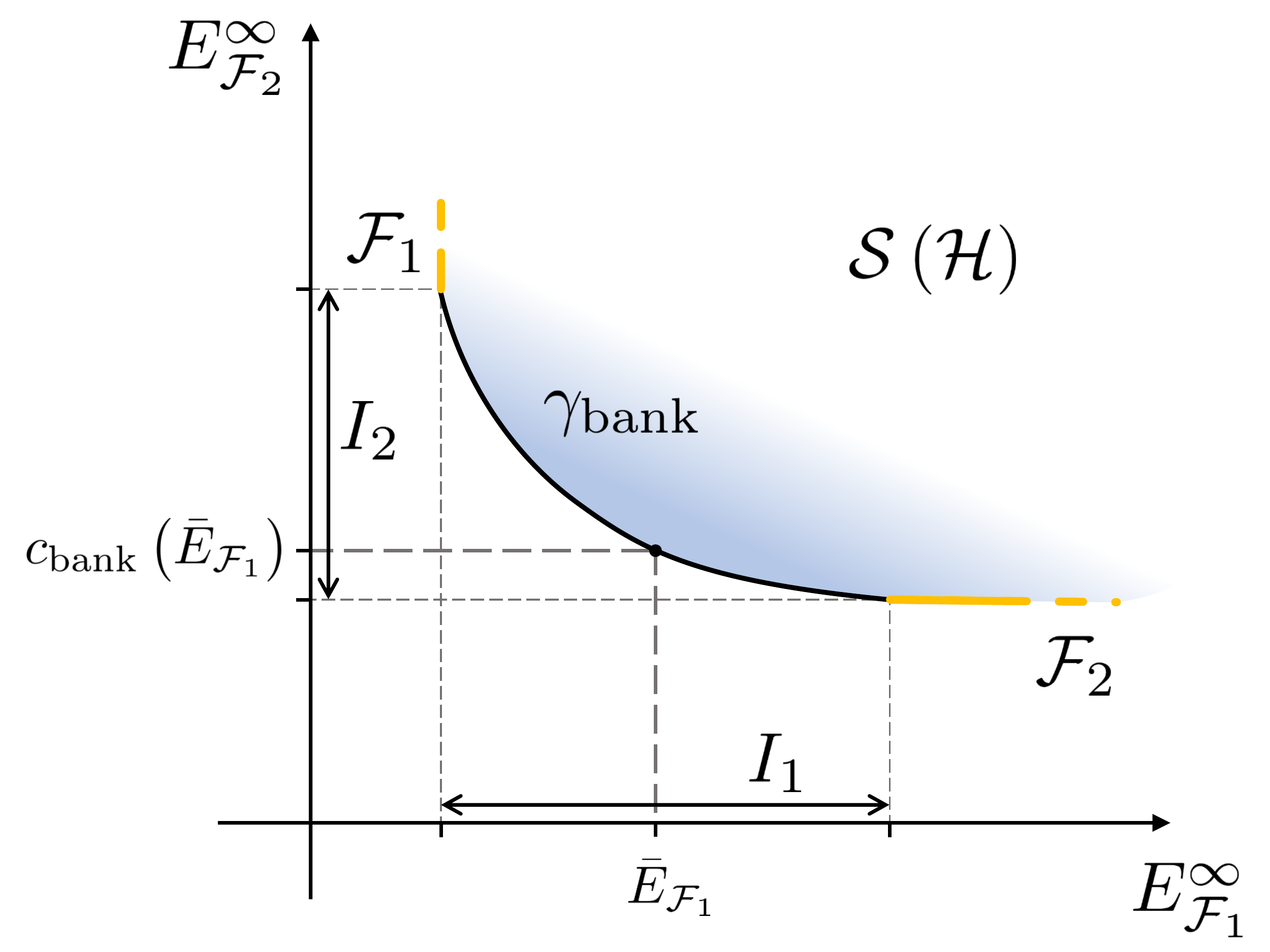}
\caption{We represent part of the state-space $\SH$ in the $E^{\infty}_{\f_1}$--$E^{\infty}_{\f_2}$ diagram. In
the figure, the green segment is the invariant set $\f_1$, the yellow one is $\f_2$, and the black
curve connecting these two segments is $\gamma_{\text{bank}}$, the curve of bank states of
the theory, see Eq.~\eqref{curve_bank}. On the $E^{\infty}_{\f_1}$-axis we highlight the interval $I_1$
defined in Eq.~\eqref{interval_1}, and similarly for the interval $I_2$ on the $E^{\infty}_{\f_2}$-axis. Furthermore,
the action of the function $c_{\text{bank}} : I_1 \rightarrow I_2$, defined in Eq.~\eqref{bank_function},
is shown for the input value $\bar{E}_{\f_1}$.}
\label{fig:bank_curve}
\end{figure}
\par
We will now prove the following two propositions, which assure that the monotone
$f_{\text{bank}}^{\bar{E}_{\f_1},\bar{E}_{\f_2}}$ of Eq.~\eqref{f3_monotone} satisfies
the property~\ref{item:B2}. This first proposition essentially tells us that the function
$c_{\text{bank}}$ is monotonic decreasing.
\begin{prop}
\label{monotone_curve}
For all $P_A, P_B \in \gamma_{\text{bank}}$, where $P_A = \left( E_{\f_1}^{(A)} , E_{\f_2}^{(A)} \right)$
and $P_B = \left( E_{\f_1}^{(B)} , E_{\f_2}^{(B)} \right)$, we have that
\begin{equation}
E_{\f_1}^{(A)} < E_{\f_1}^{(B)} \Leftrightarrow E_{\f_2}^{(A)} > E_{\f_2}^{(B)}.
\end{equation}
\end{prop}
\begin{proof}
We prove the propositions in a single direction, as the other follows in analogue manner.
Suppose that $E_{\f_1}^{(A)} < E_{\f_1}^{(B)}$, and consider the states $\rho_A$, $\rho_B \in
\f_{\text{bank}}$ such that $E^{\infty}_{\f_1}(\rho_A) = E_{\f_1}^{(A)}$, and $E^{\infty}_{\f_1}(\rho_B) = E_{\f_1}^{(B)}$.
Since $\rho_B$ belongs to the set of bank states, we have that one of the following conditions,
see Eq.~\eqref{set_f3}, has to be satisfied for all states $\sigma \in \SH$,
\begin{enumerate}
\item $E^{\infty}_{\f_1}(\sigma) > E^{\infty}_{\f_1}(\rho_B)$.
\item $E^{\infty}_{\f_2}(\sigma) > E^{\infty}_{\f_2}(\rho_B)$.
\item $E^{\infty}_{\f_1}(\sigma) = E^{\infty}_{\f_1}(\rho_B)$ and $E^{\infty}_{\f_2}(\sigma) = E^{\infty}_{\f_2}(\rho_B)$.
\end{enumerate}
Let us then take $\sigma = \rho_A$. In this case, options 1 and 3 are not possible, since they
contradict the hypothesis. Therefore, option 2 has to be valid, which implies that
$E^{\infty}_{\f_2}(\rho_A) > E^{\infty}_{\f_2}(\rho_B)$. In a similar manner, if $E_{\f_1}^{(A)} = E_{\f_1}^{(B)}$,
the only possible option for $\rho_B$ would have been $E^{\infty}_{\f_2}(\rho_A) = E^{\infty}_{\f_2}(\rho_B)$,
which concludes the proof.
\end{proof}
The second propositions tells us, instead, that the function $c_{\text{bank}}$ is convex.
\begin{prop}
\label{convex_curve}
For all $P_A, P_B \in \gamma_{\text{bank}}$, where $P_A = \left( E_{\f_1}^{(A)} , E_{\f_2}^{(A)} \right)$
and $P_B = \left( E_{\f_1}^{(B)} , E_{\f_2}^{(B)} \right)$, and for all $\lambda \in [0,1]$, there exists a
$P_C \in \gamma_{\text{bank}}$, where $P_C = \left( E_{\f_1}^{(C)} , E_{\f_2}^{(C)} \right)$, such that
\begin{align}
E_{\f_1}^{(C)} &= \lambda \, E_{\f_1}^{(A)} + ( 1 - \lambda ) \, E_{\f_1}^{(B)}, \\
E_{\f_2}^{(C)} &\leq \lambda \, E_{\f_2}^{(A)} + ( 1 - \lambda ) \, E_{\f_2}^{(B)}
\end{align}
\end{prop}
\begin{proof}
Let us consider, without losing in generality, that $E_{\f_1}^{(A)} < E_{\f_1}^{(B)}$, and take
$\rho_C \in \f_{\text{bank}}$ such that $E^{\infty}_{\f_1}(\rho_C) =  \lambda \, E_{\f_1}^{(A)} +
( 1 - \lambda ) \, E_{\f_1}^{(B)}$. This state always exists since $I_1$ is a closed interval
(and therefore is path-connected). Let us now define $\rho_A$, $\rho_B \in \f_{\text{bank}}$
such that $E^{\infty}_{\f_1}(\rho_A) = E_{\f_1}^{(A)}$, and $E^{\infty}_{\f_1}(\rho_B) = E_{\f_1}^{(B)}$. By
convexity of the regularised relative entropy distance $E^{\infty}_{\f_1}$, it follows that
\begin{equation}
E^{\infty}_{\f_1}(\rho_C) = \lambda \, E_{\f_1}^{(A)} + ( 1 - \lambda ) \, E_{\f_1}^{(B)} \geq
E^{\infty}_{\f_1} \left( \lambda \, \rho_A + ( 1 - \lambda ) \, \rho_B \right).
\end{equation}
Then, it is easy to show that
\begin{equation}
E^{\infty}_{\f_2}(\rho_C) \leq E^{\infty}_{\f_2} \left( \lambda \, \rho_A + ( 1 - \lambda ) \, \rho_B \right)
\leq \lambda \, E_{\f_2}^{(A)} + ( 1 - \lambda ) \, E_{\f_2}^{(B)},
\end{equation}
where the first inequality follows from Prop.~\ref{monotone_curve}, and the second
one from the convexity of $E^{\infty}_{\f_2}$. Since $\rho_C \in \f_{\text{bank}}$, the point
$P_C = \left(E^{\infty}_{\f_1}(\rho_C), E^{\infty}_{\f_2}(\rho_C) \right)$ is a point on the curve
$\gamma_{\text{bank}}$.
\end{proof}
It is easy to see that the above propositions imply that $c_{\text{bank}}$ is (strictly)
monotonic decreasing, and convex. Since this function is defined on the closed interval
$I_1 \in \R$, we have that $c_{\text{bank}}$ is continuous (except, maybe, at its endpoints).
Therefore, we can always define the monotone $f_{\text{bank}}^{\bar{E}_{\f_1},
\bar{E}_{\f_2}}$ of Eq.~\eqref{f3_monotone}, and it always satisfies condition~\ref{item:B2}.
Finally, it is worth noticing that all the results apply if one (or both) the monotones are of the
form of Eq.~\eqref{montone_average}, since they satisfy all the necessary properties, in
particular they are linear in both the tensor product and the admixture of states.
\section{Energy-entanglement interconversion protocol}
\label{protocol_example}
In this section we provide a protocol, based on the compression
theorems~\cite{schumacher_quantum_1995} known in quantum information theory, to
perform interconversion of energy and entanglement using two batteries and a bank,
see Sec.~\ref{setup_en_ent} for revising the set-up we use. In our protocol, we assume that
the bank is initially described by $n \gg 1$ copies of a generic state $\rho \in \s_{\p_0}$, where
$\p_0 > \frac{1}{2}$, see Eq.~\eqref{bank_subset_en_ent}, while the batteries $B_W$ and $B_E$
are initially in the states $\omega_W(k)$ and $\omega_E(h)$, respectively.
\par
Our first step consists in using the allowed operation $\chn_{\eta} \in \A_{\text{multi}}$, see
Eq.~\eqref{rev_ent_en}, with $\eta = \ket{\theta_1}\bra{\theta_1}$, to map the generic bank state
$\rho$ into
\begin{equation}
\rho_{\text{in}} = \p_0 \ket{\Psi_\text{singlet}}\bra{\Psi_\text{singlet}}
+ \left( 1 - \p_0 \right) \ket{\Psi_\text{triplet}^{(1)}}\bra{\Psi_\text{triplet}^{(1)}}.
\end{equation}
Thus, the bank system is now described by $n$ copies of the state $\rho_{\text{in}}$. Due to
the central limit theorem, we can well approximate the state of the bank with an ensemble of its
typical states, and in the following we will focus on the strongly typical ensemble,
\begin{equation}
\Pi_{\text{st}} = \frac{1}{d_{\text{st}}} \sum_{i=1}^{d_{\text{st}}} \pi_i
\left(
\ket{\Psi_\text{singlet}}\bra{\Psi_\text{singlet}}^{\otimes n \, \p_0}
\otimes
\ket{\Psi_\text{triplet}^{(1)}}\bra{\Psi_\text{triplet}^{(1)}}^{\otimes n \, ( 1 - \p_0)}
\right),
\end{equation}
where $d_{\text{st}} \approx 2^{n \h(\p_0)}$ is the number of states contained in the strongly typical
set, the $\pi_i$'s are the elements of the symmetric group acting on $n$ copies of the two-qubit system,
and $\h(\cdot)$ is the binary entropy. Then, we can use a unitary operation to re-order the states in
$\Pi_{\text{st}}$ so as to obtain
\begin{equation}
\Pi_{\text{st}}' =
\sigma_{\text{mm}}^{\otimes n \h(\p_0)}
\otimes
\ket{\Psi_\text{singlet}}\bra{\Psi_\text{singlet}}^{\otimes n \left( 1 - \h(\p_0) \right)},
\end{equation}
where $\sigma_{\text{mm}}$ is the separable state introduced in Eq.~\eqref{max_mix_sin_trip}.
It is easy to see that this transformation, while leaving the amount of entanglement in the bank
constant, $E_{\f_{\text{css}}}(\rho_{\text{in}}^{\otimes n}) = E_{\f_{\text{css}}}(\Pi_{\text{st}}')$,
might not preserve the average energy. For this reason, while transforming the bank we also
transform the energy battery, mapping $\omega_W(k)$ into $\omega_W(k+\Delta k)$ to keep
the energy fixed.
\par
We can now exchange some singlets with the entanglement battery. For example, we can perform
a swap between the bank and the battery, moving in this way an integer number $r$ of singlets from
the bank into the battery. This transformation maps the state of the bank into 
\begin{equation}
\Pi_{\text{st}}'' =
\sigma_{\text{mm}}^{\otimes n \h(\p_0) + r}
\otimes
\ket{\Psi_\text{singlet}}\bra{\Psi_\text{singlet}}^{\otimes n \left( 1 - \h(\p_0) \right) - r},
\end{equation}
and transforms the state of the entanglement battery from $\omega_E(h)$ into $\omega_E(h + r)$.
Furthermore, the transformation also modify the energy of the bank, so that we need to map the
state of the energy battery from $\omega_W(k+\Delta k)$ to $\omega_W(k+\Delta k')$.
It is then possible to map the state $\Pi_{\text{st}}''$ into
\begin{equation}
\Pi_{\text{st}}''' = \frac{1}{d'_{\text{st}}} \sum_{i=1}^{d'_{\text{st}}} \pi_i
\left(
\ket{\Psi_\text{singlet}}\bra{\Psi_\text{singlet}}^{\otimes n \, \p'_0}
\otimes
\ket{\Psi_\text{triplet}^{(1)}}\bra{\Psi_\text{triplet}^{(1)}}^{\otimes n \, ( 1 - \p'_0)}
\right),
\end{equation}
where $\p'_0$ is chosen in order to satisfy the equality 
\begin{equation}
n \h(\p_0) + r = n \h(\p'_0),
\end{equation}
and $d'_{\text{st}} = 2^{n \h(\p'_0)}$. The state $\Pi_{\text{st}}'''$ is the strongly typical
ensemble associated with $n$ copies of the state
\begin{equation}
\rho_{\text{fin}} = \p'_0 \ket{\Psi_\text{singlet}}\bra{\Psi_\text{singlet}}
+ \left( 1 - \p'_0 \right) \ket{\Psi_\text{triplet}^{(1)}}\bra{\Psi_\text{triplet}^{(1)}},
\end{equation}
where it is easy to show that the probability of occupation of the singlet is $\p'_0
\approx \p_0 - \frac{r}{n} \frac{1}{\log \frac{\p_0}{1-\p_0}}$ for $n \gg 1$. The
transformation mapping $\Pi_{\text{st}}''$ into $\Pi_{\text{st}}'''$ preserves the entanglement
of the bank, while changing its energy. Therefore, while acting on the bank we have to modify the
state of the energy battery as well, from $\omega_W(k+\Delta k')$ to $\omega_W(k+\Delta k'')$.
In this way, we have modified the bank system by mapping $n$ copies of $\rho_{\text{in}}$ into
$n$ copies of $\rho_{\text{fin}}$, and we kept entanglement and energy fixed on the global system
by modifying the states of the batteries. Notice that the protocol can be extended to the typical
ensembles by using a sub-linear ancillary system, and by considering corrections to the exchanged
energy and entanglement of order $O(\sqrt{n})$. 
\par
During the protocol, the bank has exchanged $r$ singlets with the battery $B_E$, so that the
gain in entanglement for this battery is
\begin{equation}
\Delta W_E = E_{\f_{\text{css}}} \left( \omega_E(h + r) \right)
- E_{\f_{\text{css}}} \left( \omega_E(h) \right) = r.
\end{equation}
In order to compute the amount of energy exchanged between the bank and the battery
$B_W$, we consider the difference in average energy between $\rho_{\text{in}}^{\otimes n}$
and $\rho_{\text{fin}}^{\otimes n}$. In this way, we find that the amount of energy exchanged is
\begin{equation}
\Delta W_W  = M_H \left( \omega_W(k+\Delta k'') \right) - M_H \left( \omega_W(k) \right) =
- \frac{\Delta E}{\log \frac{\p_0}{1-\p_0}} r, 
\end{equation}
that is, energy has been paid in order to gain entanglement during the process. The interconversion
relation between the two resources is given by
\begin{equation}
\label{int_rel_ent_ene_app}
\Delta W_W = - \frac{\Delta E}{\log \frac{\p_0}{1-\p_0}} \, \Delta W_E,
\end{equation}
and we only need to show that the bank state has changed in a negligible way with respect
to the related bank monotone. It is worth noting that, since the current theory satisfies
all the properties we have considered in the main text, the bank monotone coincides, modulo
a multiplicative constant, with the relative entropy distance from the set of states $\s_{\p_0}$
initially describing the bank.
\par
Indeed, it is easy to show that the relative entropy distance from this set is given by a linear
combination of the monotones $E_{\f_{\text{css}}}$ and $M_H$. For $\rho \in \s_1$ we find
that
\begin{equation}
\label{rel_ent_lin_ent_energy}
E_{\s_{\p_0}}(\rho) = \inf_{\sigma \in \s_{\p_0}} \re{\rho}{\sigma}
= \left( E_{\f_{\text{css}}}(\rho) - \bar{E}_{\f_{\text{css}}} \right) +
\frac{\log \frac{\p_0}{1-\p_0}}{\Delta E} \left( M_H(\rho) - \bar{M}_H \right),
\end{equation}
where we recall that $\bar{E}_{\f_{\text{css}}} = E_{\f_{\text{css}}}(\sigma)$ and $\bar{M}_H =
M_H(\sigma)$, for any state $\sigma \in \s_{\p_0}$. The linear coefficient in the rhs of
Eq.~\eqref{rel_ent_lin_ent_energy} is the (inverse) exchange rate that we find in the
interconversion relation, Eq.~\eqref{int_rel_ent_ene_app}. If we now consider the initial
and final state of the bank, and we study how much the state is changed by the above
protocol with respect to $E_{\s_{\p_0}}$, we find that
\begin{equation}
E_{\s_{\p_0}}(\rho_{\text{fin}}^{\otimes n}) - E_{\s_{\p_0}}(\rho_{\text{in}}^{\otimes n}) = O(n^{-1}),
\end{equation}
so that, when $n \rightarrow \infty$, we obtain that the state of the bank is only infinitesimally changed,
and can be used again to perform another resource interconversion with the same initial exchange rate.
\section{Proofs}
\subsection{Main results}
\label{main_results}
In the first part of this appendix we provide the proofs of the results presented in the main text.
We start with the proof of the following theorem, where it is shown that a multi-resource theory which
satisfies the asymptotic equivalence property of Def.~\ref{def:asympt_equivalence_multi} has a unique
quantifier for each of the resources present in the theory. This theorem is introduced in
Sec.~\ref{multi_rev_unique}.
\uniquemeas*
\begin{proof}
Let us prove that $f_1^{\infty}$ uniquely quantifies the amount of $1$-st resource contained in the main
system (the proof for the other $f_{i \neq 1}$'s is analogous). We prove the theorem by contradiction.
Suppose that there exists two monotones $f_1$ and $g_1$ satisfying the properties~\ref{item:M1} --
\ref{item:M7}, such that
\begin{enumerate}
\item $\exists \, \rho \in \SHs$, where $\rho \not\in \f_1$, for which $f_1^{\infty}(\rho) = g_1^{\infty}(\rho)$
(this is always possible by rescaling the monotone $g$).
\item $\exists \, \sigma \in \SHs$, where $\sigma \not\in \f_1$, for which $f_1^{\infty}(\sigma) \neq
g_1^{\infty}(\sigma)$ (that is, $f_1$ is not unique).
\end{enumerate}
Consider now the values of $f_1^{\infty}(\rho)$ and $f_1^{\infty}(\sigma)$. If these are equal, it is easy to
see, using the asymptotic equivalence property, that $f_1$ is unique. Suppose instead that they are not
equal. Then, there exists $n, k \in \N$\footnote{Where we assume that all physically meaningful
values of the $f_i^{\infty}$'s are in $\Q$, which we recall is dense in $\R$.} such that
\begin{equation} \label{relation_rho_sigma_f1}
n \, f_1^{\infty}(\rho) = k \, f_1^{\infty}(\sigma).
\end{equation}
\par
Let us consider the system together with the batteries $B_i$'s, initially in the state $\rho^{\otimes n}
\otimes \omega_1 \otimes \ldots \otimes \omega_m$. Then, we take the states $\omega'_i \in \SHbI$,
where $i = 1 , \ldots , m$, such that
\begin{align}
\label{unique_condition_i} 
f_i^{\infty}( \rho^{\otimes n} \otimes \omega_1 \otimes \ldots \otimes \omega_m )
&=
f_i^{\infty}( \gamma_n \otimes \omega'_1 \otimes \ldots \otimes \omega'_m )
\ , \ \forall \, i \in \left\{ 1 , \ldots , m \right\}, \\
f_j^{\infty}( \omega_i ) &= f_j^{\infty}( \omega'_i )
\ , \ \forall \, i, j \in \left\{ 1 , \ldots , m \right\}, \ i \neq j,
\end{align}
where $\gamma_n \in \f_1^{(n)}$. Due to the asymptotic equivalence property, the
conditions in Eq.~\eqref{unique_condition_i} imply that there exists a sequence of
maps $\left\{ \tilde{\chn}_N \right\}_N$ of the form of Eq.~\eqref{allowed_ancilla}
such that
\begin{equation} \label{rev_map}
\lim_{N \rightarrow \infty}
\left\|
\tilde{\chn}_N \left(
\left( \rho^{\otimes n} \otimes \omega_1 \otimes \ldots \otimes \omega_m  \right)^{\otimes N}
\right)
-
\left(\gamma_n \otimes \omega'_1 \otimes \ldots \otimes \omega'_m  \right)^{\otimes N}
\right\|_1 = 0,
\end{equation}
as well as another sequence of maps performing the reverse transformation. From
the asymptotic continuity of $g_1$, property~\ref{item:M7}, it then follows that
\begin{equation}
g_1 \left(
\tilde{\chn}_N \left(
\left( \rho^{\otimes n} \otimes \omega_1 \otimes \ldots \otimes \omega_m  \right)^{\otimes N}
\right)
\right)
=
g_1 \left( \left(\gamma_n \otimes \omega'_1 \otimes \ldots \otimes \omega'_m  \right)^{\otimes N} \right)
+
o(N).
\end{equation}
Let us consider the lhs of the above equation, and recall that the map $\tilde{\chn}_N$ is
obtained by applying an allowed operation to $N$ copies of the system together with a sub-linear
ancilla $\eta^{(A)}_N$, see Eq.~\eqref{allowed_ancilla}. For simplicity, let us refer to $\rho^{\otimes n}
\otimes \omega_1 \otimes \ldots \otimes \omega_m$ as $\Omega$ in the following chain of inequalities,
\begin{align}
g_1 \left( \tilde{\chn}_N \left( \Omega^{\otimes N} \right) \right)
&= g_1 \left( \Tr{A}{\chn_N \left( \Omega^{\otimes N} \otimes \eta^{(A)}_N \right)} \right)
\leq g_1 \left( \chn_N \left( \Omega^{\otimes N} \otimes \eta^{(A)}_N \right) \right)
\leq g_1 \left( \Omega^{\otimes N} \otimes \eta^{(A)}_N \right) \nonumber \\
&\leq g_1 \left( \Omega^{\otimes N} \right) + g_1 \left( \eta^{(A)}_N \right)
\leq g_1 \left( \Omega^{\otimes N} \right) + o(N)
\end{align}
where the first inequality follows from property~\ref{item:M4}, the second one from the monotonicity
of $g_1$ under allowed operations, the third one from the sub-additivity of $g_1$, property~\ref{item:M5},
and the last inequality from property~\ref{item:M6} and the fact that the ancilla is sub-linear in $N$. If
we now combine this equation with the previous one, we divide both sides by $N$, and we send it to
infinity, we obtain that the regularised version of $g_1$ is such that,
\begin{equation}
g_1^{\infty} \left( \rho^{\otimes n} \otimes \omega_1 \otimes \ldots \otimes \omega_m \right)
\geq
g_1^{\infty} \left( \gamma_n \otimes \omega'_1 \otimes \ldots \otimes \omega'_m  \right).
\end{equation}
By using the same argument for the sequence of maps performing the reverse transformation,
we find that the above equation needs to hold as an equality, that is,
\begin{equation}
g_1^{\infty} \left( \rho^{\otimes n} \otimes \omega_1 \otimes \ldots \otimes \omega_m \right)
=
g_1^{\infty} \left( \gamma_n \otimes \omega'_1 \otimes \ldots \otimes \omega'_m  \right).
\end{equation}
We can now separate each contribution to $g_1$ thanks to the property~\ref{item:M2}, use the
fact that the batteries $B_{i \neq 1}$'s are not changing their value of $g_1$, property~\ref{item:M1},
and the fact that the final state of the system does not contain any resource associated with
$g_1$, property~\ref{item:M3}. Then, we find that
\begin{equation} \label{g1_rho_battery}
n \, g_1^{\infty} \left( \rho \right) = g_1^{\infty} \left( \omega'_1 \right) - g_1^{\infty} \left( \omega_1 \right),
\end{equation}
where we have also used Lem.~\ref{lem:add_regularised_mon}. The same result follows for
$f_1$, so that we find that
\begin{equation} \label{f1_rho_battery}
n \, f_1^{\infty} \left( \rho \right) = f_1^{\infty} \left( \omega'_1 \right) - f_1^{\infty} \left( \omega_1 \right).
\end{equation}
\par
If we now consider Eqs.~\eqref{relation_rho_sigma_f1} and \eqref{f1_rho_battery}, we find that
\begin{equation}
k \, f_1^{\infty} \left( \sigma \right) = f_1^{\infty} \left( \omega'_1 \right) - f_1^{\infty} \left( \omega_1 \right).
\end{equation}
We can add to the above equation the term $f_1^{\infty} \left( \gamma_k \right)$, where $\gamma_k \in
\f_1^{(k)}$, since this term is equal to zero due to property~\ref{item:M3}. Then, we find
\begin{equation} \label{f1_sigma_cont}
k \, f_1^{\infty} \left( \sigma \right) + f_1^{\infty} \left( \omega_1 \right)
=
f_1^{\infty} \left( \gamma_{k} \right) + f_1^{\infty} \left( \omega'_1 \right).
\end{equation}
Now, we want to introduce the initial and final states of the batteries $B_{i\neq1}$'s, so as to be sure
that the transformation from $\sigma^{\otimes k}$ into $\gamma_k$ does not violate the conservation
of the other resources. Specifically, we introduce $\omega_{i}$, $\omega''_{i} \in \SHbI$ for $i \neq 1$,
such that
\begin{align}
f_i^{\infty} \left(
\sigma^{\otimes k} \otimes \omega_1 \otimes \omega_2 \otimes \ldots \otimes \omega_m
\right)
&=
f_i^{\infty} \left(
\gamma_k \otimes \omega'_1 \otimes \omega''_2 \otimes \ldots \otimes \omega''_m
\right)
\ , \ \forall \, i \in \left\{ 2 , \ldots , m \right\}, \label{fi_sigma_gamma} \\
f_1^{\infty}( \omega_i ) &= f_1^{\infty}( \omega''_i )
\ , \ \forall \, i \in \left\{ 2 , \ldots , m \right\}, \label{f1_Bi} \\
f_j^{\infty}( \omega_i ) &= f_j^{\infty}( \omega''_i )
\ , \ \forall \, i, j \in \left\{ 2 , \ldots , m \right\}, \ i \neq j.
\end{align}
Then, using the constraints of Eq.~\eqref{f1_Bi} over the states of the $B_{i\neq1}$'s batteries,
we can re-write Eq.~\eqref{f1_sigma_cont} as
\begin{equation}
k \, f_1^{\infty} \left( \sigma \right) + f_1^{\infty} \left( \omega_1 \right) + f_1^{\infty} \left( \omega_2 \right)
+ \ldots + f_1^{\infty} \left( \omega_m \right)
=
f_1^{\infty} \left( \gamma_{k} \right) + f_1^{\infty} \left( \omega'_1 \right) + f_1^{\infty} \left( \omega''_2 \right)
+ \ldots + f_1^{\infty} \left( \omega''_m \right).
\end{equation}
If we now use Lem.~\ref{lem:add_regularised_mon} and property~\ref{item:M1}, we find that
\begin{equation} \label{f1_sigma_gamma}
f_1^{\infty} \left(
\sigma^{\otimes k} \otimes \omega_1 \otimes \omega_2 \otimes \ldots \otimes \omega_m
\right)
=
f_1^{\infty} \left(
\gamma_k \otimes \omega'_1 \otimes \omega''_2 \otimes \ldots \otimes \omega''_m
\right)
\end{equation}
From Eqs.~\eqref{fi_sigma_gamma} and \eqref{f1_sigma_gamma} it follows, using the asymptotic
equivalence property, that there exists a sequence of maps $\left\{ \tilde{\chn}'_N \right\}_N$
such that
\begin{equation}
\lim_{N \rightarrow \infty}
\left\|
\tilde{\chn}'_N \left(
\left( \sigma^{\otimes k} \otimes \omega_1 \otimes \omega_2 \otimes \ldots \otimes \omega_m  \right)^{\otimes N}
\right)
-
\left( \gamma_k \otimes \omega'_1 \otimes \omega''_2 \otimes \ldots \otimes \omega''_m  \right)^{\otimes N}
\right\|_1 = 0,
\end{equation}
as well as a related sequence of maps performing the reverse transformation. Using the properties of
$g_1$, as we did before, we find that
\begin{equation} \label{g1_sigma_battery}
k \, g_1^{\infty} \left( \sigma \right) = g_1^{\infty} \left( \omega'_1 \right) - g_1^{\infty} \left( \omega_1 \right).
\end{equation}
Then, combining Eqs.~\eqref{g1_rho_battery} and \eqref{g1_sigma_battery}, we obtain that
\begin{equation}
n \, g_1^{\infty} \left( \rho \right)  = k \, g_1^{\infty} \left( \sigma \right).
\end{equation}
Finally, using Eq.~\eqref{relation_rho_sigma_f1} and the initial assumption on the state $\rho$,
we find that
\begin{equation}
f_1^{\infty} \left( \sigma \right)  = g_1^{\infty} \left( \sigma \right),
\end{equation}
which contradicts our initial assumption. Therefore, $f_1^{\infty}$ uniquely quantify the amount of
$1$-st resource contained in the main system.
\end{proof}
In the next theorem, first stated in Sec.~\ref{bank_interconvert}, we show that in the
presence of a bank two resources can always be exchanged one for the other, while the state of
the bank is only infinitesimally modified by the resource interconversion.
\interconvert*
\begin{proof}
Let us consider the resource interconversion of Eq.~\eqref{interconv_trasf}, where a global
operation is performed over bank and batteries, and the sole constraint over the bank system is
given by condition~\ref{item:X1}. As we discussed in Sec.~\ref{quant_res}, in order for the transformation
to happen, the conditions of Eq.~\eqref{eq:fi_condition_Ri} need to be satisfied for both monotones
$E_{\f_1}$ and $E_{\f_2}$, which in particular implies that the amount of resources exchanged with
the batteries is
\begin{equation}
\label{intconv_conds}
\Delta W_i = n \left( E_{\f_i}(\rho) - E_{\f_i}(\tilde{\rho}) \right) , \qquad i = 1, 2,
\end{equation}
where we have used property~\ref{item:F5b}. Furthermore, since $f_{\text{bank}}^{\bar{E}_{\f_1},\bar{E}_{\f_2}}$
is monotonic under the set of allowed operations, property~\ref{item:B7}, we find that
\begin{equation}
f_{\text{bank}}^{\bar{E}_{\f_1},\bar{E}_{\f_2}}(\rho^{\otimes n} \otimes \omega_1 \otimes \omega_2) =
f_{\text{bank}}^{\bar{E}_{\f_1},\bar{E}_{\f_2}}(\tilde{\rho}^{\otimes n} \otimes \omega'_1 \otimes \omega'_2).
\end{equation}
Then, since the global system is given by many copies of $\hil$, and since the bank monotone
is additive, property~\ref{item:B3}, we can separate the contribution given by bank and batteries.
Furthermore, from the definition of bank monotone, Eq.~\eqref{f3_monotone}, and the main
property of the batteries, condition~\ref{item:M1}, it follows that
\begin{equation}
\alpha \left( E_{\f_1}(\rho^{\otimes n}) + E_{\f_1}(\omega_1) \right)
+
\beta \left( E_{\f_2}(\rho^{\otimes n}) +E_{\f_2}(\omega_2) \right)
=
\alpha \left( E_{\f_1}(\tilde{\rho}^{\otimes n}) + E_{\f_1}(\omega'_1) \right)
+
\beta \left( E_{\f_2}(\tilde{\rho}^{\otimes n}) + E_{\f_2}(\omega'_2) \right).
\end{equation}
Now, if we re-order the terms in the above equation, and we use Eq.~\eqref{f3_monotone} again,
we obtain
\begin{equation}
f_{\text{bank}}^{\bar{E}_{\f_1},\bar{E}_{\f_2}}(\rho^{\otimes n})
-
f_{\text{bank}}^{\bar{E}_{\f_1},\bar{E}_{\f_2}}(\tilde{\rho}^{\otimes n})
=
\alpha \left( E_{\f_1}(\omega'_1) - E_{\f_1}(\omega_1) \right)
+
\beta \left( E_{\f_2}(\omega'_2) - E_{\f_2}(\omega_2) \right).
\end{equation}
If we use property~\ref{item:X1} together with the definitions of $\Delta W_1$ and $\Delta W_2$ given in
Eq.~\eqref{work_Ri}, we get that
\begin{equation}
\alpha \, \Delta W_1 = - \beta \, \Delta W_2 + \delta_n,
\end{equation}
where $\delta_n \rightarrow 0$ as $n$ tends to infinity. However, we are still left to show that, when
$n \rightarrow \infty$, the amount of resources exchanged by the batteries remains finite.
\par
Let us first recall that the way in which the monotone $f_{\text{bank}}^{\bar{E}_{\f_1}, \bar{E}_{\f_2}}$
is built implies that this monotone is tangent to the state-space, see property~\ref{item:B2} and
Fig.~\ref{fig:tangent_monotone}. As a result, we have that the curve of bank states, see
Eq.~\eqref{curve_bank} in appendix~\ref{convex_bound}, can be approximate, in the neighbourhood
of $\f_{\text{bank}}\left(\bar{E}_{\f_1},\bar{E}_{\f_2}\right)$, by a line. This implies that, if we take
the state $\tilde{\rho}$ in the set of bank states $\f_{\text{bank}}$, such that
\begin{equation}
E_{\f_1}(\tilde{\rho}) = E_{\f_1}(\rho) - \epsilon,
\end{equation}
where we recall $\rho \in \f_{\text{bank}}\left(\bar{E}_{\f_1},\bar{E}_{\f_2}\right)$, and $\epsilon \ll 1$,
we find that the value of the monotone $E_{\f_2}$ for this state is
\begin{equation}
E_{\f_1}(\tilde{\rho}) = E_{\f_2}(\rho) + \frac{\alpha}{\beta} \epsilon + O(\epsilon^2).
\end{equation}
Then, it is easy to see that, if we map $\rho$ into $\tilde{\rho}$ during the resource interconversion,
we obtain the following
\begin{equation}
\label{rate_of_convergence}
\Delta W_1 = n \, \epsilon \quad , \quad
\Delta W_2 = - n \, \frac{\alpha}{\beta} \epsilon + O(n \, \epsilon^2) \quad , \quad
\delta_n = O(n \, \epsilon^2),
\end{equation}
where the first two equations follow from Eq.~\eqref{intconv_conds}, while the last one is given
by Eq.~\eqref{condition_x1}. Thus, if we take $\epsilon \propto \frac{1}{n}$, and we send $n$ to
infinity, we get that the amount of resources $\Delta W_i$ exchanged during the transformations
are finite and their value is arbitrary, while the change in the bank monotone over the bank system
$\delta_n$ is infinitesimal. 
\end{proof}
The next theorem can be found in Sec.~\ref{bank_monotone}. The theorem states that, given a
multi-resource theory with a non-empty set of bank states, we can always build a single-resource theory
out of it, by extending the class of allowed operations with the possibility of adding ancillary systems described
by the bank states, see Def.~\ref{def:sing_res_constr}. In particular, we show that if the multi-resource
theory satisfies the asymptotic equivalence property, so does the single-resource theory with respect to the
bank monotone of Eq.~\eqref{f3_monotone}.
\singres*
\begin{proof}
{\bf (a)} We start the proof by showing that, for the single resource theory $\rt_{\mathrm{single}}$,
the second statement in Def.~\ref{def:asympt_equivalence_multi} implies the first one. In other words, we
want to show that for any two states $\rho, \sigma \in \SH$ which can be asymptotically mapped into one
another with the allowed operations $\A_{\mathrm{single}}$, the value of the bank monotone on the two
states is the same. Suppose there exists a sequence of operations $\left\{ \tilde{\chn}_N^{\mathrm{(s)}}
\right\}_N$ such that $\lim_{N \rightarrow \infty} \left\| \tilde{\chn}_N^{\mathrm{(s)}} (\rho^{\otimes N})
- \sigma^{\otimes N} \right\|_1 = 0$, where these maps are of the form
\begin{equation}
\label{allowed_ancilla_single}
\tilde{\chn}_N^{\mathrm{(s)}} (\cdot) = \Tr{A}{ \chn_N^{\mathrm{(s)}} ( \cdot \otimes \eta_N^{(A)} ) },
\end{equation}
with $\eta_N^{(A)} \in \SHn{o(N)}$ an arbitrary state of a sub-linear ancilla, and $\chn_N^{\mathrm{(s)}}$
an allowed operation for $\rt_{\mathrm{single}}$. Likewise, suppose there is a sequence of maps
that perform the reverse transformation. If we use the asymptotic continuity of the bank monotone,
property~\ref{item:B6}, it follows that
\begin{equation}
f_{\mathrm{bank}}^{\bar{E}_{\f_1}, \bar{E}_{\f_2}} \left( \tilde{\chn}_N^{\mathrm{(s)}} (\rho^{\otimes N}) \right)
=
f_{\mathrm{bank}}^{\bar{E}_{\f_1}, \bar{E}_{\f_2}} \left( \sigma^{\otimes N} \right) + o(N).
\end{equation}
Then, by using the properties~\ref{item:B1} -- \ref{item:B7} of the bank monotone, we can prove the
following chain of inequalities for the lhs of the above equation
\begin{align}
f_{\mathrm{bank}}^{\bar{E}_{\f_1}, \bar{E}_{\f_2}} \left( \tilde{\chn}_N^{\mathrm{(s)}} (\rho^{\otimes N}) \right)
&=
f_{\mathrm{bank}}^{\bar{E}_{\f_1}, \bar{E}_{\f_2}}
\left(
\Tr{A}{ \chn_N^{\mathrm{(s)}} ( \rho^{\otimes N} \otimes \eta_N^{(A)} ) }
\right)
\leq
f_{\mathrm{bank}}^{\bar{E}_{\f_1}, \bar{E}_{\f_2}}
\left(
\chn_N^{\mathrm{(s)}} ( \rho^{\otimes N} \otimes \eta_N^{(A)} )
\right) \nonumber \\
&\leq
f_{\mathrm{bank}}^{\bar{E}_{\f_1}, \bar{E}_{\f_2}}
\left(
\rho^{\otimes N} \otimes \eta_N^{(A)}
\right)
=
f_{\mathrm{bank}}^{\bar{E}_{\f_1}, \bar{E}_{\f_2}} \left( \rho^{\otimes N} \right)
+
f_{\mathrm{bank}}^{\bar{E}_{\f_1}, \bar{E}_{\f_2}} \left( \eta_N^{(A)} \right) \nonumber \\
&\leq
f_{\mathrm{bank}}^{\bar{E}_{\f_1}, \bar{E}_{\f_2}} \left( \rho^{\otimes N} \right)
+
o(N)
\end{align}
where the first inequality follows from monotonicity under partial trace, property~\ref{item:B4}, the
second one from monotonicity under the allowed operations $\A_{\mathrm{single}}$ (that we still
need to show), the equality follows from additivity, property~\ref{item:B3}, and the last inequality
from the sub-extensivity of the monotone, property~\ref{item:B5}. If we use the same argument for the
sequence of maps performing the reverse transformation, and we regularise the monotones by dividing
the equations by the number of copies $N$, and sending $N$ to infinity, we find that
\begin{equation}
f_{\mathrm{bank}}^{\bar{E}_{\f_1}, \bar{E}_{\f_2}} \left( \rho \right)
=
f_{\mathrm{bank}}^{\bar{E}_{\f_1}, \bar{E}_{\f_2}} \left( \sigma \right),
\end{equation}
which proves the asymptotic equivalence property in one direction.
\par
We still need to show that the bank monotone is monotonic under the allowed operations
$\A_{\mathrm{single}}$ of the single-resource theory. Recall that the most general of these
operations, Eq.~\eqref{sin_res_map}, is given by
\begin{equation}
\chn^{\mathrm{(s)}} ( \rho ) = \Tr{P^{(n)}}{\chn ( \rho \otimes \rho_P^{\otimes n} )},
\end{equation}
where $\chn \in \A_{\mathrm{multi}}$, and we add $n \in \N$ copies of the bank state $\rho_P
\in \f_{\mathrm{bank}}\left(\bar{E}_{\f_1}, \bar{E}_{\f_2}\right)$. Then, using the properties of the
bank monotone, we can show that
\begin{align}
f_{\mathrm{bank}}^{\bar{E}_{\f_1}, \bar{E}_{\f_2}} \left( \chn^{\mathrm{(s)}} ( \rho ) \right)
&=
f_{\mathrm{bank}}^{\bar{E}_{\f_1}, \bar{E}_{\f_2}} \left( \Tr{P^{(n)}}{\chn ( \rho \otimes \rho_P^{\otimes n} ) } \right)
\leq
f_{\mathrm{bank}}^{\bar{E}_{\f_1}, \bar{E}_{\f_2}} \left( \chn ( \rho \otimes \rho_P^{\otimes n} ) \right)
\nonumber \\
&\leq
f_{\mathrm{bank}}^{\bar{E}_{\f_1}, \bar{E}_{\f_2}} \left( \rho \otimes \rho_P^{\otimes n} \right) 
=
f_{\mathrm{bank}}^{\bar{E}_{\f_1}, \bar{E}_{\f_2}} \left( \rho \right)
+
f_{\mathrm{bank}}^{\bar{E}_{\f_1}, \bar{E}_{\f_2}} \left( \rho_P^{\otimes n} \right) \nonumber \\
&=
f_{\mathrm{bank}}^{\bar{E}_{\f_1}, \bar{E}_{\f_2}} \left( \rho \right),
\end{align}
where the first inequality follows from property~\ref{item:B4}, the second one from the monotonicity under
the allowed operations $\A_{\mathrm{multi}}$, property~\ref{item:B7}, and the last two equalities from additivity,
property~\ref{item:B3}, and the fact that the bank monotone is equal to zero over the bank states,
property~\ref{item:B1}, respectively.
\par
{\bf (b)} We now want to prove the other direction of the asymptotic equivalence property for the resource
theory $\rt_{\mathrm{single}}$, i.e., that the first statement in Def.~\ref{def:asympt_equivalence_multi}
implies the second one. In other words, we want to show that for all states $\rho$, $\sigma \in \SH$ such that
$f_{\mathrm{bank}}^{\bar{E}_{\f_1}, \bar{E}_{\f_2}} \left( \rho \right) = f_{\mathrm{bank}}^{\bar{E}_{\f_1}, \bar{E}_{\f_2}}
\left( \sigma \right)$, there exists a sequence of operations $\left\{ \tilde{\chn}^{(s)}_N \right\}_N$ of
the form given in Eq.~\eqref{allowed_ancilla_single}, mapping $N$ copies of $\rho$ into $N$ copies of $\sigma$,
where $N \rightarrow \infty$. Before proving this part of the theorem, we recall that, given the bank state
$\rho_P \in \f_{\mathrm{bank}} \left(\bar{E}_{\f_1},\bar{E}_{\f_2}\right)$, all other bank states $\tilde{\rho}_P \in
\f_{\mathrm{bank}}$ are such that, if $E_{\f_1} (\tilde{\rho}_P) = E_{\f_1} (\rho_P) + \delta$ with $\delta \ll 1$, then
\begin{equation}
\label{first_order_approx}
E_{\f_2} ( \tilde{\rho}_P ) = E_{\f_2} (\rho_P) - \frac{\alpha}{\beta} \, \delta + O(\delta^2),
\end{equation}
which follows from the fact that $f_{\mathrm{bank}}^{\bar{E}_{\f_1}, \bar{E}_{\f_2}} = 0$ parametrises the
line which is tangent to the state space and passes through the point $\left( \bar{E}_{\f_1}, \bar{E}_{\f_2}
\right)$, see appendix~\ref{convex_bound}.
\par
Given the two states $\rho, \sigma \in \SH$ with same value of the monotone $f_{\mathrm{bank}}
^{\bar{E}_{\f_1}, \bar{E}_{\f_2}}$, let us introduce the sequences of states $\left\{ \sigma_n \in
\SH \right\}_n$ and $\left\{ \tilde{\rho}_{P,n} \in \f_{\mathrm{bank}} \right\}_n$ such that, for $n \in \N$ big
enough, we have
\begin{align}
\label{constraint_sigma_1}
E_{\f_1} (\sigma_n) &= E_{\f_1} (\sigma) \\
\label{constraint_rho_tilde}
E_{\f_1} (\rho \otimes \rho_P^{\otimes n}) &= E_{\f_1} (\sigma_n \otimes ( \tilde{\rho}_{P,n} )^{\otimes n}), \\
\label{constraint_existence}
E_{\f_2}(\rho \otimes \rho_P^{\otimes n}) &= E_{\f_2}(\sigma_n \otimes ( \tilde{\rho}_{P,n} )^{\otimes n}),
\end{align}
where $\rho_P \in \f_{\mathrm{bank}}\left(\bar{E}_{\f_1},\bar{E}_{\f_2}\right)$. From the above equations,
and from the additivity of $E_{\f_1}$, which follows from property~\ref{item:F5b}, we obtain that
\begin{equation}
E_{\f_1}(\tilde{\rho}_{P,n}) = E_{\f_1}(\rho_P) + \frac{1}{n} \left( E_{\f_1}(\rho) - E_{\f_1}(\sigma) \right).
\end{equation}
Notice that, for $n \rightarrow \infty$, we have that $\frac{1}{n} \left( E_{\f_1}(\rho) - E_{\f_1}(\sigma) \right) \rightarrow
0$, and therefore, for $n$ sufficiently big, it follows from Eq.~\eqref{first_order_approx} that
\begin{equation}
\label{first_order_passive}
E_{\f_2}(\tilde{\rho}_{P,n}) = E_{\f_2}(\rho_P) - \frac{\alpha}{\beta}
\, \frac{1}{n} \left( E_{\f_1}(\rho) - E_{\f_1}(\sigma) \right) + O(n^{-2}).
\end{equation}
If we now combine Eq.~\eqref{constraint_existence} and \eqref{first_order_passive} together, we use the additivity of
$E_{\f_2}$, and we use the fact that $\rho$ and $\sigma$ have the same value of the bank monotone, we obtain the
following
\begin{equation}
\label{constraint_sigma_2}
E_{\f_2}(\sigma_n) = E_{\f_2}(\sigma) + O(n^{-1}).
\end{equation}
\par
Let us now focus on the operations mapping $\rho$ into $\sigma$. We do this in two steps. First, we
use the fact that the theory $\rt_{\mathrm{multi}}$ satisfies asymptotic equivalence, and we consider the
Eqs.~\eqref{constraint_rho_tilde} and \eqref{constraint_existence}. These equations imply that, for all
$n \in \N$, there exists of a sequence of maps $\left\{ \tilde{\chn}_{N,n} \right\}_N$ such that
\begin{equation}
\label{asy_eq_mult}
\lim_{N \rightarrow \infty} \left\|
\tilde{\chn}_{N,n} \left(
\left( \rho \otimes \rho_P^{\otimes n} \right)^{\otimes N}
\right) 
-
\left( \sigma_n \otimes ( \tilde{\rho}_{P,n} )^{\otimes n} \right)^{\otimes N}
\right\|_1
=
0.
\end{equation}
As per definition of asymptotic equivalence, the maps $\tilde{\chn}_{N,n} \ : \ \SHn{N (n+1)}
\rightarrow \SHn{N (n+1)}$ are of the form
\begin{equation}
\tilde{\chn}_{N,n} ( \cdot ) = \Tr{A}{\chn_{N,n} \left( \cdot \otimes \eta^{(A)}_N  \right)}
\end{equation}
where the map $\chn_{N,n}$ is an allowed operation of $\rt_{\mathrm{multi}}$ acting on system and ancilla,
and the state of the ancilla is $\eta^{(A)}_N \in \mathcal{S} \left( \left( \hil^{\otimes n+1} \right)^{\otimes f(N)}
\right)$, where $f(N) = o(N)$. Notice that, in particular, we can take $n$ to be a monotonic function of $N$, $n = g(N)$,
such that $\lim_{N \rightarrow \infty} g(N) = \infty$ and $f(N) g(N) = o(N)$. For example, if $f(N) \propto N^{1/2}$,
we can chose $g(N) \propto N^{1/4}$, so that their product is $N^{3/4} = o(N)$.
\par
We can now define the sequence of maps $\left\{ \tilde{\chn}^{\mathrm{(s)}}_{N} \right\}_N$ acting on
$\SHn{N}$. These maps are defined as
\begin{equation}
\tilde{\chn}^{\mathrm{(s)}}_{N} (\rho^{\otimes N}) =
\Tr{P}{\tilde{\chn}_{N,g(N)} \left(
 \rho^{\otimes N} \otimes \rho_P^{\otimes N g(N)}
\right)},
\end{equation}
where we are tracing out the part of the system which was initially in the state $\rho_P^{\otimes N g(N)}$.
It is interesting to notice that this system is super-linear in the number of copies $N$ of $\rho$, a condition
that seems to be necessary to achieve the conversion, see Ref.~\cite{brandao_resource_2013} for an
example in thermodynamics. We can re-write these maps as
\begin{equation}
\label{single_sequence_ancilla}
\tilde{\chn}^{\mathrm{(s)}}_{N} (\rho^{\otimes N}) =
\Tr{A}{\chn^{\mathrm{(s)}}_{N} \left(
 \rho^{\otimes N} \otimes \eta^{(A)}_N
\right)},
\end{equation}
where we recall that the ancillary system still lives on a sub-linear number of copies of $\hil$, due
to our choice of the function $g(N)$, and the operation $\chn^{\mathrm{(s)}}_{N}$ is an allowed
operations for the theory $\rt_{\mathrm{single}}$ -- compare it with Eq.~\eqref{sin_res_map} --
defined as
\begin{equation}
\chn^{\mathrm{(s)}}_{N} ( \cdot ) = \Tr{P}{ \chn_{N,g(N)} \left( \cdot \otimes \rho_P^{\otimes N g(N)} \right) }.
\end{equation}
If we now use Eq.~\eqref{asy_eq_mult} together with the monotonicity of the trace distance under partial tracing, we find that
\begin{equation}
\lim_{N \rightarrow \infty} \left\|
\tilde{\chn}^{\mathrm{(s)}}_{N} (\rho^{\otimes N})
-
\left( \sigma_{g(N)} \right)^{\otimes N}
\right\|_1
=
0.
\end{equation}
\par
To conclude the proof, we notice that the sequence of states $\left\{ \sigma_{g(N)} \right\}_N$ does not
need to converge to $\sigma$ with respect to the trace distance. However, if we consider the regularisation
of the $E_{\f_i}$'s on these states, we find that
\begin{equation}
\lim_{N \rightarrow \infty} \frac{1}{N} E_{\f_i}( \sigma_{g(N)}^{\otimes N} ) = E_{\f_i}(\sigma)
, \quad i = 1,2,
\end{equation}
which follows from Eqs.~\eqref{constraint_sigma_1} and \eqref{constraint_sigma_2}.
Then, we can use the asymptotic equivalence of $\rt_{\text{multi}}$, which tells us that there
exists a second sequence of allowed operations, and a sub-linear ancilla, such that we
can asymptotically transform the state of the system into $\sigma$. This concludes the proof.
\end{proof}
The following corollary is stated in Sec.~\ref{bank_monotone}, and it shows that the bank monotone
introduced in Eq.~\eqref{f3_monotone} coincides with the relative entropy distance from the set of bank
states $\f_{\text{bank}}\left(\bar{E}_{\f_1}, \bar{E}_{\f_2}\right)$.
\bankmonrelent*
\begin{proof}
We first notice that Thm.~\ref{unique_f3} promises us that, under the current assumptions over the
theory $\rt_{\text{multi}}$, we can construct a single-resource theory $\rt_{\text{single}}$ with
allowed operations $\A_{\text{single}}$ as in Def.~\ref{def:sing_res_constr}, which satisfies asymptotic equivalence
with respect to the bank monotone $f_{\text{bank}}^{\bar{E}_{\f_1}, \bar{E}_{\f_2}}$. Furthermore, since
this monotone satisfies the properties~\ref{item:SM1} -- \ref{item:SM6} listed in appendix~\ref{rev_theory_sing},
we can use Thm.~\ref{thm:reversible_asympt_equiv}
in the same appendix to prove that this single resource theory is reversible. If we then use the results of
Ref.~\cite{horodecki_quantumness_2012}, we obtain that this monotone is the unique measure of resource for
the theory $\rt_{\text{single}}$.
\par
What we need to show in this proof is that, actually, both the bank monotone defined in
Eq.~\eqref{f3_monotone} and the relative entropy distance from the set of bank states
$\f_{\text{bank}}\left(\bar{E}_{\f_1}, \bar{E}_{\f_2}\right)$ satisfy the properties
from~\ref{item:SM1} to \ref{item:SM6}, and therefore by uniqueness these two functions
need to coincide (modulo a multiplicative constant). That the bank monotone satisfies these
properties is easy to show. Indeed, its monotonicity under the class of operations $\A_{\text{single}}$,
property~\ref{item:SM1}, is proved in part {\bf (a)} of Thm.~\ref{unique_f3}. Furthermore,
all other properties directly follow from property~\ref{item:B1} and the ones listed in
Prop.~\ref{prop:properties_bank_mon}.
\par
Showing that the relative entropy distance from the set of states $\f_{\text{bank}}\left(\bar{E}_{\f_1},
\bar{E}_{\f_2}\right)$ satisfies the same properties is not difficult either. First, we recall that the invariant
sets of the theory, $\f_1$ and $\f_2$, satisfy the properties~\ref{item:F1}, \ref{item:F2}, \ref{item:F3}
and~\ref{item:F5b} by hypothesis. This in turn implies that the subset of bank states under consideration
satisfies properties~\ref{item:F1}, \ref{item:F2} and~\ref{item:F5b}, as it follows from the Props.~\ref{convex_f3}
and \ref{additive_f3} in appendix~\ref{additional}. That the subset $\f_{\text{bank}}\left(\bar{E}_{\f_1},
\bar{E}_{\f_2}\right)$ contains a full-rank state, property~\ref{item:F3}, is an hypothesis of this corollary.
\par
With the help of the above properties we can show that the relative entropy distance from
$\f_{\text{bank}}\left(\bar{E}_{\f_1},\bar{E}_{\f_2}\right)$ satisfies the same properties
of the bank monotone. That this relative entropy is monotonic under the set of operations
$\A_{\text{single}}$, property~\ref{item:SM1}, is shown in Prop.~\ref{monotonicity_passive}.
Furthermore, property~\ref{item:SM2} follows from the definition of relative entropy distance,
see Eq.~\eqref{rel_entr_dist}, while property~\ref{item:SM3} follows from the fact that
$\f_{\text{bank}}\left(\bar{E}_{\f_1}, \bar{E}_{\f_2}\right)$ satisfies the properties~\ref{item:F1},
\ref{item:F2}, and \ref{item:F3}. The properties~\ref{item:SM4} and \ref{item:SM5} follow from
the additivity of the set of bank states, property~\ref{item:F5b}. Finally, the fact that the monotone
scales sub-extensively, property~\ref{item:SM6}, is a consequences of the additivity of the set of bank states, as well
as of the fact that a full-rank state is contained in this set, properties~\ref{item:F5b} and
\ref{item:F3}, respectively.
\end{proof}
\subsection{Technical results}
\label{additional}
In this section we provide some minor results that are used to prove some of the main theorems
in the paper. In particular, the next proposition is used in Sec.~\ref{multi_rev_unique}, together with
Thm.~\ref{thm:reversible_multi}, to show that a multi-resource theory satisfying asymptotic
equivalence with respect to the relative entropy distances from its invariant sets has unique
resource quantifiers. This proposition is already known in the literature, see the references
inside the proof.
\relentproperties*
\begin{proof}
Let us first show that the relative entropy distance $E_{\f_i}$ is a monotone for the multi-resource
theory $\rt_{\text{multi}}$, and that its regularisation is well-defined. These are necessary assumptions
we have made in Def.~\ref{def:asympt_equivalence_multi}. The fact that $E_{\f_i}$ is monotonic under
the class of allowed operations $\A_{\text{multi}}$, and that in particular it is monotonic under the allowed
operations in $\A_i$, follows from the argument provided in the last paragraph of Sec.~\ref{single_resource}, and
from the fact that $\A_{\text{multi}}$ is obtained from the intersection of all the other classes of allowed
operations, see Eq.~\eqref{all_ops_multi}. Furthermore, that the regularisation of $E_{\f_i}$ exists follows
from the properties~\ref{item:F3} and \ref{item:F4}. In fact, for all $\rho \in \SH$, we have that
\begin{equation}
\label{regularisation_proof}
\frac{1}{n} E_{\f_i}(\rho^{\otimes n}) = \frac{1}{n} \inf_{\gamma_n \in \f^{(n)}} \re{\rho^{\otimes n}}{\gamma_n}
\leq \frac{1}{n} \inf_{\gamma \in \f} \re{\rho^{\otimes n}}{\gamma^{\otimes n}} =
\inf_{\gamma \in \f} \re{\rho}{\gamma} \leq \re{\rho}{\gamma_{\text{full-rank}}}
\end{equation}
where the first inequality follows from the fact that the invariant sets are closed under
tensor product, property~\ref{item:F4}, and the second inequality from the fact that they
contain at least one full-rank state $\gamma_{\text{full-rank}}$, property~\ref{item:F3}.
Since the rhs of Eq.~\eqref{regularisation_proof} is finite, and independent of $n$, we
have that the regularisation of the $E_{\f_i}$'s is well-defined.
\par
In order for the monotone to satisfy the property~\ref{item:M1}, we can simply choose the states of the battery
$B_i$ to have a fixed value of the monotones $E_{\f_{j \neq i}}$, for all $j \in \left\{1, \ldots, m \right\}$.
Property~\ref{item:M2}, instead, follows from the fact that we want the batteries to be independent from
each other, so as to address them individually. As a result, we choose the global invariant sets to be of the
form $\f_i = \f_{i,S} \otimes \f_{i,B_1} \otimes \ldots \otimes \f_{i,B_m}$, where $i = 1, \ldots, m$, the main
system is $S$, and the $B_i$'s refer to the batteries. This implies that the relative entropy
distances from these sets are additive over system and batteries. However, it is still possible for $\f_i^{\otimes
n}$ to be a proper subset of $\f_i^{(n)}$, since on the main systems or batteries we do not ask any additivity
property. The validity of property~\ref{item:M3} for $E_{\f_i}$ follows straightforwardly from the definition
of relative entropy distance, see Eq.~\eqref{rel_entr_dist}. That $E_{\f_i}$ satisfies property~\ref{item:M4}
follows from property~\ref{item:F5}, since for all $\rho_n \in \SHn{n}$ we have that
\begin{align}
E_{\f_i}(\Tr{k}{\rho_n})
&= \inf_{\gamma_{n-k} \in \f_i^{(n-k)}} \re{\Tr{k}{\rho_n}}{\gamma_{n-k}}
\leq \inf_{\gamma_n \in \f_i^{(n)}} \re{\Tr{k}{\rho_n}}{\Tr{k}{\gamma_n}} \nonumber \\
&\leq \inf_{\gamma_n \in \f_i^{(n)}} \re{\rho_n}{\gamma_n}
= E_{\f_i}(\rho_n),
\end{align}
where the first inequality follows from property~\ref{item:F5}, and the second one from the monotonicity
of the relative entropy under CPTP maps. The monotones $E_{\f_i}$'s are also sub-additive,
property~\ref{item:M5}, since for any two states $\rho_n \in \SHn{n}$ and $\rho_k \in \SHn{k}$ we
have that
\begin{align}
E_{\f_i}(\rho_n \otimes \rho_k)
&= \inf_{\gamma_{n+k} \in \f_i^{(n+k)}} \re{\rho_n \otimes \rho_k}{\gamma_{n+k}}
\leq \inf_{\gamma_n \in \f_i^{(n)},\gamma_k \in \f_i^{(k)}}
\re{\rho_n \otimes \rho_k}{\gamma_n \otimes \gamma_k} \nonumber \\
&= \inf_{\gamma_n \in \f_i^{(n)}} \re{\rho_n}{\gamma_n} + \inf_{\gamma_k \in \f_i^{(k)}} \re{\rho_k}{\gamma_k}
= E_{\f_i}(\rho_n) + E_{\f_i}(\rho_k),
\end{align}
where the inequality follows from property~\ref{item:F4} of the set $\f_i$. Property~\ref{item:M6} for
the relative entropy distance $E_{\f_i}$ follows from similar considerations to the one presented in
Eq.~\eqref{regularisation_proof}. In fact, we have that for all $\rho_n \in \SHn{n}$,
\begin{align}
E_{\f_i}(\rho_n)
&= \inf_{\gamma_n \in \f_i^{(n)}} \re{\rho_n}{\gamma_n}
\leq \re{\rho_n}{\gamma_{\text{full-rank}}^{\otimes n}}
= - S(\rho_n) - \tr{ \rho_n \log \gamma_{\text{full-rank}}^{\otimes n} } \nonumber \\
&\leq - \tr{ \rho_n \log \gamma_{\text{full-rank}}^{\otimes n} }
\leq n \log \lambda_{\text{min}}^{-1},
\end{align}
where the first inequality follows from the fact that $\f_i$ contains a full-rank state, property~\ref{item:F3},
the second one from the fact that the von Neumann entropy is non-negative, and the last one from the
fact that the optimal case is obtained when $\rho_n$ is given by $n$ copies of the pure state associated
with the minimum eigenvalue $\lambda_{\text{min}}$ of the full-rank state $\gamma_{\text{full-rank}}$.
Finally, in Ref.~\cite{synak-radtke_asymptotic_2006}, Lem.~1, it was shown that the relative entropy
distance from a set $\f$ satisfying properties~\ref{item:F1}, \ref{item:F2}, and \ref{item:F3} is asymptotic
continuous. In the proof, it was required the set $\f$ to contain the maximally-mixed state. However, as
it was noticed in Ref.~\cite{brandao_generalization_2010}, Lem.~C.3, one simply needs $\f$ to contain
a full-rank state. Thus, under the above properties on the free set, we have that $E_{\f_i}$ satisfies the
property~\ref{item:M7}.
\end{proof}
The next proposition collects the properties of the bank monotone defined in Eq.~\eqref{f3_monotone}.
\bankproperties*
\begin{proof}
Most of the properties listed above follows straightforwardly from the ones of the invariant sets
$\f_i$'s. Here, we only focus on property~\ref{item:B4}, stating that
\begin{equation}
f_{\text{bank}}^{\bar{E}_{\f_1},\bar{E}_{\f_2}}(\Tr{k}{\rho_n})
\leq
f_{\text{bank}}^{\bar{E}_{\f_1},\bar{E}_{\f_2}}(\rho_n)
, \qquad
\forall \, n, k \in \N \ , \ k < n \ , \ \forall \, \rho_n \in \SHn{n}.
\end{equation}
In order to prove the above property, we make use of Lem.~\ref{f_i_inequality} and of the definition
of bank monotone, see Eq.~\eqref{f3_monotone}. First, let us divide the $n$ copies of the system into
two sets, so that $\hil^{\otimes n} = \hil^{\otimes k} \otimes \hil^{\otimes n-k}$, and in the following
equation we refer to $S_1$ as the system described by the first $k$ copies, and to $S_2$ as the
system described by the last $n-k$ copies. In particular, $\rho_{S_1} = \Tr{n-k}{\rho_n} \in \SHn{k}$,
and $\rho_{S_2} = \Tr{k}{\rho_n} \in \SHn{n-k}$. Then, we have the following chain of inequalities
\begin{align}
f_{\text{bank}}^{\bar{E}_{\f_1},\bar{E}_{\f_2}}(\rho_n)
&= \alpha \left( E_{\f_1}(\rho_n) - n \, \bar{E}_{\f_1} \right)
+
\beta \left( E_{\f_2}(\rho_n) - n \, \bar{E}_{\f_2} \right) \nonumber \\
&\geq \alpha \left( E_{\f_1}(\rho_{S_1}) + E_{\f_1}(\rho_{S_2}) - n \, \bar{E}_{\f_1} \right)
+ \beta \left( E_{\f_2}(\rho_{S_1}) + E_{\f_2}(\rho_{S_2}) - n \, \bar{E}_{\f_2} \right) \nonumber \\
&=
\alpha \left( E_{\f_1}(\rho_{S_1}) - k \, \bar{E}_{\f_1} \right)
+
\beta \left( E_{\f_2}(\rho_{S_1}) - k \, \bar{E}_{\f_2} \right) \nonumber \\
&+
\alpha \left( E_{\f_1}(\rho_{S_2}) - (n-k) \bar{E}_{\f_1} \right)
+
\beta \left( E_{\f_2}(\rho_{S_2}) - (n-k) \bar{E}_{\f_2} \right) \nonumber \\
&= f_{\text{bank}}^{\bar{E}_{\f_1},\bar{E}_{\f_2}}(\rho_{S_1})
+
f_{\text{bank}}^{\bar{E}_{\f_1},\bar{E}_{\f_2}}(\rho_{S_2})
\geq
f_{\text{bank}}^{\bar{E}_{\f_1},\bar{E}_{\f_2}}(\rho_{S_2})
=
f_{\text{bank}}^{\bar{E}_{\f_1},\bar{E}_{\f_2}}(\Tr{k}{\rho_n}),
\end{align}
where the first inequality follows from Lem.~\ref{f_i_inequality}, and the second one from the
fact that the bank monotone is non-negative, which itself follows from properties~\ref{item:B1}
and \ref{item:B2}.
\end{proof}
The following proposition is used in Sec.~\ref{average_non_increasing} to show that single-resource
theories whose class of allowed operations does not increase the average value of a given observable admit
a monotone that is asymptotic continuous, see property~\ref{item:M7}.
\begin{prop}
\label{average_asymp_cont}
Consider an Hilbert space $\hil$ with dimension $d$, an Hermitian operator $A \in \BH$, and
the function $M_A \, : \, \SH \rightarrow \R$ defined as
\begin{equation}
M_A(\rho) = \tr{A \rho} - a_0,
\end{equation}
where $\rho \in \SH$ is an element of the state-space, and $a_0$ is the minimum eigenvalue of $A$.
When $n$ copies of the Hilbert space are considered, $\hil_n = \otimes_{i=1}^n \hil^{(i)}$, the above
operator is extended as $A_n = \sum_{i=1}^n A^{(i)}$, where $A^{(i)} \in \BH$ acts on the $i$-th copy
of the Hilbert space. Then, the function $M_A$ is asymptotic continuous.
\end{prop}
\begin{proof}
Consider two states $\rho_n$, $\sigma_n \in \mathcal{S}(\hil^{\otimes n})$, such that $\left\|
\rho_n - \sigma_n \right\|_1 \rightarrow 0$ for $n \rightarrow \infty$. We are interested in the
difference between the value of the function $M_A$ evaluated on $\rho_n$ and $\sigma_n$.
By definition,
\begin{equation}
\left| M_A(\rho_n) - M_A(\sigma_n) \right| = \left| \tr{ \left( \rho_n - \sigma_n \right) A_n } \right|.
\end{equation}
Now, we can diagonalise the operator $\rho_n - \sigma_n = \sum_{\lambda} \lambda
\ket{\psi_{\lambda}} \bra{\psi_{\lambda}}$. Then, we find
\begin{equation}
\left| \tr{ \left( \rho_n - \sigma_n \right) A_n } \right| =
\left| \sum_{\lambda} \lambda \bra{\lambda} A_n \ket{\lambda} \right| \leq
\sum_{\lambda} \left| \lambda \right| \left| \bra{\lambda} A_n \ket{\lambda} \right| \leq
\sum_{\lambda} \left| \lambda \right| \left\| A_n \right\|_{\infty},
\end{equation}
where we are using the operator norm $\| O \|_{\infty} = \sup_{\ket{\psi} \in \hil}
\frac{\| O \ket{\psi} \|}{\| \ket{\psi} \|}$, and the last inequality straightforwardly follows
from the definition of operator norm. Then, due to the way in which we have defined
$A_n$, it is easy to show that $\| A_n \|_{\infty} = n \, \| A \|_{\infty}$, and therefore
\begin{equation}
\sum_{\lambda} \left| \lambda \right| \left\| A_n \right\|_{\infty} =
n \left\| A \right\|_{\infty} \sum_{\lambda} \left| \lambda \right| =
n \left\| A \right\|_{\infty} \left\| \rho_n - \sigma_n \right\|_1.
\end{equation}
Finally, notice that $\dim \hil_n = d^n$, where $d$ is fixed by the initial choice of $\hil$.
Then, we have,
\begin{equation}
\left| M_A(\rho_n) - M_A(\sigma_n) \right| \leq
n \, \log d \, \left\| A \right\|_{\infty} \left\| \rho_n - \sigma_n \right\|_1.
\end{equation}
If we now divide by $n$ both side of the inequality, we get that
\begin{equation}
\frac{\left| M_A(\rho_n) - M_A(\sigma_n) \right|}{n}
\leq
\log d \, \left\| A \right\|_{\infty} \left\| \rho_n - \sigma_n \right\|_1,
\end{equation}
and if we send $n \rightarrow \infty$, we obtain that $\frac{1}{n} \,
\left| M_A(\rho_n) - M_A(\sigma_n) \right| \rightarrow 0$,
which proves the theorem.
\end{proof}
\par
The following lemma states that, when the sets $\f_i$'s are such that $\f_i^{(n)} = \f_i^{\otimes n}$
for all $n \in \N$, property~\ref{item:F5b}, the relative entropy distances from these sets are super-additive. This
lemma is used in Prop.~\ref{additive_f3} and Thm.~\ref{unique_f3}.
\begin{lem}
\label{f_i_inequality}
Consider a state $\rho_{S_1,S_2} \in \SHn{2}$, and suppose that the sets $\f_1$ and $\f_2$ satisfy
property~\ref{item:F5b}, that is, $\f_i^{(n)} = \f_i^{\otimes n}$ for all $n \in \N$, $i = 1,2$.
Then, the relative entropy distances from these sets, $E_{\f_1}$ and $E_{\f_2}$, are such that
\begin{equation}
\label{eq:fi_ineq}
E_{\f_i}(\rho_{S_1,S_2}) \geq E_{\f_i}(\rho_{S_1}) + E_{\f_i}(\rho_{S_2}) \ , \ i = 1,2,
\end{equation}
where $\rho_{S_1} = \Tr{S_2}{\rho_{S_1,S_2}}$, and similarly $\rho_{S_2} = \Tr{S_1}{\rho_{S_1,S_2}}$.
Furthermore, the above inequality is saturated if and only if $\rho_{S_1,S_2} = \rho_{S_1} \otimes
\rho_{S_2}$. The result extends trivially to the case in which $n > 2$ copies of the system are considered.
\end{lem}
\begin{proof}
Let us consider the monotone $E_{\f_1}$, as the following argument can be equally applied to $E_{\f_2}$.
By definition of relative entropy distance, we have that
\begin{equation}
\label{part_1_fi_ineq}
E_{\f_1} (\rho_{S_1,S_2}) = \inf_{\sigma_{S_1,S_2} \in \f_1^{(2)}} D( \rho_{S_1,S_2} \| \sigma_{S_1,S_2} )
= - S(\rho_{S_1,S_2}) + \inf_{\sigma_{S_1,S_2} \in \f_1^{(2)}} \left( - \tr{\rho_{S_1,S_2} \log \sigma_{S_1,S_2}} \right),
\end{equation}
where $S(\rho_{S_1,S_2}) = - \tr{\rho_{S_1,S_2} \log \rho_{S_1,S_2}}$ is the Von Neumann entropy
of the state $\rho_{S_1,S_2}$. From the sub-additivity of the Von Neumann entropy, we have that
\begin{equation}
\label{superadd_negent}
- S(\rho_{S_1,S_2}) \geq - S(\rho_{S_1}) - S(\rho_{S_2}),
\end{equation}
while from the property~\ref{item:F5b} it follows that
\begin{align}
\inf_{\sigma_{S_1,S_2} \in \f_1^{(2)}}
\left( - \tr{\rho_{S_1,S_2} \log \sigma_{S_1,S_2}} \right)
&= \inf_{\sigma_{S_1}, \sigma_{S_2} \in \f_1}
\left( - \tr{\rho_{S_1,S_2} \log \sigma_{S_1} \otimes \sigma_{S_2}} \right) \nonumber \\
&= \inf_{\sigma_{S_1}, \sigma_{S_2} \in \f_1} 
\left( - \tr{\rho_{S_1} \log \sigma_{S_1}} - \tr{\rho_{S_2} \log \sigma_{S_2}} \right) \nonumber \\
&= \inf_{\sigma_{S_1} \in \f_1} \left( - \tr{\rho_{S_1} \log \sigma_{S_1}} \right)
+ \inf_{\sigma_{S_2} \in \f_1} \left( - \tr{\rho_{S_2} \log \sigma_{S_2}} \right).
\label{part_2_fi_ineq}
\end{align}
From Eqs.~\eqref{part_1_fi_ineq}, \eqref{superadd_negent}, and \eqref{part_2_fi_ineq} it follows that
\begin{align}
E_{\f_1} (\rho_{S_1,S_2}) &\geq
\inf_{\sigma_{S_1} \in \f_1} \left( - S(\rho_{S_1}) - \tr{\rho_{S_1} \log \sigma_{S_1}} \right)
+ \inf_{\sigma_{S_2} \in \f_1} \left( - S(\rho_{S_2}) - \tr{\rho_{S_2} \log \sigma_{S_2}} \right) \nonumber \\
&= E_{\f_1} (\rho_{S_1}) + E_{\f_1} (\rho_{S_2}).
\end{align}
\end{proof}
The following proposition is used in Sec.~\ref{bank_interconvert}, in Prop.~\ref{monotonicity_passive},
and in Cor.~\ref{bank_equal_rel_ent}. The proposition states that, when the curve of bank states is
strictly convex, and we consider $n$ copies of a bank system, the set of bank states $\f^{(n)}_{\text{bank}}$
is given by the tensor product of $n$ copies of states that are in the set $\f_{\text{bank}}$, each of them
with the same value of monotones $E_{\f_1}$ and $E_{\f_2}$.
\begin{prop}
\label{additive_f3}
Suppose the sets $\f_1$ and $\f_2$ satisfy property~\ref{item:F5b}, that is, $\f_i^{(n)} = \f_i^{\otimes n}$
for all $n \in \N$, $i = 1,2$, and the set of bank states $\f_{\text{bank}}$ is represented by a strictly convex
curve in the resource diagram. Consider the set of bank states $\f_{\text{bank}}\left(\bar{E}_{\f_1},\bar{E}_{\f_2}
\right)$ defined in Eq.~\eqref{subset_bank}, where $E_{\f_1}$ and $E_{\f_2}$ are the relative entropy distances
from the sets $\f_1$ and $\f_2$, respectively. Then, when $n \in \N$ copies of the bank system are considered,
we find that the set of bank states coincides with
\begin{equation}
\f^{(n)}_{\text{bank}} = \left\{ \rho_1 \otimes \ldots \otimes \rho_n \in \SHn{n} \ | \
\exists \, \bar{E}_{\f_1} , \bar{E}_{\f_2} \ \text{such that} \
\rho_1, \ldots , \rho_n \in \f_{\text{bank}}\left(\bar{E}_{\f_1},\bar{E}_{\f_2}\right) \right\}.
\end{equation}
Furthermore, we have that for all subset of bank state $\f_{\text{bank}}\left(\bar{E}_{\f_1},\bar{E}_{\f_2}\right)
\subset \SH$, the corresponding bank subset in $\SHn{n}$ is such that
\begin{equation}
\label{additivity_bank}
\f^{(n)}_{\text{bank}}\left(\bar{E}_{\f_1},\bar{E}_{\f_2}\right)
=
\f^{\otimes n}_{\text{bank}}\left(\bar{E}_{\f_1},\bar{E}_{\f_2}\right).
\end{equation}
\end{prop}
\begin{proof}
We prove the theorem for $n = 2$, as the argument extends trivially for $n > 2$.
Consider a state $\sigma_{S_1,S_2} \in \SHn{2}$. From Lem.~\ref{f_i_inequality},
it follows that
\begin{equation}
\label{correlat_ineq}
E_{\f_i}(\sigma_{S_1,S_2}) \geq E_{\f_i}(\sigma_{S_1}) + E_{\f_i}(\sigma_{S_2}) \ , \ i = 1,2, 
\end{equation}
where $\sigma_{S_1} = \Tr{S_2}{\sigma_{S_1,S_2}}$, $\sigma_{S_2} = \Tr{S_1}{\sigma_{S_1,S_2}}$,
and the inequality is saturated iff $\sigma_{S_1,S_2} = \sigma_{S_1} \otimes \sigma_{S_2}$.
Now, for both the states $\sigma_{S_1}, \sigma_{S_2} \in \SH$, select the bank states
$\rho_{P_1}, \rho_{P_2} \in \f_{\text{bank}}$ such that
\begin{equation}
\label{passiv_ineq}
E_{\f_i}(\sigma_{S_j}) \geq E_{\f_i}(\rho_{P_j}) \ , \ i,j = 1,2.
\end{equation}
Recall now that, in the $E_{\f_1}$--$E_{\f_2}$ diagram, the curve of bank state is convex (see Prop.~\ref{convex_curve}),
and therefore given $\rho_{P_1}, \rho_{P_2} \in \f_{\text{bank}}$, we can find another $\rho_{P_3} \in \f_{\text{bank}}$
such that
\begin{equation}
\label{convex_ineq}
\frac{1}{2} E_{\f_i}(\rho_{P_1}) + \frac{1}{2} E_{\f_i}(\rho_{P_2}) \geq E_{\f_i}(\rho_{P_3}) \ , \ i = 1,2,
\end{equation}
where the inequality (when the curve is strictly convex) is saturated iff $\rho_{P_1}$, $\rho_{P_2}$, and
$\rho_{P_3}$ all belong to the same subset $\f_{\text{bank}}\left(\bar{E}_{\f_1},\bar{E}_{\f_2}\right)$. By
combining Eqs.~\eqref{correlat_ineq}, \eqref{passiv_ineq}, and \eqref{convex_ineq}, together with
property~\ref{item:F5b} of the sets $\f_1$ and $\f_2$ (that implies the additivity of the corresponding relative
entropy distances), we find that for all $\sigma_{S_1,S_2} \in \SHn{2}$, it exists a $\rho_{P_3} \in \f_{\text{bank}}$
such that
\begin{equation}
E_{\f_i}(\sigma_{S_1,S_2}) \geq E_{\f_i}(\rho_{P_3}^{\otimes 2}) \ , \ i = 1, 2
\end{equation}
where the inequality is saturated iff $\sigma_{S_1,S_2} = \sigma_{S_1} \otimes \sigma_{S_2}$, and
both $\sigma_{S_1}$ and $\sigma_{S_2}$ belong to the same subset $\f_{\text{bank}}\left(\bar{E}_{\f_1},
\bar{E}_{\f_2}\right)$. Due to the definition of bank states given in Eq.~\eqref{set_f3}, the thesis of this
proposition follows.
\end{proof}
The next proposition shows that, when the invariant sets $\f_1$ and $\f_2$ are convex
sets, the set of bank states $\f_{\text{bank}}\left(\bar{E}_{\f_1},\bar{E}_{\f_2} \right)$, defined in
Eq.~\eqref{subset_bank}, is convex as well. This proposition is used in Sec.~\ref{bank_interconvert},
as well as in Thm.~\ref{bank_equal_rel_ent}.
\begin{prop}
\label{convex_f3}
Suppose that $\f_1$ and $\f_2$ are convex sets, property~\ref{item:F2}, and consider the relative
entropy distances from these two sets, $E_{\f_1}$ and $E_{\f_2}$. Then, the set of bank states
$\f_{\mathrm{bank}}\left(\bar{E}_{\f_1},\bar{E}_{\f_2}\right)$ is convex, as well as its extension to
the $n$-copy case, $\f^{(n)}_{\mathrm{bank}}\left(\bar{E}_{\f_1},\bar{E}_{\f_2}\right)$, defined in
Eq.~\eqref{additivity_bank}.
\end{prop}
\begin{proof}
Let us consider two states $\rho_1$, $\rho_2 \in \f_{\text{bank}}\left(\bar{E}_{\f_1},\bar{E}_{\f_2}\right)$.
For these two states, there exists $\sigma_1, \sigma_2 \in \f_1$ such that
\begin{subequations}
\label{add_init_cond}
\begin{align}
E_{\f_1}(\rho_1) = \re{\rho_1}{\sigma_1} = \bar{E}_{\f_1}, \\
E_{\f_1}(\rho_2) = \re{\rho_2}{\sigma_2} = \bar{E}_{\f_1}.
\end{align}
\end{subequations}
Then, for all $\lambda \in [0,1]$, we have
\begin{align}
E_{\f_1} \big( \lambda \, \rho_1 + (1 - \lambda) \, \rho_2 \big) &=
\inf_{\sigma \in \f_1} \re{\lambda \, \rho_1 + (1 - \lambda) \, \rho_2}{\sigma} \nonumber \\
&\leq \re{\lambda \, \rho_1 + (1 - \lambda) \, \rho_2}{\lambda \, \sigma_1 + (1 - \lambda) \, \sigma_2}
\nonumber \\
&\leq \lambda \, \re{\rho_1}{\sigma_1} + (1 - \lambda) \, \re{\rho_2}{\sigma_2} = \bar{E}_{\f_1},
\end{align}
where the first inequality follows from the fact that $\f_1$ is convex, property~\ref{item:F2}, and the second inequality
from the joint convexity of the relative entropy. In the same way, it follows that
\begin{equation}
E_{\f_2} \big( \lambda \, \rho_1 + (1 - \lambda) \, \rho_2 \big) \leq \bar{E}_{\f_2}.
\end{equation}
Since $\rho_1$, $\rho_2 \in \f_{\text{bank}}\left(\bar{E}_{\f_1},\bar{E}_{\f_2}\right)$, they satisfy the properties
of Eq.~\eqref{set_f3}, and therefore it has to be that, for all $\lambda \in [0,1]$,
\begin{equation}
E_{\f_1} \big( \lambda \, \rho_1 + (1 - \lambda) \, \rho_2 \big) = \bar{E}_{\f_1} \ \text{and} \
E_{\f_2} \big( \lambda \, \rho_1 + (1 - \lambda) \, \rho_2 \big) = \bar{E}_{\f_2}.
\end{equation}
Thus, we have that $\lambda \, \rho_1 + (1 - \lambda) \, \rho_2 \in \f_{\text{bank}}\left(\bar{E}_{\f_1},\bar{E}_{\f_2}\right)$.
This result can be extended to the case of $n \in \N$ copies of the system, where
the bank set $\f^{(n)}_{\mathrm{bank}}\left(\bar{E}_{\f_1},\bar{E}_{\f_2}\right)$ is defined as in
Eq.~\eqref{additivity_bank}. In this case, the proof is analogous to the one considered above,
with the exception that in the rhs of Eqs.~\eqref{add_init_cond}, and of the following ones, we
add the multiplicative factor $n$.
\end{proof}
The next lemma is used in Prop.~\ref{monotonicity_passive}. The lemma states that, given the
class of operations $\A_{\text{multi}}$ for which $\f_1$ and $\f_2$ are invariant sets, the set of bank
states $\f_{\text{bank}}$, defined in Eq.~\eqref{set_f3}, is invariant as well.
\begin{lem}
\label{lem:inv_f3}
Consider a resource theory $\rt_{\text{multi}}$ with allowed operations $\A_{\text{multi}}$, and two
invariant sets $\f_1$ and $\f_2$. Consider the subset of bank states $\f_{\text{bank}}\left(\bar{E}_{\f_1},
\bar{E}_{\f_2}\right)$ as defined in Eq.~\eqref{subset_bank}. Then, for all $\chn \in \A_{\text{multi}}$,
we have that $\f_{\text{bank}}\left(\bar{E}_{\f_1}, \bar{E}_{\f_2}\right)$ is an invariant set, that is
\begin{equation}
\chn \left( \rho \right) \in \f_{\mathrm{bank}}\left(\bar{E}_{\f_1}, \bar{E}_{\f_2}\right)
, \quad \forall \, \rho \in \f_{\mathrm{bank}}\left(\bar{E}_{\f_1}, \bar{E}_{\f_2}\right)
\end{equation}
Analogously, the set of bank states describing $n$ copies of the bank system is invariant under the
class of allowed operations $\A^{(n)}_{\text{multi}}$.
\end{lem}
\begin{proof}
Let us consider $\rho \in \f_{\text{bank}}\left(\bar{E}_{\f_1}, \bar{E}_{\f_2}\right)$, as well as the state
$\chn(\rho)$ obtained by applying the map $\chn \in \A_{\text{multi}}$ to the bank state.
Due to the monotonicity of $E_{\f_1}$ and $E_{\f_2}$, we have that $E_{\f_1}\left(\chn(\rho)\right)
\leq E_{\f_1}\left(\rho\right)$, and $E_{\f_2}\left(\chn(\rho)\right) \leq E_{\f_2}\left(\rho\right)$.
Recall now that $\rho$ is a bank state, which implies that $\forall \, \sigma \in \SH$, one (or more) of
the following options holds
\begin{enumerate}
\item \label{ineq_f1} $E_{\f_1}(\sigma) > E_{\f_1}(\rho)$.
\item \label{ineq_f2} $E_{\f_2}(\sigma) > E_{\f_2}(\rho)$.
\item \label{eq_f1_f2} $E_{\f_1}(\sigma) = E_{\f_1}(\rho)$ and $E_{\f_2}(\sigma) = E_{\f_2}(\rho)$.
\end{enumerate}
However, the monotonicity conditions given by $E_{\f_1}$ and $E_{\f_2}$ implies that $\chn(\rho)$
violates options \ref{ineq_f1} and \ref{ineq_f2}, so that option \ref{eq_f1_f2} is the only possible one.
But this implies that $E_{\f_1}(\chn(\rho)) = E_{\f_1}(\rho)$ and $E_{\f_2}(\chn(\rho)) =
E_{\f_2}(\rho)$, meaning that $\chn(\rho) \in \f_{\text{bank}}\left(\bar{E}_{\f_1}, \bar{E}_{\f_2}\right)$.
The same argument applies to the set $\f_{\text{bank}}^{(n)}$, when $n$ copies of the system are
considered. Indeed, this case is analogous to the one considered above, with the sole difference that now
the state $\rho \in \f_{\text{bank}}^{(n)}$, the state $\sigma \in \SHn{n}$, and the operations we use
are in the class $\A_{\text{multi}}^{(n)}$ defined in Sec.~\ref{multi_resource}.
\end{proof}
The last proposition of this section shows that the relative entropy distance from the set
$\f_{\text{bank}}\left(\bar{E}_{\f_1},\bar{E}_{\f_2}\right)$ is monotonic under the class of operations
$\A_{\text{single}}$, introduced in Def.~\ref{def:sing_res_constr}. This proposition is used in
Cor.~\ref{bank_equal_rel_ent}.
\begin{prop}
\label{monotonicity_passive}
Consider a multi-resource theory $\rt_{\text{multi}}$ with two resources, whose allowed operations
$\A_{\text{multi}}$ leave the sets $\f_1$ and $\f_2$ invariant. Suppose these invariant sets satisfy
the properties~\ref{item:F1}, \ref{item:F2}, \ref{item:F3}, and \ref{item:F5b}. Then, the relative entropy
distance from the subset of bank states $\f_{\text{bank}}\left(\bar{E}_{\f_1},\bar{E}_{\f_2}\right)$ is
monotonic under both the class of operations $\A_{\text{multi}}$ and the class $\A_{\text{single}}$
introduced in Def.~\ref{def:sing_res_constr}.
\end{prop}
\begin{proof}
{\bf 1}. Here we show monotonicity of the relative entropy distance with respect to the addition
of an ancillary system described by $n \in \N$ copies of a bank states. Consider the state $\rho \in
\SH$, and the bank state $\rho_P \in \f_{\text{bank}}\left(\bar{E}_{\f_1},\bar{E}_{\f_2}\right)$.
Then, we have
\begin{align}
E_{\f_{\text{bank}}\left(\bar{E}_{\f_1},\bar{E}_{\f_2}\right)}(\rho \otimes \rho_P^{\otimes n})
&= \inf_{\sigma, \sigma_{P_1}, \ldots , \sigma_{P_n} \in \f_{\text{bank}}\left(\bar{E}_{\f_1},\bar{E}_{\f_2}\right)}
\re{\rho \otimes \rho_P^{\otimes n}}{\sigma \otimes \sigma_{P_1} \otimes \ldots \otimes \sigma_{P_n}} \nonumber \\
&= \inf_{\sigma \in \f_{\text{bank}}\left(\bar{E}_{\f_1},\bar{E}_{\f_2}\right)} \re{\rho}{\sigma}
+ \sum_{i=1}^n \inf_{\sigma_{P_i} \in \f_{\text{bank}}\left(\bar{E}_{\f_1},\bar{E}_{\f_2}\right)}
\re{\rho_P}{\sigma_{P_i}} \nonumber \\
&= \inf_{\sigma \in \f_{\text{bank}}\left(\bar{E}_{\f_1},\bar{E}_{\f_2}\right)}
\re{\rho}{\sigma} = E_{\f_{\text{bank}}\left(\bar{E}_{\f_1},\bar{E}_{\f_2}\right)}(\rho),
\end{align}
where the first equality follows from Prop.~\ref{additive_f3}, and the last one
from the fact that $\rho_P \in \f_{\text{bank}}\left(\bar{E}_{\f_1},\bar{E}_{\f_2}\right)$.
\par
{\bf 2}. Now we show monotonicity of the relative entropy distance with respect to the allowed operations
$\A_{\text{mulit}}$. Let us consider a state $\rho \in \SH$, together with an operation $\chn \in
\A_{\text{multi}}$. Then, we have that
\begin{align}
\label{monotonicity_eq}
E_{\f_{\text{bank}}\left(\bar{E}_{\f_1},\bar{E}_{\f_2}\right)}\big(\chn(\rho)\big) &=
\inf_{\sigma \in \f_{\text{bank}}\left(\bar{E}_{\f_1},\bar{E}_{\f_2}\right)}
\re{\chn(\rho)}{\sigma} \leq
\inf_{\sigma \in \f_{\text{bank}}\left(\bar{E}_{\f_1},\bar{E}_{\f_2}\right)}
\re{\chn(\rho)}{\chn(\sigma)} \nonumber \\
&\leq \inf_{\sigma \in \f_{\text{bank}}\left(\bar{E}_{\f_1},\bar{E}_{\f_2}\right)}
\re{\rho}{\sigma} = E_{\f_{\text{bank}}\left(\bar{E}_{\f_1},\bar{E}_{\f_2}\right)}(\rho),
\end{align}
where the first inequality follows from Lem.~\ref{lem:inv_f3}, and the second one from the
monotonicity of the relative entropy under CPTP maps. This result trivially extends to the case
in which we have multiple copies of the system, since in Lem.~\ref{lem:inv_f3} we have shown
that $\f^{(n)}_{\text{bank}}$ is invariant under the allowed operations $\A_{\text{multi}}^{(n)}$ for
all $n \in \N$.
\par
{\bf 3}. We show the monotonicity of the relative entropy with respect to partial tracing when
the ancillary system is composed by just one copy. However, the result straightforwardly extends
to the case in which the ancillary system is composed by $n \in \N$ copies. Let us consider the state
$\rho_{S_1,S_2} \in \SHn{2}$. Then, we have that
\begin{align}
E_{\f_{\text{bank}}\left(\bar{E}_{\f_1},\bar{E}_{\f_2}\right)}(\Tr{S_2}{\rho_{S_1,S_2}})
&= \inf_{\sigma_{S_1} \in \f_{\text{bank}}\left(\bar{E}_{\f_1},\bar{E}_{\f_2}\right)}
\re{\Tr{S_2}{\rho_{S_1,S_2}}}{\sigma_{S_1}} \nonumber \\
&= \inf_{\sigma_{S_1}, \sigma_{S_2} \in \f_{\text{bank}}\left(\bar{E}_{\f_1},\bar{E}_{\f_2}\right)}
\re{\Tr{S_2}{\rho_{S_1,S_2}}}{\Tr{S_2}{\sigma_{S_1} \otimes \sigma_{S_2}}} \nonumber \\
&\leq \inf_{\sigma_{S_1}, \sigma_{S_2} \in \f_{\text{bank}}\left(\bar{E}_{\f_1},\bar{E}_{\f_2}\right)}
\re{\rho_{S_1,S_2}}{\sigma_{S_1} \otimes \sigma_{S_2}} \nonumber \\
&=
E_{\f_{\text{bank}}\left(\bar{E}_{\f_1},\bar{E}_{\f_2}\right)}(\rho_{S_1,S_2}),
\end{align}
where the second equality follows from Prop.~\ref{additive_f3}, while the inequality follows from
the monotonicity of the relative entropy distance under CPTP maps.
\end{proof}
\bibliographystyle{unsrtnat}
\bibliography{./biblio_firstlaw}
\end{document}

%% file: quantumdefinitions.tex
\usepackage{amsthm}
\usepackage{latexsym}
\usepackage{amsfonts}
\usepackage{mathtools}
\usepackage{bbm}
\usepackage{dsfont}
\usepackage{graphicx,xcolor}
\usepackage{braket}
\usepackage{enumerate}
\usepackage{thmtools}
\usepackage{thm-restate}
\usepackage[sort&compress,numbers]{natbib}
\usepackage{doi}
\usepackage{hyperref}
\makeatletter
\def\l@subsubsection#1#2{}
\makeatother

\newcommand{\R}{\mathbb{R}}
\newcommand{\C}{\mathbb{C}}
\newcommand{\N}{\mathbb{N}}

\newcommand{\Q}{\mathbb{Q}}

\newcommand{\chn}{\mathcal{E}}
\newcommand{\BH}{\mathcal{B} \left( \mathcal{H} \right)}
\newcommand{\BHn}{\mathcal{B} \left( \mathcal{H}^{\otimes n} \right)}
\newcommand{\SH}{\mathcal{S} \left( \mathcal{H} \right)}
\newcommand{\SHs}{\mathcal{S} \left( \mathcal{H}_S \right)}
\newcommand{\SHi}{\mathcal{S} \left( \mathcal{H}_1 \right)}
\newcommand{\SHii}{\mathcal{S} \left( \mathcal{H}_2 \right)}

\newcommand{\SHbi}{\mathcal{S} \left( \mathcal{H}_{B_1} \right)}
\newcommand{\SHbii}{\mathcal{S} \left( \mathcal{H}_{B_2} \right)}

\newcommand{\SHbI}{\mathcal{S} \left( \mathcal{H}_{B_i} \right)}

\newcommand{\SHn}[1]{\mathcal{S} \left( \mathcal{H}^{\otimes {#1}} \right)}
\newcommand{\hil}{\mathcal{H}}
\newcommand{\Id}{\mathbb{I}}
\newcommand{\tr}[1]{\mathrm{Tr}\left[ {#1} \right]} 
\newcommand{\Tr}[2]{\mathrm{Tr}_{#1}\left[ {#2} \right]} 
\newcommand{\p}{\mathrm{p}}
\newcommand{\q}{\mathrm{q}}
\newcommand{\h}{\mathrm{h}}
\newcommand{\rt}{\text{R}}
\newcommand{\f}{\mathcal{F}}
\newcommand{\s}{\mathcal{S}}
\newcommand{\A}{\mathcal{C}}

\newcommand{\supp}[1]{\text{supp}\left( {#1} \right)}
\newcommand{\re}[2]{D( {#1} \, \| \, {#2} )}
\newtheorem{thm}{Theorem}
\newtheorem{prop}[thm]{Proposition}
\newtheorem{lem}[thm]{Lemma}
\newtheorem{coro}[thm]{Corollary}
\newtheorem{definition}[thm]{Definition}
\usepackage{enumitem}
\usepackage{xpatch}
\newlist{describe}{description}{1}
\setlist[describe,1]{%
  font=\normalfont\textsf,
  itemindent=0pt,
  wide,
  itemsep=0pt,topsep=2pt,
  format={\normalfont\textsfcolor}
}
\newcommand*{\textsfcolor}[1]{\textsf{#1:}}
\makeatletter
\xpatchcmd{\enit@description@i}{%
  \labelsep\z@
}{%
  \phantomsection
  \let\org@label\label
  \let\label\@gobble
  \protected@edef\@currentlabel{##1}%
  \let\label\org@label
  \labelsep\z@
}{}{\undefined}
\makeatother
\usepackage{tikz}
\usetikzlibrary{arrows,shapes,positioning}
\tikzstyle{box} = [rectangle, rounded corners, minimum width=3cm, minimum height=1cm,text centered, text width=4cm, draw=black]
\tikzstyle{finbox} = [rectangle, rounded corners, draw, double, double distance =1pt, minimum width=3cm, minimum height=1cm,text centered, text width=4cm, draw=black]
\tikzstyle{longbox} = [rectangle, rounded corners, draw, double, double distance =1pt, minimum width=3cm, minimum height=1cm,text centered, text width=8cm, draw=black]
\tikzstyle{initbox} = [rectangle, draw, double, double distance =1pt, minimum width=3cm, minimum height=1cm,text centered, text width=12cm, draw=black]
\tikzstyle{arrow} = [thick,->,>=angle 60]